\documentclass[final,twoside,11pt]{entics} 
\usepackage{enticsmacro}
\usepackage{graphicx}
\usepackage[all]{xy}
\usepackage{amsmath}
\usepackage{tikz}
\usetikzlibrary{cd}
\usepackage{tikzit}

\tikzstyle{blue dot}=[fill={rgb,255: red,128; green,128; blue,255}, draw={rgb,255: red,128; green,128; blue,255}, shape=circle, minimum size=0.15 cm, inner sep=0 pt]
\tikzstyle{red dot}=[fill={rgb,255: red,238; green,150; blue,121}, draw={rgb,255: red,238; green,150; blue,121}, shape=circle, minimum size=0.15 cm, inner sep=0 pt]
\tikzstyle{yellow dot}=[fill={rgb,255: red,252; green,199; blue,18}, draw={rgb,255: red,252; green,199; blue,18}, shape=circle, minimum size=0.15 cm, inner sep=0 pt]
\tikzstyle{black dot}=[fill=black, draw=black, shape=circle, minimum size=0.15 cm, inner sep=0 pt]

\tikzstyle{blue edge}=[-, draw={rgb,255: red,128; green,128; blue,255}, thick, fill=none]
\tikzstyle{red edge}=[-, draw={rgb,255: red,238; green,150; blue,121}, thick, fill=none]
\tikzstyle{thick white}=[-, draw=white, line width=3 pt]
\tikzstyle{blue fill}=[-, fill={rgb,255: red,128; green,128; blue,255}, draw=none, thick]
\tikzstyle{red fill}=[-, fill={rgb,255: red,238; green,150; blue,121}, draw=none, thick]
\tikzstyle{green edge}=[-, draw={rgb,255: red,123; green,194; blue,30}, thick]
\tikzstyle{black edge}=[-, fill=none, draw=black, thick]
\tikzstyle{white fill}=[-, fill={rgb,255: red,247; green,247; blue,247}, draw=none]
\tikzstyle{big dashes}=[-, dashed, fill={rgb,255: red,221; green,221; blue,221}]
\tikzstyle{yellow edge}=[-, draw={rgb,255: red,252; green,199; blue,18}, thick]
\tikzstyle{yellow fill}=[-, thick, draw=none, fill={rgb,255: red,252; green,199; blue,18}]

\usepackage{adjustbox}

\newcommand{\stick}{%
	\mkern-2.5mu
	\mathchoice{}{}{\mkern0.2mu}{\mkern0.5mu}%
}
\newcommand{\comp}[3]{{#1} \circ_{#2} \stick {#3}}
\newcommand{\wcomp}[3]{{#1} \bullet_{#2} \stick {#3}}

\newcommand{\C}{\mathcal{C}}
\newcommand{\Set}{\mathsf{Set}}
\newcommand{\nto}{\Rightarrow}

\newcommand{\inr}{\mathsf{inr}}
\newcommand{\inl}{\mathsf{inl}}

\renewcommand{\phi}{\varphi}
\renewcommand{\epsilon}{\varepsilon}
\newcommand{\supp}{\mathsf{supp}}

\def\conf{MFPS 2022} 	
\volume{1}			
\def\volu{1}			
\def\papno{3}			
\def\mydoi{{\normalshape \href{https://doi.org/10.46298/entics.10482}{10.46298/entics.10482}}}
\def\copynames{A. Goy} 

\def\runtitle{Weakening and Iterating Laws Using String Diagrams}

\def\lastname{Goy}
\begin{document}
\begin{frontmatter}
  \title{Weakening and Iterating Laws Using String Diagrams} 
  \author{Alexandre Goy\thanksref{ALL}\thanksref{myemail}}	
  \address{LIP\\ ENS Lyon\\			
    France}  							
  \thanks[ALL]{The author thanks Ana Sokolova and Daniela Petri{\c s}an for numerous discussions and support, Ralph Sarkis for bringing up the topic of string diagrams, and the anonymous reviewers for their helpful comments and suggestions. This work has been supported by the French project ANR-20-CE48-0005 QuaReMe.}
   \thanks[myemail]{Email: \href{mailto:alexandre.goy@pm.me} {\texttt{\normalshape
        alexandre.goy@pm.me}}} 
\begin{abstract} 
    Distributive laws are a standard way of combining two monads, providing a compositional approach for reasoning about computational effects in semantics. Situations where no such law exists can sometimes be handled by weakening the notion of distributive law, still recovering a composite monad. A celebrated result from Eugenia Cheng shows that combining $n$ monads is possible by iterating more distributive laws, provided they satisfy a coherence condition called the Yang-Baxter equation. Moreover, the order of composition does not matter, leading to a form of associativity.

    The main contribution of this paper is to generalise the associativity of iterated composition to weak distributive laws in the case of $n = 3$ monads. To this end, we use string-diagrammatic notation, which significantly helps make increasingly complex proofs more readable. We also provide examples of new weak distributive laws arising from iteration.
\end{abstract}
\end{frontmatter}

\section{Introduction} \label{sec:intro}

Monads have been very successful in computer science to model computational effects~\cite{Moggi91,Plotkin02,Wadler95}. In this context, being able to compose monads is of central importance. The primary framework for monad composition is undoubtedly Beck's celebrated theory of distributive laws~\cite{Beck69}. In practice, finding distributive laws between monads is hard, although some recipes are known to build them~\cite{Manes07,Manes08}. In contrast, once a law is found, it often contains a significant semantic content -- thus, it is common to find papers entirely dedicated to describing the construction and the consequences of a distributive law, see \textit{e.g.}~\cite{Jacobs21}. More generally, monad composition is still a very active field in the computer science community nowadays, with possible applications in coalgebra~\cite{Silva10,Jacobs15,Silva21,Bonchi21b} or in the semantics of computation~\cite{Parlant18,Parlant20}.

With the help of a distributive law, one can compose two monads. Composing more monads can be performed by \emph{iterating}~\cite{Cheng11a}, \textit{i.e.}, by using more distributive laws -- yet, these must meet coherence axioms called the~\emph{Yang-Baxter} equations. While distributive laws are the first step towards monad compositionality, \emph{iterated} distributive laws are all the following steps.

Although category theory is a language of choice to work on compositionality properties, there is one major obstacle in the field of distributive laws: monads may not always compose. And consequently, distributive laws may not always exist. The literature exhibits many such negative results, called \emph{no-go theorems}~\cite{Varacca06,Klin18,Dahlqvist18,Zwart20,Salamanca20,Zwart22}. A possible fix consists in \emph{weakening} the notion of distributive law to recover transformations that, even if not satisfying all the usual axioms, still enable some weak kind of monad composition. This simple idea has received substantial coverage in the computer science community since the paper of Garner~\cite{Garner20}. Simplifying earlier work from~\cite{Street09,Bohm10,Bohm11}, Garner proposed to delete the one axiom that was obstructing, in many concrete cases, the existence of a law. Several papers followed on this revived notion of \emph{weak distributive law}~\cite{Goy20,Goy21,Bonchi21,Sarkis21}. In general, finding examples of weak distributive laws remains difficult, and very few techniques are known to produce them. Moreover, for such a law to be relevant, it needs to fit into the sweet spot where no distributive law can exist, but there is still enough structure for a weak one. This situation arises in a handful of known cases, \textit{e.g.}, when distributing the distribution monad over the powerset monad~\cite{Goy20} or the powerset monad over itself~\cite{Garner20,Goy21}. 

Iterating weak distributive laws is a promising target: one would enable building new weak distributive laws out of some existing ones at low cost. However, two factors are now interacting with the basic theory of Beck.  The proof of the iteration theorem~\cite[Appendix~A]{Cheng11a}, due to Eugenia Cheng, is not very difficult per se, but, already for the simplest case of three monads, it contains large commutative diagrams that are hard to comprehend in one glimpse. On the other hand, although weak distributive laws are conceptually close to distributive laws, they make the framework grow in complexity. Manipulating weak laws involves manipulating more natural transformations arising from splitting of idempotents, making proofs less transparent. In short, \emph{iterating} adds one layer of complexity, \emph{weakening} adds another layer, and computations become untractable using the traditional tools of category theory. For a similar situation, consider the paper of Winter~\cite{Winter16} whose main result is to prove a variation of Cheng's theorem in the case of three monads. In a nutshell, the variation consists in modifying the type of \emph{one} of the laws from $TF \nto FT$ into $TF \nto FST$ -- with, of course, consequences for the coherence conditions. Despite this seemingly mild modification, one of the required commutative diagrams becomes so large that finding the appropriate tiling required Winter to use automation via a Prolog program, mainly because of the explosion in the number of naturality squares. As he insists on, the main difficulty is about finding the proof -- verifying it remains straightforward.

In the present paper, we aim to illustrate yet another case where \emph{string diagrams} bring conceptual clarity and simplicity. Under the name string diagrams, there are many graphical calculi based on topological operations such as string bending, dragging, and sliding. Such frameworks have raised more and more interest over the past years: we can mention, \textit{e.g.}, the book of Coecke and Kissinger~\cite{Coecke17} and the line of work of Bonchi, Zanasi \textit{et al.}~\cite{Zanasi15,Zanasi15b,Zanasi17,Zanasi18,Zanasi21,Zanasi22}. String diagrams are nowadays recognised as a helpful tool. Inside published papers, commutative diagrams have long been preferred, even when they are less enlighting (see \textit{e.g.}~\cite[Propositions~2.3,~3.3,~and~4.1]{Street09}). One of the aims of this paper is to emphasise the following point: when it comes to concepts that rely on several complexity layers, such as \emph{iterating} plus \emph{weakening}, string diagrams may even be \emph{required} as a computation-assisting technology. The string diagrams we will make use of are those introduced in the style of Hinze and Marsden~\cite{Marsden14,Hinze16,Hinze16b}. These were already used to produce new results within the theory of distributive laws by Zwart~\cite{Zwart20} and in the author's PhD thesis~\cite{GoyTh}. There seems to be also some work in progress concerning no-go theorems~\cite{Shah22}, showing that string diagrams are imposing themselves in the area to visualise categorical transformations, provide elegant proofs, and communicate ideas.

\paragraph*{Contributions.}
This paper contains contributions of two kinds: technical (laws) and practical (strings). On the practical side, the paper introduces in Section~\ref{sec:laws} a convenient string-diagrammatic notation that helps manipulate concepts at play. In particular, we introduce a new notation for idempotents of the form $\kappa \colon ST \nto ST$ in equation~\eqref{eq:kappa_notations}, which is central in the whole development. The main technical contribution of the paper consists in generalising Cheng's Theorem~\ref{theo:cheng} to weak distributive laws in the case of $n = 3$ monads. For clarity, this result is split into three distinct statements: Theorem~\ref{theo:first_comp}, Theorem~\ref{theo:second_comp}, and Theorem~\ref{theo:comp_equal}. Along the way, and building on a key technical tool which is Lemma~\ref{lem:four_eq}, we derive several results of the theory of weak distributive laws, some of them already presented in~\cite{GoyTh}, some of them unpublished. In Section~\ref{subsec:examples}, we also provide four families of examples where our main result applies: trivial weak distributive laws, and generic distributive laws involving the exception monad, the reader monad, and the writer monad. Finally, in Section~\ref{subsec:outputs}, we give a new perspective on algebras for a monad by identifying them as distributive laws (Proposition~\ref{prop:output_1}). In this vision, algebras for a distributive law are identified as iterated distributive laws (Proposition~\ref{prop:output_2}).

\paragraph*{Related Work.} The content of the present paper is partly based on the author's PhD thesis~\cite{GoyTh}, with a slightly different perspective. In~\cite{GoyTh} the focus is put on the correspondence weak distributive laws - weak liftings - weak extensions. In contrast, in the present paper, we instead follow the approach of~\cite[Section~2]{Street09} consisting in building the weak composite monad directly out of the weak distributive law. It has to be mentioned that iterated distributive laws have already been studied from a string-diagrammatic perspective in a paper from Hinze and Marsden~\cite{Hinze16b}. 

\section{Background} \label{sec:preli}

\subsection{Preliminaries and Notation}

Definitional equality is denoted by $\triangleq$. We assume basic knowledge of category theory, including: category, functor, natural transformation. Let $\C$ be a category and $S$, $T$ be endofunctors on $\C$. We use the Lie bracket notation $[S,T]$ as a shortcut for the type of natural transformations $TS \nto ST$. For example, $\lambda \colon [S,T]$ means that $\lambda$ is a natural transformation of type $TS \nto ST$. This notation is convenient because $(a)$ laws $\lambda : [S,T]$ are related to composite monads $\comp{S}{}{T}$ and $(b)$ the type $TS \nto ST$ looks like a logical version of the formal expression $-TS + ST$. Identity morphisms will be denoted by $1_X \colon X \to X$, the identity functor by $1 \colon \C \to \C$, and the identity natural transformation on a functor $F$ by $1_F \colon F \nto F$.

In the whole paper we assume that $\C$ is idempotent complete, that is, for every morphism $k \colon X \to X$ in $\C$ such that $k \circ k = k$, there are morphisms $p \colon X \to Y$ and $i \colon Y \to X$ such that $i \circ p = k$ and $p \circ i = 1_Y$. A direct consequence is that for any functor $F \colon \C \to \C$, every idempotent natural transformation $\kappa \colon F \nto F$ splits. Explicitly, there are a functor $K$ and natural transformations $\pi \colon F \nto K$, $\iota \colon K \nto F$ such that $\iota \circ \pi = \kappa$ and $\pi \circ \iota = 1_K$. All the forecoming examples live in the category $\Set$ of sets and functions, which is idempotent complete. So are, \textit{e.g.}, toposes and the category of compact Hausdorff spaces~\cite{Goy21}. We globally impose idempotent completeness to be able to split natural transformations denoted by $\kappa$ in the sequel -- all our results remain valid by only asserting locally that these $\kappa$ split.

\subsection{Introducing String Diagrams}

We strongly encourage the reader to refer to~\cite[Sections~2 and~3]{Hinze16} for a thorough introduction to string diagrams. In this short section, we will nonetheless provide a crash course in string diagrams, mainly to fix notations.

The following example sums up all useful notation. On the left part, we provide several basic natural transformations $\lambda \colon TS \nto ST$, $\nu \colon ST \nto U$, $\epsilon \colon U \nto 1$, $\eta \colon 1 \nto T$. On the right part, we provide two equivalent string diagrams representing the same composite natural transformation $\nu \circ \lambda\epsilon \circ \eta SU = \nu \circ \lambda \circ \eta S \circ S\epsilon \colon SU \nto U$.
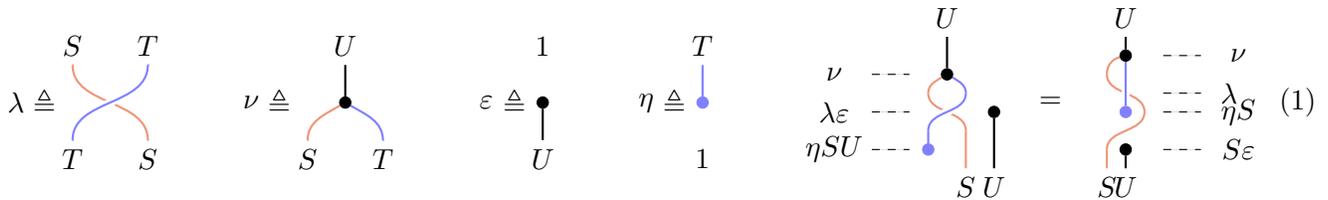
\begin{equation} \label{eq:example_diagram}
\begin{tikzpicture}
	\begin{pgfonlayer}{nodelayer}
		\node [style=none] (0) at (-15.75, 1) {};
		\node [style=none] (1) at (-13.75, 1) {};
		\node [style=none] (2) at (-15.75, -1) {};
		\node [style=none] (3) at (-13.75, -1) {};
		\node [style=none] (4) at (-13.75, 1) {};
		\node [style=none] (5) at (-15.75, -1) {};
		\node [style=none] (6) at (-17.25, 0) {$\lambda$};
		\node [style=none] (11) at (-8.5, 1) {};
		\node [style=black dot] (12) at (-8.5, 0) {};
		\node [style=none] (13) at (-9.5, -1) {};
		\node [style=none] (14) at (-7.5, -1) {};
		\node [style=black dot] (15) at (-3.25, 0) {};
		\node [style=none] (16) at (-3.25, -1) {};
		\node [style=blue dot] (17) at (1, 0) {};
		\node [style=none] (18) at (1, 1) {};
		\node [style=none] (19) at (-11, 0) {$\nu$};
		\node [style=none] (20) at (-4.75, 0) {$\epsilon$};
		\node [style=none] (21) at (-0.5, 0) {$\eta$};
		\node [style=none] (22) at (-15.75, -1.5) {$T$};
		\node [style=none] (23) at (-13.75, -1.5) {$S$};
		\node [style=none] (24) at (-15.75, 1.5) {$S$};
		\node [style=none] (25) at (-13.75, 1.5) {$T$};
		\node [style=none] (26) at (-7.5, -1.5) {$T$};
		\node [style=none] (27) at (-9.5, -1.5) {$S$};
		\node [style=none] (28) at (-8.5, 1.5) {$U$};
		\node [style=none] (29) at (-3.25, -1.5) {$U$};
		\node [style=none] (30) at (1, 1.5) {$T$};
		\node [style=none] (31) at (-3.25, 1.5) {$1$};
		\node [style=none] (32) at (1, -1.5) {$1$};
		\node [style=none] (33) at (7, 0.25) {};
		\node [style=none] (34) at (8, 0.25) {};
		\node [style=none] (35) at (7, -0.75) {};
		\node [style=none] (36) at (8, -0.75) {};
		\node [style=none] (37) at (8, 0.25) {};
		\node [style=none] (38) at (7, -0.75) {};
		\node [style=none] (39) at (7.5, 1.75) {};
		\node [style=black dot] (40) at (7.5, 0.75) {};
		\node [style=none] (41) at (7, 0.25) {};
		\node [style=none] (42) at (8, 0.25) {};
		\node [style=blue dot] (43) at (7, -1.25) {};
		\node [style=none] (44) at (7, -0.75) {};
		\node [style=none] (45) at (8, -1.75) {};
		\node [style=black dot] (46) at (8.75, -0.25) {};
		\node [style=none] (47) at (8.75, -1.75) {};
		\node [style=none] (48) at (4.5, 0.75) {$\nu$};
		\node [style=none] (49) at (4.5, -0.25) {$\lambda\epsilon$};
		\node [style=none] (50) at (4.5, -1.25) {$\eta SU$};
		\node [style=none] (52) at (10.25, 0) {$=$};
		\node [style=none] (53) at (12.25, 1.75) {};
		\node [style=black dot] (54) at (12.25, 1.25) {};
		\node [style=none] (55) at (11.75, 0.75) {};
		\node [style=none] (57) at (12.75, -0.25) {};
		\node [style=none] (58) at (12.75, -0.25) {};
		\node [style=blue dot] (59) at (12.25, -0.25) {};
		\node [style=black dot] (60) at (12.25, -1.25) {};
		\node [style=none] (61) at (12.25, -1.75) {};
		\node [style=none] (62) at (11.75, -1.25) {};
		\node [style=none] (63) at (11.75, -1.75) {};
		\node [style=none] (64) at (15.25, 1.25) {$\nu$};
		\node [style=none] (65) at (15, 0.25) {$\lambda$};
		\node [style=none] (66) at (15.25, -0.25) {$\eta S$};
		\node [style=none] (67) at (15.25, -1.25) {$S\epsilon$};
		\node [style=none] (68) at (-9.5, -1) {};
		\node [style=none] (69) at (-7.5, -1) {};
		\node [style=none] (70) at (-13.75, 1) {};
		\node [style=none] (71) at (-15.75, -1) {};
		\node [style=none] (72) at (-13.75, -1) {};
		\node [style=none] (73) at (-15.75, 1) {};
		\node [style=none] (74) at (7.5, 2.25) {$U$};
		\node [style=none] (75) at (12.25, 2.25) {$U$};
		\node [style=none] (76) at (8.75, -2.25) {$U$};
		\node [style=none] (77) at (12.25, -2.25) {$U$};
		\node [style=none] (78) at (8, -2.25) {$S$};
		\node [style=none] (79) at (11.75, -2.25) {$S$};
		\node [style=none] (80) at (5.5, -0.25) {};
		\node [style=none] (81) at (6.5, -0.25) {};
		\node [style=none] (84) at (5.5, 0.75) {};
		\node [style=none] (85) at (6.5, 0.75) {};
		\node [style=none] (88) at (5.5, -1.25) {};
		\node [style=none] (89) at (6.5, -1.25) {};
		\node [style=none] (92) at (13.25, 1.25) {};
		\node [style=none] (93) at (14.25, 1.25) {};
		\node [style=none] (98) at (13.25, -0.25) {};
		\node [style=none] (99) at (14.25, -0.25) {};
		\node [style=none] (102) at (13.25, 0.25) {};
		\node [style=none] (103) at (14.25, 0.25) {};
		\node [style=none] (106) at (13.25, -1.25) {};
		\node [style=none] (107) at (14.25, -1.25) {};
		\node [style=none] (109) at (-10.25, 0) {$\triangleq$};
		\node [style=none] (110) at (-16.5, 0) {$\triangleq$};
		\node [style=none] (111) at (-4, 0) {$\triangleq$};
		\node [style=none] (112) at (0.25, 0) {$\triangleq$};
	\end{pgfonlayer}
	\begin{pgfonlayer}{edgelayer}
		\draw [style=red edge, in=90, out=-90] (0.center) to (3.center);
		\draw [style=thick white, in=-90, out=90] (2.center) to (1.center);
		\draw [style=blue edge, in=-90, out=90] (5.center) to (4.center);
		\draw [style=black edge] (11.center) to (12);
		\draw [style=red edge, in=90, out=-150] (12) to (13.center);
		\draw [style=blue edge, in=90, out=-30] (12) to (14.center);
		\draw [style=black edge] (15) to (16.center);
		\draw [style=blue edge] (18.center) to (17);
		\draw [style=red edge, in=90, out=-90] (33.center) to (36.center);
		\draw [style=thick white, in=-90, out=90] (35.center) to (34.center);
		\draw [style=blue edge, in=-90, out=90] (38.center) to (37.center);
		\draw [style=black edge] (39.center) to (40);
		\draw [style=red edge, in=90, out=-150] (40) to (41.center);
		\draw [style=blue edge, in=90, out=-30] (40) to (42.center);
		\draw [style=blue edge] (44.center) to (43);
		\draw [style=red edge] (36.center) to (45.center);
		\draw [style=black edge] (46) to (47.center);
		\draw [style=black edge] (53.center) to (54);
		\draw [style=red edge, in=90, out=-165] (54) to (55.center);
		\draw [style=red edge, in=90, out=-90] (55.center) to (57.center);
		\draw [style=red edge] (57.center) to (58.center);
		\draw [style=thick white] (54) to (59);
		\draw [style=blue edge] (54) to (59);
		\draw [style=black edge] (60) to (61.center);
		\draw [style=red edge, in=90, out=-90] (58.center) to (62.center);
		\draw [style=red edge] (62.center) to (63.center);
		\draw [style=red edge] (13.center) to (68.center);
		\draw [style=blue edge] (14.center) to (69.center);
		\draw [style=blue edge] (70.center) to (4.center);
		\draw [style=blue edge] (5.center) to (71.center);
		\draw [style=red edge] (3.center) to (72.center);
		\draw [style=red edge] (73.center) to (0.center);
		\draw [style=big dashes] (80.center) to (81.center);
		\draw [style=big dashes] (84.center) to (85.center);
		\draw [style=big dashes] (88.center) to (89.center);
		\draw [style=big dashes] (92.center) to (93.center);
		\draw [style=big dashes] (98.center) to (99.center);
		\draw [style=big dashes] (102.center) to (103.center);
		\draw [style=big dashes] (106.center) to (107.center);
	\end{pgfonlayer}
\end{tikzpicture}
\end{equation}

We use the following convention: diagrams are read from right to left and from bottom to top (note that~\cite{Hinze16} rather uses top to bottom). A functor $F \colon \C \to \mathcal{D}$, represented by a vertical string, delineates two regions: the region on the right represents the category $\C$, while the region on the left represents the category $\mathcal{D}$. In what follows, we only study endofunctors in $\C$\footnote{To be precise, we are using the standard formalism of string diagrams for monoidal categories~\cite{Joyal91} within the category of $\C$-endofunctors, where monoidal product is given by functor composition.}. Therefore all delineated regions remain blank and represent the category $\C$. The identity functor $1 \colon \C \to \C$ is not depicted. A natural transformation $\alpha \colon F \nto G$ is depicted by a symbol between functor $F$ (below) and functor $G$ (above). Using colours, we will define on-the-fly unambiguous symbols (typically, nodes) to distinguish between different natural transformations -- see equation~\eqref{eq:example_diagram} for a first example. The identity natural transformation $1_F \colon F \nto F$ is a string with no symbol, \textit{i.e.}, it coincides with the representation of the functor $F \colon \C \to \C$.

Let $F$, $G$, $H$, $K$ be functors, and $\alpha \colon F \nto G$, $\beta \colon G \nto H$, $\gamma \colon H \nto K$ be natural transformations. Composition of functors is denoted by $GF$ and depicted by horizontal juxtaposition of strings. Vertical composition of natural transformations is denoted by $\beta \circ \alpha \colon F \nto H$, defined by $(\beta \circ \alpha)_X \triangleq \beta_X \circ \alpha_X$ for all objects $X$, and depicted by vertical glueing of string diagrams. Left (respectively right) composition of a functor and a natural transformation is denoted by $H\alpha \colon HF \nto HG$ (respectively $\alpha H \colon FH \nto GH$), defined by $(H\alpha)_X \triangleq H(\alpha_X)$ (respectively $(\alpha H)_X \triangleq \alpha_{HX}$) for all objects $X$, and depicted by horizontal juxtaposition of string diagrams. Horizontal composition of natural transformations is denoted by $\gamma \alpha \colon HF \nto KG$, defined by $\gamma \alpha \triangleq K\alpha \circ \gamma F = \gamma G \circ H \alpha$, and depicted by horizontal juxtaposition of string diagrams.

Each string diagram represents a natural transformation of some type. Two string diagrams of the same type are identified up to continuous operations such as dragging nodes and bending strings (provided it does not reverse the implicit bottom-up vertical direction). Discontinuous operations include sliding a node past another on the same string, crossing two strings, and deleting or adding strings or nodes. Such operations may be allowed punctually in the presence of a natural transformation that justifies them. For instance, a natural transformation $\lambda \colon [S,T]$, having type $TS \nto ST$, represents an \emph{ad hoc} rule to cross the two strings representing $S$ and $T$. String-diagrammatic notation is coherent with the usual laws of category theory in the sense that two identified string diagrams always denote the same natural transformation.

It is customary, as done in~\cite{Hinze16}, to specify the types of the functors below and above each string diagram. In this paper, only a few functors are in play, so that our colour code would be sufficient to infer all types. For the sake of clarity, we opt for some redundancy by always indicating types.

\section{A String-Diagrammatic Theory of Weak Distributive Laws} \label{sec:laws}

The present section introduces weak distributive laws in string-diagrammatic style, starting back from the standard theory of monads and distributive laws.

\subsection{Monads and Distributive Laws}

\begin{definition} A \emph{monad} is a triple $(T,\eta^T,\mu^T)$ comprising a $\C$-endofunctor $T$ and two natural transformations, the \emph{unit} $\eta^T \colon 1 \nto T$ and the \emph{multiplication} $\mu^T \colon TT \nto T$, such that $\mu^T \circ \eta^T T = 1_T = \mu^T \circ T \eta^T $ (unitality) and $\mu^T \circ \mu^T T = \mu^T \circ T \mu^T$ (associativity).
\end{definition}

String diagrammatically, these data are denoted by
\begin{equation}
    \begin{tikzpicture}
	\begin{pgfonlayer}{nodelayer}
		\node [style=none] (0) at (2, 1) {};
		\node [style=blue dot] (1) at (2, 0) {};
		\node [style=none] (2) at (1, -1) {};
		\node [style=none] (3) at (3, -1) {};
		\node [style=none] (4) at (-3, 1) {};
		\node [style=blue dot] (5) at (-3, 0) {};
		\node [style=none] (7) at (-0.5, 0) {$\mu^T$};
		\node [style=none] (8) at (7, 0) {};
		\node [style=blue dot] (9) at (7, -0.5) {};
		\node [style=none] (10) at (7.5, 1) {};
		\node [style=blue dot] (11) at (7.5, 0.5) {};
		\node [style=none] (12) at (7, 0) {};
		\node [style=none] (13) at (8, 0) {};
		\node [style=none] (14) at (8, -1) {};
		\node [style=none] (15) at (10, -2.5) {unitality};
		\node [style=none] (16) at (12.5, 1) {};
		\node [style=blue dot] (17) at (12.5, 0.5) {};
		\node [style=none] (18) at (12, 0) {};
		\node [style=none] (19) at (13, 0) {};
		\node [style=none] (20) at (9, 0) {$=$};
		\node [style=none] (21) at (12, -1) {};
		\node [style=none] (22) at (13, 0) {};
		\node [style=blue dot] (23) at (13, -0.5) {};
		\node [style=none] (24) at (11, 0) {$=$};
		\node [style=none] (25) at (10, 1) {};
		\node [style=none] (26) at (10, -1) {};
		\node [style=none] (27) at (-4.5, 0) {$\eta^T$};
		\node [style=none] (28) at (17.5, 1) {};
		\node [style=blue dot] (29) at (17.5, 0.5) {};
		\node [style=none] (30) at (17, 0) {};
		\node [style=none] (31) at (18, 0) {};
		\node [style=none] (34) at (17, 0) {};
		\node [style=blue dot] (35) at (17, -0.5) {};
		\node [style=none] (36) at (16.5, -1) {};
		\node [style=none] (37) at (17.5, -1) {};
		\node [style=none] (38) at (18, -1) {};
		\node [style=none] (39) at (20.5, 1) {};
		\node [style=blue dot] (40) at (20.5, 0.5) {};
		\node [style=none] (41) at (20, 0) {};
		\node [style=none] (42) at (21, 0) {};
		\node [style=none] (43) at (21, 0) {};
		\node [style=blue dot] (44) at (21, -0.5) {};
		\node [style=none] (45) at (20.5, -1) {};
		\node [style=none] (46) at (21.5, -1) {};
		\node [style=none] (47) at (20, -1) {};
		\node [style=none] (48) at (19, 0) {$=$};
		\node [style=none] (49) at (19, -2.5) {associativity};
		\node [style=none] (50) at (-7.25, 1) {};
		\node [style=none] (51) at (-7.25, -1) {};
		\node [style=none] (52) at (-9, 0) {$1_T$};
		\node [style=none] (53) at (-7.25, 1.5) {$T$};
		\node [style=none] (54) at (-3, 1.5) {$T$};
		\node [style=none] (55) at (2, 1.5) {$T$};
		\node [style=none] (56) at (7.5, 1.5) {$T$};
		\node [style=none] (57) at (10, 1.5) {$T$};
		\node [style=none] (58) at (12.5, 1.5) {$T$};
		\node [style=none] (59) at (17.5, 1.5) {$T$};
		\node [style=none] (60) at (20.5, 1.5) {$T$};
		\node [style=none] (61) at (-7.25, -1.5) {$T$};
		\node [style=none] (62) at (1, -1.5) {$T$};
		\node [style=none] (63) at (3, -1.5) {$T$};
		\node [style=none] (64) at (8, -1.5) {$T$};
		\node [style=none] (65) at (10, -1.5) {$T$};
		\node [style=none] (66) at (12, -1.5) {$T$};
		\node [style=none] (67) at (16.5, -1.5) {$T$};
		\node [style=none] (68) at (17.5, -1.5) {$T$};
		\node [style=none] (69) at (18, -1.5) {$T$};
		\node [style=none] (70) at (20, -1.5) {$T$};
		\node [style=none] (71) at (20.5, -1.5) {$T$};
		\node [style=none] (72) at (21.5, -1.5) {$T$};
		\node [style=none] (73) at (-8, 0) {$\triangleq$};
		\node [style=none] (74) at (-3.75, 0) {$\triangleq$};
		\node [style=none] (75) at (0.25, 0) {$\triangleq$};
	\end{pgfonlayer}
	\begin{pgfonlayer}{edgelayer}
		\draw [style=blue edge] (4.center) to (5);
		\draw [style=blue edge] (0.center) to (1);
		\draw [style=blue edge, in=90, out=-150] (1) to (2.center);
		\draw [style=blue edge, in=90, out=-30] (1) to (3.center);
		\draw [style=blue edge] (8.center) to (9);
		\draw [style=blue edge] (10.center) to (11);
		\draw [style=blue edge, in=90, out=-150] (11) to (12.center);
		\draw [style=blue edge, in=90, out=-30] (11) to (13.center);
		\draw [style=blue edge] (13.center) to (14.center);
		\draw [style=blue edge] (16.center) to (17);
		\draw [style=blue edge, in=90, out=-150] (17) to (18.center);
		\draw [style=blue edge, in=90, out=-30] (17) to (19.center);
		\draw [style=blue edge] (18.center) to (21.center);
		\draw [style=blue edge] (22.center) to (23);
		\draw [style=blue edge] (25.center) to (26.center);
		\draw [style=blue edge] (28.center) to (29);
		\draw [style=blue edge, in=90, out=-150] (29) to (30.center);
		\draw [style=blue edge, in=90, out=-30] (29) to (31.center);
		\draw [style=blue edge] (34.center) to (35);
		\draw [style=blue edge, in=90, out=-150] (35) to (36.center);
		\draw [style=blue edge, in=90, out=-30] (35) to (37.center);
		\draw [style=blue edge] (31.center) to (38.center);
		\draw [style=blue edge] (39.center) to (40);
		\draw [style=blue edge, in=90, out=-150] (40) to (41.center);
		\draw [style=blue edge, in=90, out=-30] (40) to (42.center);
		\draw [style=blue edge] (43.center) to (44);
		\draw [style=blue edge, in=90, out=-150] (44) to (45.center);
		\draw [style=blue edge, in=90, out=-30] (44) to (46.center);
		\draw [style=blue edge] (41.center) to (47.center);
		\draw [style=blue edge] (50.center) to (51.center);
	\end{pgfonlayer}
\end{tikzpicture}
\end{equation}
For the sake of readability, we may denote a monad $(T,\eta^T,\mu^T)$ simply by $T$. Let us give some examples of monads in the category $\Set$.

\begin{example} \label{ex:identity}
    The \emph{identity} monad consists of the identity functor $1$ with both the unit and the multiplication being the identity natural transformation $1 \nto 1$. 
\end{example}

\begin{example} \label{ex:exception}
    The \emph{exception monad} $E$ is defined as follows. Let $1 = \{*\}$ be a singleton set. For any set $X$, let $\inl_X \colon X \to X + 1$ and $\inr_X \colon 1 \to X + 1$ be the canonical injections. The functor $E$ maps a set $X$ to the set $EX \triangleq X + 1$ and a function $f \colon X \to Y$ to $Ef \colon X+1 \to Y+1$, defined by $Ef(\inl_X(x)) \triangleq \inl_Y(f(x))$ and $Ef(\inr_X(*)) \triangleq \inr_Y(*)$. The unit is $\eta^E \triangleq \inl \colon 1 \nto E$ and the multiplication $\mu^E \colon EE \nto E$ is given by merging exceptions, \textit{i.e.}, $\mu^E_X(\inl_{X+1}(z)) \triangleq z$ for all $z \in X + 1$ and $\mu^E_X(\inr_{X+1}(*)) \triangleq \inr_X(*)$. 
\end{example}

\begin{example} \label{ex:reader}
    The \emph{reader monad} $R$ is defined as follows. Let $A$ be a set of labels. The functor $R$ maps a set $X$ to the set $RX \triangleq X^A$ and a function $f \colon X \to Y$ to $Rf \colon X^A \to Y^A$, defined by $Rf(h) \triangleq f \circ h$ for all $h \in X^A$. The unit $\eta^R \colon 1 \nto R$ produces a constant function $\eta^R_X(x) \triangleq (a \mapsto x)$. The multiplication $\mu^R \colon RR \nto R$ duplicates the input, \textit{i.e.}, $\mu^R_X(h) \triangleq (a \mapsto h(a)(a))$ for all $h \in (X^A)^A$.
\end{example}

\begin{example} \label{ex:writer}
    The \emph{writer monad} $W$ is defined as follows. Let $(M,\cdot,e)$ be a monoid. The functor $W$ maps a set $X$ to the set $WX \triangleq M \times X$ and a function $f \colon X \to Y$ to $Wf \colon M \times X \to M \times Y$, defined by $Wf(m,x) \triangleq (m,f(x))$. The unit $\eta^W \colon 1 \nto W$ outputs the monoid unit $\eta^W_X(x) \triangleq (e,x)$. The multiplication $\mu^W \colon WW \nto W$ implements the monoid multiplication, \textit{i.e.}, $\mu^W_X(m,(n,x)) \triangleq (m \cdot n,x)$. 
\end{example}

\begin{example} \label{ex:powerset}
    The \emph{powerset monad} $P$ is defined as follows. The functor $P$ maps a set $X$ to the set of its subsets $PX$ and a function $f\colon X \to Y$ to its direct image $Pf\colon PX \to PY$, defined by $Pf(U) \triangleq \{f(x) \mid x \in U\}$ for all $U \subseteq X$. The unit $\eta^P \colon 1 \nto P$ produces singletons $\eta^P_X(x) \triangleq \{x\}$. The multiplication $\mu^P \colon PP \nto P$ computes unions, \textit{i.e.}, $\mu^P_X(\mathcal{U}) \triangleq \bigcup \mathcal{U}$ for all $\mathcal{U} \in PPX$.
\end{example}

\begin{example} \label{ex:distribution}
    The \emph{distribution monad} $D$ is defined as follows. The functor $D$ maps a set $X$ to the set of all finitely-supported probability distributions on $X$, \textit{i.e.},
	\begin{equation}
    DX \triangleq \left\{ p \colon X \to [0,1] \mid \, \sum_{x \in X} p(x) = 1, \, \supp_X(p) \text{ is finite} \right\}
  \end{equation}
  where $\supp_X(p) = \{x \in X \mid p(x) > 0\}$. For a $p \in DX$, we may use the formal sum notation $p = \sum_{x \in X} p_x \cdot x$, where $p_x \triangleq p(x)$. A function $f \colon X \to Y$ is mapped to $Df \colon DX \to DY$, defined by $Df(\sum_{x \in X} p_x \cdot x) \triangleq \sum_{x \in X} p_x \cdot f(x)$. The unit $\eta^D \colon 1 \nto D$ produces Dirac distributions $\eta^D_X(x) \triangleq 1 \cdot x$. The multiplication $\mu^D \colon DD \nto D$ is a weighted average, defined by
  $\mu^D_X(q) \triangleq \sum_{p \in DX} q_p p_x \cdot x$ for all $q \in DDX$.
\end{example}

\begin{definition} Given two monads $S$, $T$ on $\C$, a \emph{distributive law} is a $\lambda \colon [S,T]$ such that the four following axioms hold:
\begin{align}
    & \lambda \circ T \eta^S = \eta^S T \label{ax:etap} \tag{$\eta^+$} \\
    & \lambda \circ T \mu^S = \mu^S T \circ S \lambda \circ \lambda S \label{ax:mup} \tag{$\mu^+$} \\
    & \lambda \circ \eta^T S = S \eta^T \label{ax:etam} \tag{$\eta^-$} \\
    & \lambda \circ \mu^T S = S \mu^T \circ \lambda T \circ T \lambda \label{ax:mum} \tag{$\mu^-$}
\end{align}
\end{definition}
String diagrammatically, the monads $S$, $T$, and the distributive law $\lambda \colon [S,T]$ are pictured as
\begin{equation}
    \begin{tikzpicture}
	\begin{pgfonlayer}{nodelayer}
		\node [style=none] (0) at (1.75, 1) {};
		\node [style=red dot] (1) at (1.75, 0) {};
		\node [style=none] (2) at (0.75, -1) {};
		\node [style=none] (3) at (2.75, -1) {};
		\node [style=none] (4) at (-3.25, 1) {};
		\node [style=red dot] (5) at (-3.25, 0) {};
		\node [style=none] (7) at (-0.75, 0) {$\mu^S$};
		\node [style=none] (27) at (-4.75, 0) {$\eta^S$};
		\node [style=none] (50) at (-7.25, 1) {};
		\node [style=none] (51) at (-7.25, -1) {};
		\node [style=none] (52) at (-9, 0) {$1_S$};
		\node [style=none] (53) at (16, 1) {};
		\node [style=blue dot] (54) at (16, 0) {};
		\node [style=none] (55) at (15, -1) {};
		\node [style=none] (56) at (17, -1) {};
		\node [style=none] (57) at (11, 1) {};
		\node [style=blue dot] (58) at (11, 0) {};
		\node [style=none] (59) at (13.5, 0) {$\mu^T$};
		\node [style=none] (60) at (9.5, 0) {$\eta^T$};
		\node [style=none] (61) at (7, 1) {};
		\node [style=none] (62) at (7, -1) {};
		\node [style=none] (63) at (5.25, 0) {$1_T$};
		\node [style=none] (64) at (21, 1) {};
		\node [style=none] (65) at (23, 1) {};
		\node [style=none] (66) at (21, -1) {};
		\node [style=none] (67) at (23, -1) {};
		\node [style=none] (68) at (23, 1) {};
		\node [style=none] (69) at (21, -1) {};
		\node [style=none] (70) at (19.5, 0) {$\lambda$};
		\node [style=none] (71) at (7, 1.5) {$T$};
		\node [style=none] (72) at (11, 1.5) {$T$};
		\node [style=none] (73) at (16, 1.5) {$T$};
		\node [style=none] (74) at (7, -1.5) {$T$};
		\node [style=none] (76) at (15, -1.5) {$T$};
		\node [style=none] (77) at (17, -1.5) {$T$};
		\node [style=none] (78) at (21, -1.5) {$T$};
		\node [style=none] (79) at (23, 1.5) {$T$};
		\node [style=none] (80) at (21, 1.5) {$S$};
		\node [style=none] (81) at (23, -1.5) {$S$};
		\node [style=none] (82) at (2.75, -1.5) {$S$};
		\node [style=none] (83) at (0.75, -1.5) {$S$};
		\node [style=none] (84) at (1.75, 1.5) {$S$};
		\node [style=none] (85) at (-3.25, 1.5) {$S$};
		\node [style=none] (86) at (-7.25, 1.5) {$S$};
		\node [style=none] (87) at (-7.25, -1.5) {$S$};
		\node [style=none] (95) at (-8, 0) {$\triangleq$};
		\node [style=none] (96) at (-4, 0) {$\triangleq$};
		\node [style=none] (97) at (0.25, 0) {$\triangleq$};
		\node [style=none] (98) at (6.25, 0) {$\triangleq$};
		\node [style=none] (99) at (10.25, 0) {$\triangleq$};
		\node [style=none] (100) at (14.5, 0) {$\triangleq$};
		\node [style=none] (101) at (20.25, 0) {$\triangleq$};
	\end{pgfonlayer}
	\begin{pgfonlayer}{edgelayer}
		\draw [style=red edge] (4.center) to (5);
		\draw [style=red edge] (0.center) to (1);
		\draw [style=red edge, in=90, out=-150] (1) to (2.center);
		\draw [style=red edge, in=90, out=-30] (1) to (3.center);
		\draw [style=red edge] (50.center) to (51.center);
		\draw [style=blue edge] (57.center) to (58);
		\draw [style=blue edge] (53.center) to (54);
		\draw [style=blue edge, in=90, out=-150] (54) to (55.center);
		\draw [style=blue edge, in=90, out=-30] (54) to (56.center);
		\draw [style=blue edge] (61.center) to (62.center);
		\draw [style=red edge, in=90, out=-90] (64.center) to (67.center);
		\draw [style=thick white, in=-90, out=90] (66.center) to (65.center);
		\draw [style=blue edge, in=-90, out=90] (69.center) to (68.center);
	\end{pgfonlayer}
\end{tikzpicture}
\end{equation} 
and the distributive law axioms are 
\begin{equation}
    \begin{tikzpicture}
	\begin{pgfonlayer}{nodelayer}
		\node [style=none] (85) at (-13.5, -2.5) {\eqref{ax:etap}};
		\node [style=none] (101) at (-5.5, -2.5) {\eqref{ax:mup}};
		\node [style=none] (114) at (-5.5, 0) {$=$};
		\node [style=none] (167) at (-4, 1) {};
		\node [style=red dot] (168) at (-4, 0.5) {};
		\node [style=none] (169) at (-4.5, 0) {};
		\node [style=none] (170) at (-3.5, 0) {};
		\node [style=none] (171) at (-4.5, -1) {};
		\node [style=none] (172) at (-3.5, -1) {};
		\node [style=none] (173) at (-5, -1) {};
		\node [style=none] (174) at (-3, 0.5) {};
		\node [style=none] (175) at (-7, 1) {};
		\node [style=red dot] (176) at (-7, -0.5) {};
		\node [style=none] (177) at (-7.5, -1) {};
		\node [style=none] (178) at (-6.5, -1) {};
		\node [style=none] (181) at (-8, -0.5) {};
		\node [style=none] (182) at (-6, 1) {};
		\node [style=none] (183) at (-8, -0.5) {};
		\node [style=none] (184) at (-6, 1) {};
		\node [style=none] (185) at (-8, -1) {};
		\node [style=none] (186) at (-3, 1) {};
		\node [style=none] (187) at (-13.5, 0) {$=$};
		\node [style=none] (188) at (-12, 1) {};
		\node [style=red dot] (189) at (-12, 0.5) {};
		\node [style=none] (194) at (-13, -1) {};
		\node [style=none] (195) at (-11, 0.5) {};
		\node [style=none] (196) at (-15, 1) {};
		\node [style=red dot] (197) at (-15, -0.5) {};
		\node [style=none] (200) at (-16, -0.5) {};
		\node [style=none] (201) at (-14, 1) {};
		\node [style=none] (202) at (-16, -0.5) {};
		\node [style=none] (203) at (-14, 1) {};
		\node [style=none] (204) at (-16, -1) {};
		\node [style=none] (205) at (-11, 1) {};
		\node [style=blue dot] (207) at (9, -0.5) {};
		\node [style=none] (208) at (8.5, -1) {};
		\node [style=none] (209) at (9.5, -1) {};
		\node [style=none] (210) at (10, -1) {};
		\node [style=none] (211) at (10, -0.5) {};
		\node [style=none] (212) at (8, 1) {};
		\node [style=none] (213) at (9, 1) {};
		\node [style=none] (214) at (10.5, 0) {$=$};
		\node [style=none] (215) at (12, 1) {};
		\node [style=blue dot] (216) at (12, 0.5) {};
		\node [style=none] (217) at (11.5, 0) {};
		\node [style=none] (218) at (12.5, 0) {};
		\node [style=none] (219) at (13, -1) {};
		\node [style=none] (220) at (11, 0.5) {};
		\node [style=none] (221) at (11, 1) {};
		\node [style=none] (222) at (11.5, -1) {};
		\node [style=none] (223) at (12.5, -1) {};
		\node [style=blue dot] (224) at (1, -0.5) {};
		\node [style=none] (227) at (2, -1) {};
		\node [style=none] (228) at (2, -0.5) {};
		\node [style=none] (229) at (0, 1) {};
		\node [style=none] (230) at (1, 1) {};
		\node [style=none] (231) at (2.5, 0) {$=$};
		\node [style=none] (232) at (4, 1) {};
		\node [style=blue dot] (233) at (4, 0.5) {};
		\node [style=none] (236) at (5, -1) {};
		\node [style=none] (237) at (3, 0.5) {};
		\node [style=none] (238) at (3, 1) {};
		\node [style=none] (239) at (2.5, -2.5) {\eqref{ax:etam}};
		\node [style=none] (240) at (10.5, -2.5) {\eqref{ax:mum}};
		\node [style=none] (241) at (-15, 1.5) {$S$};
		\node [style=none] (242) at (-12, 1.5) {$S$};
		\node [style=none] (243) at (-7, 1.5) {$S$};
		\node [style=none] (244) at (-4, 1.5) {$S$};
		\node [style=none] (245) at (0, 1.5) {$S$};
		\node [style=none] (246) at (3, 1.5) {$S$};
		\node [style=none] (247) at (8, 1.5) {$S$};
		\node [style=none] (248) at (11, 1.5) {$S$};
		\node [style=none] (249) at (-7.5, -1.5) {$S$};
		\node [style=none] (250) at (-6.5, -1.5) {$S$};
		\node [style=none] (251) at (-4.5, -1.5) {$S$};
		\node [style=none] (252) at (-3.5, -1.5) {$S$};
		\node [style=none] (253) at (2, -1.5) {$S$};
		\node [style=none] (254) at (5, -1.5) {$S$};
		\node [style=none] (255) at (10, -1.5) {$S$};
		\node [style=none] (256) at (13, -1.5) {$S$};
		\node [style=none] (257) at (-14, 1.5) {$T$};
		\node [style=none] (258) at (-11, 1.5) {$T$};
		\node [style=none] (259) at (-6, 1.5) {$T$};
		\node [style=none] (260) at (-3, 1.5) {$T$};
		\node [style=none] (261) at (1, 1.5) {$T$};
		\node [style=none] (262) at (4, 1.5) {$T$};
		\node [style=none] (263) at (9, 1.5) {$T$};
		\node [style=none] (264) at (12, 1.5) {$T$};
		\node [style=none] (265) at (-16, -1.5) {$T$};
		\node [style=none] (266) at (-13, -1.5) {$T$};
		\node [style=none] (267) at (-8, -1.5) {$T$};
		\node [style=none] (268) at (8.5, -1.5) {$T$};
		\node [style=none] (269) at (9.5, -1.5) {$T$};
		\node [style=none] (270) at (11.5, -1.5) {$T$};
		\node [style=none] (271) at (12.5, -1.5) {$T$};
		\node [style=none] (272) at (-5, -1.5) {$T$};
	\end{pgfonlayer}
	\begin{pgfonlayer}{edgelayer}
		\draw [style=red edge] (167.center) to (168);
		\draw [style=red edge, in=90, out=-150] (168) to (169.center);
		\draw [style=red edge, in=90, out=-30] (168) to (170.center);
		\draw [style=red edge] (169.center) to (171.center);
		\draw [style=red edge] (172.center) to (170.center);
		\draw [style=thick white, in=-90, out=90, looseness=1.25] (173.center) to (174.center);
		\draw [style=blue edge, in=-90, out=90, looseness=1.25] (173.center) to (174.center);
		\draw [style=red edge] (175.center) to (176);
		\draw [style=red edge, in=90, out=-150] (176) to (177.center);
		\draw [style=red edge, in=90, out=-30] (176) to (178.center);
		\draw [style=thick white, in=-90, out=90, looseness=0.75] (181.center) to (182.center);
		\draw [style=thick white, in=-90, out=90, looseness=0.75] (181.center) to (182.center);
		\draw [style=thick white, in=-90, out=90] (183.center) to (184.center);
		\draw [style=blue edge, in=-90, out=90] (183.center) to (184.center);
		\draw [style=blue edge] (185.center) to (183.center);
		\draw [style=blue edge] (186.center) to (174.center);
		\draw [style=red edge] (188.center) to (189);
		\draw [style=thick white, in=-90, out=90, looseness=1.25] (194.center) to (195.center);
		\draw [style=blue edge, in=-90, out=90, looseness=1.25] (194.center) to (195.center);
		\draw [style=red edge] (196.center) to (197);
		\draw [style=thick white, in=-90, out=90, looseness=0.75] (200.center) to (201.center);
		\draw [style=thick white, in=-90, out=90, looseness=0.75] (200.center) to (201.center);
		\draw [style=thick white, in=-90, out=90] (202.center) to (203.center);
		\draw [style=blue edge, in=-90, out=90] (202.center) to (203.center);
		\draw [style=blue edge] (204.center) to (202.center);
		\draw [style=blue edge] (205.center) to (195.center);
		\draw [style=blue edge, in=90, out=-150] (207) to (208.center);
		\draw [style=blue edge, in=90, out=-30] (207) to (209.center);
		\draw [style=red edge] (210.center) to (211.center);
		\draw [style=red edge, in=-90, out=90] (211.center) to (212.center);
		\draw [style=thick white] (213.center) to (207);
		\draw [style=blue edge] (213.center) to (207);
		\draw [style=blue edge] (215.center) to (216);
		\draw [style=blue edge, in=90, out=-150] (216) to (217.center);
		\draw [style=blue edge, in=90, out=-30] (216) to (218.center);
		\draw [style=red edge, in=-90, out=90] (219.center) to (220.center);
		\draw [style=red edge] (220.center) to (221.center);
		\draw [style=thick white] (222.center) to (217.center);
		\draw [style=thick white] (223.center) to (218.center);
		\draw [style=blue edge] (223.center) to (218.center);
		\draw [style=blue edge] (222.center) to (217.center);
		\draw [style=red edge] (227.center) to (228.center);
		\draw [style=red edge, in=-90, out=90] (228.center) to (229.center);
		\draw [style=thick white] (230.center) to (224);
		\draw [style=blue edge] (230.center) to (224);
		\draw [style=blue edge] (232.center) to (233);
		\draw [style=red edge, in=-90, out=90] (236.center) to (237.center);
		\draw [style=red edge] (237.center) to (238.center);
	\end{pgfonlayer}
\end{tikzpicture}
\end{equation}
Any distributive law $\lambda \colon [S,T]$ yields a monad $\comp{S}{\lambda}{T} \triangleq (ST,\eta^S \eta^T,\mu^S \mu^T \circ S\lambda T)$, also denoted by $\comp{S}{}{T}$ when no confusion can arise, and pictured as
\begin{equation} \label{sd:composite_monad}
    \begin{tikzpicture}
	\begin{pgfonlayer}{nodelayer}
		\node [style=none] (0) at (3.5, 1) {};
		\node [style=red dot] (1) at (3.5, 0.5) {};
		\node [style=none] (2) at (3, 0) {};
		\node [style=none] (3) at (4, 0) {};
		\node [style=none] (4) at (-3.75, 1) {};
		\node [style=red dot] (5) at (-3.75, 0) {};
		\node [style=none] (50) at (-10.5, 1) {};
		\node [style=none] (51) at (-10.5, -1) {};
		\node [style=none] (53) at (5.5, 1) {};
		\node [style=blue dot] (54) at (5.5, 0.5) {};
		\node [style=none] (55) at (5, 0) {};
		\node [style=none] (56) at (6, 0) {};
		\node [style=none] (57) at (-2.75, 1) {};
		\node [style=blue dot] (58) at (-2.75, 0) {};
		\node [style=none] (61) at (-9.5, 1) {};
		\node [style=none] (62) at (-9.5, -1) {};
		\node [style=none] (64) at (4, 0) {};
		\node [style=none] (65) at (5, 0) {};
		\node [style=none] (66) at (4, -1) {};
		\node [style=none] (67) at (5, -1) {};
		\node [style=none] (68) at (5, 0) {};
		\node [style=none] (69) at (4, -1) {};
		\node [style=none] (73) at (3, -1) {};
		\node [style=none] (74) at (6, -1) {};
		\node [style=none] (75) at (4, -1) {};
		\node [style=none] (76) at (5, -1) {};
		\node [style=none] (77) at (-13, 0) {$\comp{S}{}{T}$};
		\node [style=none] (78) at (-6.25, 0) {$\eta^{\comp{S}{}{T}}$};
		\node [style=none] (79) at (0.75, 0) {$\mu^{\comp{S}{}{T}}$};
		\node [style=none] (80) at (-10.5, 1.5) {$S$};
		\node [style=none] (81) at (-3.75, 1.5) {$S$};
		\node [style=none] (82) at (3.5, 1.5) {$S$};
		\node [style=none] (83) at (-10.5, -1.5) {$S$};
		\node [style=none] (84) at (3, -1.5) {$S$};
		\node [style=none] (85) at (-9.5, 1.5) {$T$};
		\node [style=none] (86) at (-2.75, 1.5) {$T$};
		\node [style=none] (87) at (5.5, 1.5) {$T$};
		\node [style=none] (88) at (-9.5, -1.5) {$T$};
		\node [style=none] (89) at (4, -1.5) {$T$};
		\node [style=none] (90) at (6, -1.5) {$T$};
		\node [style=none] (91) at (5, -1.5) {$S$};
		\node [style=none] (92) at (-11.5, 0) {$\triangleq$};
		\node [style=none] (93) at (-4.75, 0) {$\triangleq$};
		\node [style=none] (94) at (2, 0) {$\triangleq$};
	\end{pgfonlayer}
	\begin{pgfonlayer}{edgelayer}
		\draw [style=red edge] (4.center) to (5);
		\draw [style=red edge] (0.center) to (1);
		\draw [style=red edge, in=90, out=-150] (1) to (2.center);
		\draw [style=red edge, in=90, out=-30] (1) to (3.center);
		\draw [style=red edge] (50.center) to (51.center);
		\draw [style=blue edge] (57.center) to (58);
		\draw [style=blue edge] (53.center) to (54);
		\draw [style=blue edge, in=90, out=-150] (54) to (55.center);
		\draw [style=blue edge, in=90, out=-30] (54) to (56.center);
		\draw [style=blue edge] (61.center) to (62.center);
		\draw [style=red edge, in=90, out=-90] (64.center) to (67.center);
		\draw [style=thick white, in=-90, out=90] (66.center) to (65.center);
		\draw [style=blue edge, in=-90, out=90] (69.center) to (68.center);
		\draw [style=red edge] (73.center) to (2.center);
		\draw [style=blue edge] (74.center) to (56.center);
	\end{pgfonlayer}
\end{tikzpicture}
\end{equation}

\subsection{Weak Distributive Laws}

\begin{definition}
Given two monads $S$, $T$ on $\C$, a \emph{weak distributive law} is a $\lambda \colon [S,T]$ such that the axioms~\eqref{ax:etap}, \eqref{ax:mup} and \eqref{ax:mum} hold.
\end{definition}

In this definition, the weakening comes from dropping the \eqref{ax:etam} axiom, as done in~\cite{Garner20}. There are other possible weakenings that we briefly discuss in Remark~\ref{rem:other_weakenings}. An interesting natural transformation emerges from the structure of a weak distributive law. Consider 
\begin{equation} \label{eq:kappa_def}
\kappa \triangleq \begin{tikzcd}
  ST \arrow[r, "\eta^T ST", Rightarrow] & TST \arrow[r, "\lambda T", Rightarrow] & STT \arrow[r, "S\mu^T", Rightarrow] & ST
  \end{tikzcd}
\end{equation}

If $\lambda$ were a distributive law, then $\kappa$ would be the identity. In the weak case though, $\kappa$ still is an idempotent that is compatible with units and multiplications in the following sense.
\begin{lemma} \label{lem:four_eq}
The natural transformation $\kappa$ satisfies the following equations.
\begin{align}
    & \kappa \circ \kappa = \kappa \label{eq:kappa_idem} \\ 
    & \kappa \circ \lambda = \lambda \label{eq:kappa_lambda} \\ 
    & \kappa \circ \eta^S \eta^T = \eta^S \eta^T \label{eq:kappa_eta} \\ 
    & \kappa \circ \mu^S \mu^T \circ S \lambda T = \mu^S \mu^T \circ S\lambda T \circ \kappa\kappa \label{eq:kappa_mu}
\end{align}
\end{lemma}


As $\kappa \colon ST \nto ST$ is an idempotent natural transformation and we assume $\C$ to be idempotent complete, there are a $\C$-endofunctor $K$ and two natural transformations $\pi \colon ST \nto K$ and $\iota \colon K \nto ST$ such that $\iota \circ \pi = \kappa$ and $\pi \circ \iota = 1_K$. Using these, one can build a monad that is to be interpreted as a weak composition of $S$ and $T$, namely $S \bullet_\lambda T = (K,\pi \circ \eta^S\eta^T,\pi \circ \mu^S\mu^T \circ S\lambda T \circ \iota\iota)$, also denoted by $S \bullet T$ when the context is clear. See Appendix~\ref{app:laws} for a proof that this is indeed a monad. We now introduce some string diagrammatic notation that will make computations practicable:
\begin{equation} \label{eq:kappa_notations}
\begin{tikzpicture}
	\begin{pgfonlayer}{nodelayer}
		\node [style=none] (40) at (-9, -1) {};
		\node [style=none] (41) at (-8, 0) {};
		\node [style=none] (42) at (-7, -1) {};
		\node [style=none] (43) at (-8.025, -0.025) {};
		\node [style=none] (44) at (-8.025, 1) {};
		\node [style=none] (45) at (-7.975, -0.025) {};
		\node [style=none] (46) at (-7.975, 1) {};
		\node [style=none] (47) at (-3, 1) {};
		\node [style=none] (48) at (-2.025, 0) {};
		\node [style=none] (49) at (-1, 1) {};
		\node [style=none] (50) at (-2.025, -1) {};
		\node [style=none] (51) at (-2.025, 0.025) {};
		\node [style=none] (52) at (-1.975, -1) {};
		\node [style=none] (53) at (-1.975, 0.025) {};
		\node [style=none] (55) at (-13, -1) {};
		\node [style=none] (57) at (-13.025, -1.025) {};
		\node [style=none] (58) at (-13.025, 1) {};
		\node [style=none] (59) at (-12.975, -1.025) {};
		\node [style=none] (60) at (-12.975, 1) {};
		\node [style=none] (61) at (-14.75, 0) {$1_K$};
		\node [style=none] (62) at (-10.5, 0) {$\pi$};
		\node [style=none] (63) at (-4.5, 0) {$\iota$};
		\node [style=none] (64) at (11, 0) {};
		\node [style=none] (65) at (11.5, 0.5) {};
		\node [style=none] (66) at (12, 0) {};
		\node [style=none] (67) at (11.475, 0.475) {};
		\node [style=none] (68) at (11.475, 1) {};
		\node [style=none] (69) at (11.525, 0.475) {};
		\node [style=none] (70) at (11.525, 1) {};
		\node [style=none] (71) at (11, 0) {};
		\node [style=none] (72) at (11.475, -0.5) {};
		\node [style=none] (73) at (12, 0) {};
		\node [style=none] (74) at (11.475, -1) {};
		\node [style=none] (75) at (11.475, -0.475) {};
		\node [style=none] (76) at (11.525, -1) {};
		\node [style=none] (77) at (11.525, -0.475) {};
		\node [style=none] (78) at (14.025, -1) {};
		\node [style=none] (79) at (14, -1.025) {};
		\node [style=none] (80) at (14, 1) {};
		\node [style=none] (81) at (14.05, -1.025) {};
		\node [style=none] (82) at (14.05, 1) {};
		\node [style=none] (83) at (13, 0) {$=$};
		\node [style=none] (84) at (3, -1) {};
		\node [style=none] (85) at (3.5, -0.5) {};
		\node [style=none] (86) at (4, -1) {};
		\node [style=none] (87) at (3.475, -0.525) {};
		\node [style=none] (88) at (3.475, 0) {};
		\node [style=none] (89) at (3.525, -0.525) {};
		\node [style=none] (90) at (3.525, 0) {};
		\node [style=none] (91) at (3, 1) {};
		\node [style=none] (92) at (3.475, 0.5) {};
		\node [style=none] (93) at (4, 1) {};
		\node [style=none] (94) at (3.475, 0) {};
		\node [style=none] (95) at (3.475, 0.525) {};
		\node [style=none] (96) at (3.525, 0) {};
		\node [style=none] (97) at (3.525, 0.525) {};
		\node [style=none] (98) at (6.05, 0) {};
		\node [style=blue dot] (99) at (6.05, -0.5) {};
		\node [style=none] (100) at (7.05, 1) {};
		\node [style=blue dot] (101) at (7.05, 0.5) {};
		\node [style=none] (102) at (6.05, 0) {};
		\node [style=none] (103) at (7.05, 0) {};
		\node [style=none] (104) at (7.05, -1) {};
		\node [style=none] (105) at (6.5, 1) {};
		\node [style=none] (106) at (6.5, -0.5) {};
		\node [style=none] (107) at (6.5, -1) {};
		\node [style=none] (108) at (5, 0) {$=$};
		\node [style=blue dot] (129) at (7.05, 0.5) {};
		\node [style=none] (130) at (6.05, 0) {};
		\node [style=blue dot] (131) at (7.05, 0.5) {};
		\node [style=none] (132) at (6.05, 0) {};
		\node [style=none] (135) at (-13.75, 0) {$\triangleq$};
		\node [style=none] (136) at (-9.75, 0) {$\triangleq$};
		\node [style=none] (137) at (-3.75, 0) {$\triangleq$};
		\node [style=none] (138) at (-9, -1.5) {$S$};
		\node [style=none] (139) at (3, -1.5) {$S$};
		\node [style=none] (140) at (6.5, -1.5) {$S$};
		\node [style=none] (141) at (3, 1.5) {$S$};
		\node [style=none] (142) at (-3, 1.5) {$S$};
		\node [style=none] (143) at (6.5, 1.5) {$S$};
		\node [style=none] (144) at (-7, -1.5) {$T$};
		\node [style=none] (145) at (-1, 1.5) {$T$};
		\node [style=none] (146) at (4, 1.5) {$T$};
		\node [style=none] (147) at (4, -1.5) {$T$};
		\node [style=none] (148) at (7.05, 1.5) {$T$};
		\node [style=none] (149) at (7.05, -1.5) {$T$};
		\node [style=none] (150) at (-8, 1.5) {$K$};
		\node [style=none] (151) at (-13, 1.5) {$K$};
		\node [style=none] (152) at (-13, -1.5) {$K$};
		\node [style=none] (153) at (-2, -1.5) {$K$};
		\node [style=none] (154) at (11.5, 1.5) {$K$};
		\node [style=none] (155) at (11.5, -1.5) {$K$};
		\node [style=none] (156) at (14, 1.5) {$K$};
		\node [style=none] (157) at (14, -1.5) {$K$};
	\end{pgfonlayer}
	\begin{pgfonlayer}{edgelayer}
		\draw [style=red edge, in=180, out=90] (40.center) to (41.center);
		\draw [style=blue edge, in=0, out=90] (42.center) to (41.center);
		\draw [style=red edge] (43.center) to (44.center);
		\draw [style=blue edge] (45.center) to (46.center);
		\draw [style=red edge, in=630, out=-180] (48.center) to (47.center);
		\draw [style=blue edge, in=-90, out=0] (48.center) to (49.center);
		\draw [style=red edge] (50.center) to (51.center);
		\draw [style=blue edge] (52.center) to (53.center);
		\draw [style=red edge] (57.center) to (58.center);
		\draw [style=blue edge] (59.center) to (60.center);
		\draw [style=red edge, in=180, out=90] (64.center) to (65.center);
		\draw [style=blue edge, in=0, out=90] (66.center) to (65.center);
		\draw [style=red edge] (67.center) to (68.center);
		\draw [style=blue edge] (69.center) to (70.center);
		\draw [style=red edge, in=630, out=-180] (72.center) to (71.center);
		\draw [style=blue edge, in=-90, out=0] (72.center) to (73.center);
		\draw [style=red edge] (74.center) to (75.center);
		\draw [style=blue edge] (76.center) to (77.center);
		\draw [style=red edge] (79.center) to (80.center);
		\draw [style=blue edge] (81.center) to (82.center);
		\draw [style=red edge, in=180, out=90] (84.center) to (85.center);
		\draw [style=blue edge, in=0, out=90] (86.center) to (85.center);
		\draw [style=red edge] (87.center) to (88.center);
		\draw [style=blue edge] (89.center) to (90.center);
		\draw [style=red edge, in=630, out=-180] (92.center) to (91.center);
		\draw [style=blue edge, in=-90, out=0] (92.center) to (93.center);
		\draw [style=red edge] (94.center) to (95.center);
		\draw [style=blue edge] (96.center) to (97.center);
		\draw [style=blue edge] (98.center) to (99);
		\draw [style=blue edge] (100.center) to (101);
		\draw [style=blue edge] (101) to (103.center);
		\draw [style=blue edge] (103.center) to (104.center);
		\draw [style=red edge, in=-90, out=90] (106.center) to (105.center);
		\draw [style=red edge] (106.center) to (107.center);
		\draw [style=thick white, in=90, out=-150] (129) to (130.center);
		\draw [style=blue edge, in=90, out=-150] (131) to (132.center);
	\end{pgfonlayer}
\end{tikzpicture}
\end{equation} 

String diagrammatically, the equations from Lemma~\ref{lem:four_eq} can be represented as below.
\begin{equation}
\input{1_kappa_equations.tikz}
\end{equation}

The composite monad $S \bullet T$ is, in turn, as follows -- compare with diagrams \eqref{sd:composite_monad} for $\comp{S}{}{T}$:
\begin{equation}
\begin{tikzpicture}
	\begin{pgfonlayer}{nodelayer}
		\node [style=none] (4) at (-4.25, 0.25) {};
		\node [style=red dot] (5) at (-4.25, 0) {};
		\node [style=none] (57) at (-3.25, 0.25) {};
		\node [style=blue dot] (58) at (-3.25, 0) {};
		\node [style=none] (77) at (-12.5, 0) {$1_{\wcomp{S}{}{T}}$};
		\node [style=none] (78) at (-6, 0) {$\eta^{\wcomp{S}{}{T}}$};
		\node [style=none] (79) at (1.25, 0) {$\mu^{\wcomp{S}{}{T}}$};
		\node [style=none] (80) at (-4.25, 0.25) {};
		\node [style=none] (81) at (-3.75, 0.75) {};
		\node [style=none] (82) at (-3.25, 0.25) {};
		\node [style=none] (83) at (-3.775, 0.725) {};
		\node [style=none] (84) at (-3.775, 1.25) {};
		\node [style=none] (85) at (-3.725, 0.725) {};
		\node [style=none] (86) at (-3.725, 1.25) {};
		\node [style=none] (87) at (3.5, 0.25) {};
		\node [style=none] (88) at (4, 0.75) {};
		\node [style=none] (89) at (4.5, 0.25) {};
		\node [style=none] (90) at (3.975, 0.725) {};
		\node [style=none] (91) at (3.975, 1.25) {};
		\node [style=none] (92) at (4.025, 0.725) {};
		\node [style=none] (93) at (4.025, 1.25) {};
		\node [style=none] (94) at (-10.5, -1.25) {};
		\node [style=none] (95) at (-10.525, -1.275) {};
		\node [style=none] (96) at (-10.525, 1.25) {};
		\node [style=none] (97) at (-10.475, -1.275) {};
		\node [style=none] (98) at (-10.475, 1.25) {};
		\node [style=none] (114) at (3.5, -1.25) {};
		\node [style=none] (116) at (3.475, -1.275) {};
		\node [style=none] (117) at (3.475, -1.25) {};
		\node [style=none] (118) at (3.525, -1.275) {};
		\node [style=none] (119) at (3.525, -1.25) {};
		\node [style=none] (120) at (3, -0.5) {};
		\node [style=none] (121) at (3.475, -1) {};
		\node [style=none] (122) at (3.475, -1.25) {};
		\node [style=none] (123) at (3.475, -0.975) {};
		\node [style=none] (124) at (3.525, -1.25) {};
		\node [style=none] (125) at (3.525, -1) {};
		\node [style=none] (127) at (4.5, -1.25) {};
		\node [style=none] (129) at (4.475, -1.275) {};
		\node [style=none] (130) at (4.475, -1.25) {};
		\node [style=none] (131) at (4.525, -1.275) {};
		\node [style=none] (132) at (4.525, -1.25) {};
		\node [style=none] (133) at (4.475, -1) {};
		\node [style=none] (134) at (5, -0.5) {};
		\node [style=none] (135) at (4.475, -1.25) {};
		\node [style=none] (136) at (4.475, -0.975) {};
		\node [style=none] (137) at (4.525, -1.25) {};
		\node [style=none] (138) at (4.525, -1) {};
		\node [style=none] (139) at (3.5, 0.25) {};
		\node [style=red dot] (140) at (3.5, 0) {};
		\node [style=none] (141) at (3, -0.5) {};
		\node [style=none] (142) at (4.5, 0.25) {};
		\node [style=blue dot] (143) at (4.5, 0) {};
		\node [style=none] (144) at (5, -0.5) {};
		\node [style=none] (145) at (3, -0.5) {};
		\node [style=none] (146) at (5, -0.5) {};
		\node [style=none] (147) at (3.525, -1.25) {};
		\node [style=none] (148) at (3.525, -0.975) {};
		\node [style=none] (149) at (4.5, -1) {};
		\node [style=none] (150) at (5, -0.5) {};
		\node [style=none] (151) at (5, -0.5) {};
		\node [style=none] (152) at (-11.25, 0) {$\triangleq$};
		\node [style=none] (153) at (-5, 0) {$\triangleq$};
		\node [style=none] (154) at (2.25, 0) {$\triangleq$};
		\node [style=none] (155) at (-10.5, -1.75) {$K$};
		\node [style=none] (156) at (-10.5, 1.75) {$K$};
		\node [style=none] (157) at (-3.75, 1.75) {$K$};
		\node [style=none] (158) at (4, 1.75) {$K$};
		\node [style=none] (159) at (3.5, -1.75) {$K$};
		\node [style=none] (160) at (4.5, -1.75) {$K$};
	\end{pgfonlayer}
	\begin{pgfonlayer}{edgelayer}
		\draw [style=red edge] (4.center) to (5);
		\draw [style=blue edge] (57.center) to (58);
		\draw [style=red edge, in=180, out=90] (80.center) to (81.center);
		\draw [style=blue edge, in=0, out=90] (82.center) to (81.center);
		\draw [style=red edge] (83.center) to (84.center);
		\draw [style=blue edge] (85.center) to (86.center);
		\draw [style=red edge, in=180, out=90] (87.center) to (88.center);
		\draw [style=blue edge, in=0, out=90] (89.center) to (88.center);
		\draw [style=red edge] (90.center) to (91.center);
		\draw [style=blue edge] (92.center) to (93.center);
		\draw [style=red edge] (95.center) to (96.center);
		\draw [style=blue edge] (97.center) to (98.center);
		\draw [style=red edge] (116.center) to (117.center);
		\draw [style=blue edge] (118.center) to (119.center);
		\draw [style=red edge, in=630, out=-180] (121.center) to (120.center);
		\draw [style=red edge] (122.center) to (123.center);
		\draw [style=blue edge] (124.center) to (125.center);
		\draw [style=red edge] (129.center) to (130.center);
		\draw [style=blue edge] (131.center) to (132.center);
		\draw [style=red edge] (135.center) to (136.center);
		\draw [style=blue edge] (137.center) to (138.center);
		\draw [style=red edge] (139.center) to (140);
		\draw [style=red edge, in=90, out=-150] (140) to (141.center);
		\draw [style=blue edge] (142.center) to (143);
		\draw [style=blue edge, in=90, out=-30] (143) to (144.center);
		\draw [style=red edge] (145.center) to (141.center);
		\draw [style=blue edge] (146.center) to (144.center);
		\draw [style=red edge, in=-180, out=-30] (140) to (138.center);
		\draw [style=thick white, in=0, out=-150] (143) to (125.center);
		\draw [style=blue edge, in=-150, out=0] (125.center) to (143);
		\draw [style=blue edge] (147.center) to (148.center);
		\draw [style=blue edge, in=-90, out=0] (149.center) to (150.center);
	\end{pgfonlayer}
\end{tikzpicture}
\end{equation}

\begin{example} \label{ex:trivial_wdl}
    Let $\gamma \colon S \nto T$ be a monad morphism, that is, a natural such that $\gamma \circ \eta^S = \eta^T$ and $\gamma \circ \mu^S = \mu^T \circ \gamma\gamma$. Then $\eta^S\mu^T \circ T\gamma \colon [S,T]$ is a weak distributive law -- and, in general, not a distributive law. Such distributive laws are called \emph{trivial}, because $\wcomp{S}{}{T} = T$~\cite[Proposition~2.12]{GoyTh}. In diagrams:
    \begin{equation}
    \begin{tikzpicture}
	\begin{pgfonlayer}{nodelayer}
		\node [style=none] (0) at (-6, -1) {};
		\node [style=none] (1) at (-6, 1) {};
		\node [style=none] (2) at (-6.175, 0) {};
		\node [style=none] (3) at (-6, -0.175) {};
		\node [style=none] (4) at (-5.825, 0) {};
		\node [style=none] (5) at (-6, 0.175) {};
		\node [style=none] (6) at (-6, 0) {};
		\node [style=none] (7) at (-7.5, 0) {$\gamma$};
		\node [style=none] (8) at (-1.275, -0.5) {};
		\node [style=red dot] (9) at (-1.275, -0.5) {};
		\node [style=none] (10) at (-1.275, -0.5) {};
		\node [style=none] (11) at (-1.275, 1) {};
		\node [style=none] (12) at (-1.45, 0.25) {};
		\node [style=none] (13) at (-1.275, 0.075) {};
		\node [style=none] (14) at (-1.1, 0.25) {};
		\node [style=none] (15) at (-1.275, 0.425) {};
		\node [style=none] (16) at (-1.275, 0.25) {};
		\node [style=none] (17) at (0.725, 1) {};
		\node [style=blue dot] (18) at (0.725, -0.5) {};
		\node [style=none] (19) at (-0.25, 0) {$=$};
		\node [style=none] (20) at (3.975, -0.5) {};
		\node [style=red dot] (21) at (3.975, -0.5) {};
		\node [style=none] (22) at (3.475, -1) {};
		\node [style=none] (23) at (4.475, -1) {};
		\node [style=none] (24) at (3.975, -0.5) {};
		\node [style=none] (25) at (3.975, 1) {};
		\node [style=none] (26) at (3.8, 0.25) {};
		\node [style=none] (27) at (3.975, 0.075) {};
		\node [style=none] (28) at (4.15, 0.25) {};
		\node [style=none] (29) at (3.975, 0.425) {};
		\node [style=none] (30) at (3.975, 0.25) {};
		\node [style=none] (31) at (5.25, 0) {$=$};
		\node [style=none] (32) at (7, 1) {};
		\node [style=blue dot] (33) at (7, 0.5) {};
		\node [style=none] (34) at (6.5, 0) {};
		\node [style=none] (35) at (7.5, 0) {};
		\node [style=none] (36) at (6.5, -1) {};
		\node [style=none] (37) at (6.5, 0) {};
		\node [style=none] (38) at (6.325, -0.5) {};
		\node [style=none] (39) at (6.5, -0.675) {};
		\node [style=none] (40) at (6.675, -0.5) {};
		\node [style=none] (41) at (6.5, -0.325) {};
		\node [style=none] (42) at (6.5, -0.5) {};
		\node [style=none] (43) at (7.5, -1) {};
		\node [style=none] (44) at (7.5, 0) {};
		\node [style=none] (45) at (7.325, -0.5) {};
		\node [style=none] (46) at (7.5, -0.675) {};
		\node [style=none] (47) at (7.675, -0.5) {};
		\node [style=none] (48) at (7.5, -0.325) {};
		\node [style=none] (49) at (7.5, -0.5) {};
		\node [style=none] (50) at (3, -2.5) {monad morphism axioms};
		\node [style=none] (51) at (13, 1) {};
		\node [style=red dot] (52) at (13, 0.5) {};
		\node [style=none] (53) at (14, 1) {};
		\node [style=blue dot] (54) at (14, 0.5) {};
		\node [style=none] (55) at (13.5, 0) {};
		\node [style=none] (56) at (14.5, 0) {};
		\node [style=none] (57) at (14.5, -1) {};
		\node [style=none] (58) at (14.5, -1) {};
		\node [style=none] (59) at (14.5, 0) {};
		\node [style=none] (60) at (14.325, -0.5) {};
		\node [style=none] (61) at (14.5, -0.675) {};
		\node [style=none] (62) at (14.675, -0.5) {};
		\node [style=none] (63) at (14.5, -0.325) {};
		\node [style=none] (64) at (14.5, -0.5) {};
		\node [style=none] (65) at (13.5, -1) {};
		\node [style=none] (66) at (14, -2.125) {trivial weak};
		\node [style=none] (67) at (14, -2.875) {distributive law};
		\node [style=none] (68) at (-6, -1.5) {$S$};
		\node [style=none] (69) at (3.5, -1.5) {$S$};
		\node [style=none] (70) at (4.5, -1.5) {$S$};
		\node [style=none] (71) at (6.5, -1.5) {$S$};
		\node [style=none] (72) at (7.5, -1.5) {$S$};
		\node [style=none] (73) at (14.5, -1.5) {$S$};
		\node [style=none] (74) at (13, 1.5) {$S$};
		\node [style=none] (75) at (-6, 1.5) {$T$};
		\node [style=none] (76) at (-1.25, 1.5) {$T$};
		\node [style=none] (77) at (0.725, 1.5) {$T$};
		\node [style=none] (78) at (3.975, 1.5) {$T$};
		\node [style=none] (79) at (7, 1.5) {$T$};
		\node [style=none] (80) at (14, 1.5) {$T$};
		\node [style=none] (81) at (13.5, -1.5) {$T$};
		\node [style=none] (82) at (-6.75, 0) {$\triangleq$};
	\end{pgfonlayer}
	\begin{pgfonlayer}{edgelayer}
		\draw [style=blue fill] (2.center)
			 to [in=180, out=90] (5.center)
			 to [in=90, out=0] (4.center)
			 to [in=-30, out=-150] cycle;
		\draw [style=red fill] (3.center)
			 to [in=-90, out=0] (4.center)
			 to (2.center)
			 to [in=-180, out=-90] cycle;
		\draw [style=blue edge] (6.center) to (1.center);
		\draw [style=red edge] (0.center) to (6.center);
		\draw [style=red edge] (8.center) to (9);
		\draw [style=blue fill] (12.center)
			 to [in=180, out=90] (15.center)
			 to [in=90, out=0] (14.center)
			 to [in=-30, out=-150] cycle;
		\draw [style=red fill] (13.center)
			 to [in=-90, out=0] (14.center)
			 to (12.center)
			 to [in=-180, out=-90] cycle;
		\draw [style=blue edge] (16.center) to (11.center);
		\draw [style=red edge] (10.center) to (16.center);
		\draw [style=blue edge] (17.center) to (18);
		\draw [style=red edge] (20.center) to (21);
		\draw [style=red edge, in=90, out=-150] (21) to (22.center);
		\draw [style=red edge, in=90, out=-30] (21) to (23.center);
		\draw [style=blue fill] (26.center)
			 to [in=180, out=90] (29.center)
			 to [in=90, out=0] (28.center)
			 to [in=-30, out=-150] cycle;
		\draw [style=red fill] (27.center)
			 to [in=-90, out=0] (28.center)
			 to (26.center)
			 to [in=-180, out=-90] cycle;
		\draw [style=blue edge] (30.center) to (25.center);
		\draw [style=red edge] (24.center) to (30.center);
		\draw [style=blue edge] (32.center) to (33);
		\draw [style=blue edge, in=90, out=-150] (33) to (34.center);
		\draw [style=blue edge, in=90, out=-30] (33) to (35.center);
		\draw [style=blue fill] (38.center)
			 to [in=180, out=90] (41.center)
			 to [in=90, out=0] (40.center)
			 to [in=-30, out=-150] cycle;
		\draw [style=red fill] (39.center)
			 to [in=-90, out=0] (40.center)
			 to (38.center)
			 to [in=-180, out=-90] cycle;
		\draw [style=blue edge] (42.center) to (37.center);
		\draw [style=red edge] (36.center) to (42.center);
		\draw [style=blue fill] (45.center)
			 to [in=180, out=90] (48.center)
			 to [in=90, out=0] (47.center)
			 to [in=-30, out=-150] cycle;
		\draw [style=red fill] (46.center)
			 to [in=-90, out=0] (47.center)
			 to (45.center)
			 to [in=-180, out=-90] cycle;
		\draw [style=blue edge] (49.center) to (44.center);
		\draw [style=red edge] (43.center) to (49.center);
		\draw [style=red edge] (51.center) to (52);
		\draw [style=blue edge] (53.center) to (54);
		\draw [style=blue edge, in=90, out=-150] (54) to (55.center);
		\draw [style=blue edge, in=90, out=-30] (54) to (56.center);
		\draw [style=blue fill] (60.center)
			 to [in=180, out=90] (63.center)
			 to [in=90, out=0] (62.center)
			 to [in=-30, out=-150] cycle;
		\draw [style=red fill] (61.center)
			 to [in=-90, out=0] (62.center)
			 to (60.center)
			 to [in=-180, out=-90] cycle;
		\draw [style=blue edge] (64.center) to (59.center);
		\draw [style=red edge] (58.center) to (64.center);
		\draw [style=blue edge] (55.center) to (65.center);
	\end{pgfonlayer}
\end{tikzpicture}    
    \end{equation}
    We provide two concrete examples:
    \begin{itemize}
        \item The identity monad morphism $1_P \colon P \nto P$ yields a trivial weak distributive law $\eta^P \mu^P \colon [P,P]$ defined by $\mathcal{U} \mapsto \left\{ \bigcup \mathcal{U} \right\}$ for any subset of subsets $\mathcal{U}$.
        \item The support monad morphism $\supp \colon D \nto P$ yields a trivial weak distributive law $\eta^D\mu^P \circ P\supp \colon [D,P]$ defined by $U \mapsto 1 \cdot \bigcup_{p \in U} \supp_X(p)$ for any subset of distributions $U \in PDX$. 
    \end{itemize}
\end{example}

\begin{example}
    Given a monad $T$ and a natural transformation $\alpha \colon T \nto T$, one can view $\alpha$ as having type $[1,T]$ or $[T,1]$, where $1$ is the identity monad described in Example~\ref{ex:identity}. Table~\ref{table:1T} describes how axioms of distributive laws can then be satisfied by $\alpha$.
    \begin{table}[htp]
    \caption{Required axioms for the law $\alpha$}
    \begin{center}
    \begin{tabular}{ccc}
    \hline 
    axiom & $\alpha \colon [1,T]$ & $\alpha \colon [T,1]$ \\
    \hline
    \eqref{ax:etap} & $\alpha = 1$ & $\alpha \circ \eta^T = \eta^T$ \\
    \eqref{ax:etam} & $\alpha \circ \eta^T = \eta^T$ & $\alpha = 1$ \\ 
    \eqref{ax:mup} & $\alpha = \alpha \circ \alpha$ & $\alpha \circ \mu^T = \mu^T \circ \alpha\alpha$ \\ 
    \eqref{ax:mum} & $\alpha \circ \mu^T = \mu^T \circ \alpha\alpha$ & $\alpha = \alpha \circ \alpha$ \\
    \hline 
    \end{tabular} \label{table:1T}
    \end{center}
    \end{table}
    One can infer that $\alpha \colon [1,T]$ is a weak distributive law if and only if $\alpha$ is the identity, whereas $\alpha \colon [T,1]$ is a weak distributive law exactly when $\alpha$ is an idempotent monad morphism. For example, given any monad morphism $\gamma \colon T \nto 1$ (that is, a natural transformation which is pointwise Eilenberg-Moore), the corresponding trivial weak distributive law is $\eta^T \circ \gamma \colon [T,1]$ which is indeed an idempotent monad morphism. We provide two concrete examples:
    \begin{itemize}
        \item The map $\mathsf{proj}^2 \colon W \nto 1$ from the writer monad, defined by $\mathsf{proj}^2_X(m,x) \triangleq x$ for any $(m,x) \in WX$, is a monad morphism. Therefore $\eta^W \circ \mathsf{proj}^2 \colon [W,1]$, defined by $(m,x) \mapsto (e,x)$, is a weak distributive law.
        \item Given any letter $a \in A$, the map $\hat{a} \colon R \nto 1$ from the reader monad, defined by $\hat{a}_X(h) \triangleq h(a)$ for any $h \in RX$, is a monad morphism. Therefore $\eta^R \circ \hat{a} \colon [R,1]$, defined by $h \mapsto (b \mapsto h(a))$, is a weak distributive law.
    \end{itemize}
\end{example}

We provide two non-trivial examples of weak distributive laws, both of them arising from situations when no distributive law of their type exists~\cite{Klin18,Varacca06}.
\begin{example} \label{ex:law_pd}
    The natural transformation $\lambda \colon [P,D]$ defined for any distribution of subsets $p = \sum_i p_i \cdot U_i$ by 
    \begin{equation} 
        \lambda_X(p) = \left\{ \sum_i p_i q^i \mid q^i \in DX, \, \supp(q^i) \subseteq U_i\right\} \label{eq:pd_def}
    \end{equation}
    is the unique monotone weak distributive law of its type -- see~\cite{Goy20} 
    for a detailed study.
\end{example}

\begin{example} \label{ex:law_pp}
    The natural transformation $\lambda \colon [P,P]$ defined for any subset of subsets $\mathcal{U}$ by 
    \begin{equation} 
        \lambda_X(\mathcal{U}) = \left\{ V \subseteq \bigcup \mathcal{U} \mid \forall U \in \mathcal{U}, V \cap U \neq \emptyset \right\} \label{eq:pp_def}
    \end{equation}
    is the unique monotone weak distributive law of its type -- see~\cite{Garner20,Goy21}
    for more details.
\end{example}

One can see that the theory of weak distributive laws is very similar to Beck's theory of distributive laws~\cite{Beck69} \emph{modulo} the transformations $\pi$ and $\iota$. For example, if $\lambda \colon [S,T]$ is a distributive law, then the classical theory shows that $S\eta^T \colon S \nto \comp{S}{\lambda}{T}$ and $\eta^S T \colon T \nto \comp{S}{\lambda}{T}$ are monad morphisms. Similarly, if $\lambda \colon [S,T]$ is only weak, then $\pi \circ S\eta^T \colon S \nto \wcomp{S}{\lambda}{T}$ and $\pi \circ \eta^S T \colon T \nto \wcomp{S}{\lambda}{T}$ are monad morphisms. Whenever the base category is idempotent complete, as is assumed in this paper, the renowned correspondence distributive laws - extensions - liftings also generalises (see~\cite{Garner20}).

\begin{remark} \label{rem:other_weakenings}
    Deleting the~\eqref{ax:etam} axiom is one of the many ways of weakening the notion of distributive law. In recent years, this approach has been very fruitful for computer science applications~\cite{Goy20,Goy21,GoyTh,Bonchi21}. One may ask, symmetrically, what happens by deleting the~\eqref{ax:etap} axiom instead. This idea gives rise to the notion of \emph{coweak distributive law} developed by the author in~\cite{GoyTh}. Coweak distributive laws enjoy many symmetries with respect to weak distributive laws, but no non-trivial examples of them are known. Yet another way of weakening the notion of distributive law is the one of Street~\cite{Street09}, who suggests keeping a weakened version of both the~\eqref{ax:etap} and the~\eqref{ax:etam} axiom. The presentation of the present paper -- including the string diagrammatic approach -- is greatly inspired by the one of Street-weak distributive laws in~\cite{Street09}. A broader, $2$-categorical account of distributive law weakenings can be found in~\cite{Bohm10,Bohm11}.
\end{remark}

\section{Iterating Laws} \label{sec:iteration}

\subsection{Adapting Cheng's Theorem}

In this section, we are interested in what happens when combining more laws. For the sake of simplicity, we will only consider the case when three laws are in play. In the whole section, let $R$, $S$, $T$ be monads
\begin{equation}
    \begin{tikzpicture}
	\begin{pgfonlayer}{nodelayer}
		\node [style=none] (0) at (-2.75, 1) {};
		\node [style=red dot] (1) at (-2.75, 0) {};
		\node [style=none] (2) at (-3.75, -1) {};
		\node [style=none] (3) at (-1.75, -1) {};
		\node [style=none] (4) at (-6.75, 1) {};
		\node [style=red dot] (5) at (-6.75, 0) {};
		\node [style=none] (7) at (-5, 0) {$\mu^S$};
		\node [style=none] (27) at (-8.25, 0) {$\eta^S$};
		\node [style=none] (50) at (-10, 1) {};
		\node [style=none] (51) at (-10, -1) {};
		\node [style=none] (52) at (-11.25, 0) {$1_S$};
		\node [style=none] (53) at (9.25, 1) {};
		\node [style=blue dot] (54) at (9.25, 0) {};
		\node [style=none] (55) at (8.25, -1) {};
		\node [style=none] (56) at (10.25, -1) {};
		\node [style=none] (57) at (5.25, 1) {};
		\node [style=blue dot] (58) at (5.25, 0) {};
		\node [style=none] (59) at (7, 0) {$\mu^T$};
		\node [style=none] (60) at (3.75, 0) {$\eta^T$};
		\node [style=none] (61) at (2, 1) {};
		\node [style=none] (62) at (2, -1) {};
		\node [style=none] (63) at (0.75, 0) {$1_T$};
		\node [style=none] (71) at (-14.75, 1) {};
		\node [style=yellow dot] (72) at (-14.75, 0) {};
		\node [style=none] (73) at (-15.75, -1) {};
		\node [style=none] (74) at (-13.75, -1) {};
		\node [style=none] (75) at (-18.75, 1) {};
		\node [style=yellow dot] (76) at (-18.75, 0) {};
		\node [style=none] (77) at (-22, 1) {};
		\node [style=none] (78) at (-22, -1) {};
		\node [style=none] (79) at (-23.25, 0) {$1_R$};
		\node [style=none] (80) at (-20.25, 0) {$\eta^R$};
		\node [style=none] (81) at (-17, 0) {$\mu^R$};
		\node [style=none] (82) at (-22.5, 0) {$\triangleq$};
		\node [style=none] (83) at (-19.5, 0) {$\triangleq$};
		\node [style=none] (84) at (-16.25, 0) {$\triangleq$};
		\node [style=none] (85) at (-10.5, 0) {$\triangleq$};
		\node [style=none] (86) at (-7.5, 0) {$\triangleq$};
		\node [style=none] (87) at (-4.25, 0) {$\triangleq$};
		\node [style=none] (88) at (1.5, 0) {$\triangleq$};
		\node [style=none] (89) at (4.5, 0) {$\triangleq$};
		\node [style=none] (90) at (7.75, 0) {$\triangleq$};
		\node [style=none] (91) at (-22, -1.5) {$R$};
		\node [style=none] (92) at (-22, 1.5) {$R$};
		\node [style=none] (93) at (-18.75, 1.5) {$R$};
		\node [style=none] (94) at (-14.75, 1.5) {$R$};
		\node [style=none] (95) at (-15.75, -1.5) {$R$};
		\node [style=none] (96) at (-13.75, -1.5) {$R$};
		\node [style=none] (97) at (-10, 1.5) {$S$};
		\node [style=none] (98) at (-10, -1.5) {$S$};
		\node [style=none] (99) at (-6.75, 1.5) {$S$};
		\node [style=none] (100) at (-2.75, 1.5) {$S$};
		\node [style=none] (101) at (-3.75, -1.5) {$S$};
		\node [style=none] (102) at (-1.75, -1.5) {$S$};
		\node [style=none] (103) at (2, 1.5) {$T$};
		\node [style=none] (104) at (2, -1.5) {$T$};
		\node [style=none] (105) at (5.25, 1.5) {$T$};
		\node [style=none] (106) at (9.25, 1.5) {$T$};
		\node [style=none] (107) at (8.25, -1.5) {$T$};
		\node [style=none] (108) at (10.25, -1.5) {$T$};
	\end{pgfonlayer}
	\begin{pgfonlayer}{edgelayer}
		\draw [style=red edge] (4.center) to (5);
		\draw [style=red edge] (0.center) to (1);
		\draw [style=red edge, in=90, out=-150] (1) to (2.center);
		\draw [style=red edge, in=90, out=-30] (1) to (3.center);
		\draw [style=red edge] (50.center) to (51.center);
		\draw [style=blue edge] (57.center) to (58);
		\draw [style=blue edge] (53.center) to (54);
		\draw [style=blue edge, in=90, out=-150] (54) to (55.center);
		\draw [style=blue edge, in=90, out=-30] (54) to (56.center);
		\draw [style=blue edge] (61.center) to (62.center);
		\draw [style=yellow edge] (75.center) to (76);
		\draw [style=yellow edge] (71.center) to (72);
		\draw [style=yellow edge, in=90, out=-150] (72) to (73.center);
		\draw [style=yellow edge, in=90, out=-30] (72) to (74.center);
		\draw [style=yellow edge] (77.center) to (78.center);
	\end{pgfonlayer}
\end{tikzpicture}
\end{equation}
and $\lambda \colon [S,T]$, $\sigma \colon [R,S]$, $\tau \colon [R,T]$ be natural transformations such that the following equation holds.
\begin{equation} \label{eq:yb} \tag{YB}
    \sigma T \circ S\tau \circ \lambda R = R\lambda \circ \tau S \circ T\sigma 
\end{equation}    
Equation~\eqref{eq:yb}, known as the \emph{Yang-Baxter equation}, intuitively states that these three natural transformations are coherent with each other. Diagrammatically, it means that when performing the three crossings in a row, their order does not matter:
\begin{equation}
    \begin{tikzpicture}
	\begin{pgfonlayer}{nodelayer}
		\node [style=none] (0) at (-11, 1) {};
		\node [style=none] (1) at (-9, 1) {};
		\node [style=none] (2) at (-11, -1) {};
		\node [style=none] (3) at (-9, -1) {};
		\node [style=none] (4) at (-9, 1) {};
		\node [style=none] (5) at (-11, -1) {};
		\node [style=none] (6) at (-12.5, 0) {$\lambda$};
		\node [style=none] (7) at (3, 1) {};
		\node [style=none] (8) at (5, 1) {};
		\node [style=none] (9) at (3, -1) {};
		\node [style=none] (10) at (5, -1) {};
		\node [style=none] (11) at (5, 1) {};
		\node [style=none] (12) at (3, -1) {};
		\node [style=none] (13) at (-4, 1) {};
		\node [style=none] (14) at (-2, 1) {};
		\node [style=none] (15) at (-4, -1) {};
		\node [style=none] (16) at (-2, -1) {};
		\node [style=none] (17) at (-2, 1) {};
		\node [style=none] (18) at (-4, -1) {};
		\node [style=none] (19) at (1.5, 0) {$\tau$};
		\node [style=none] (20) at (-5.5, 0) {$\sigma$};
		\node [style=none] (21) at (9.5, 1) {};
		\node [style=none] (22) at (9.5, -1) {};
		\node [style=none] (23) at (9, 1) {};
		\node [style=none] (24) at (11, -1) {};
		\node [style=none] (25) at (9, -1) {};
		\node [style=none] (26) at (11, 1) {};
		\node [style=none] (27) at (12, 0) {$=$};
		\node [style=none] (28) at (13, 1) {};
		\node [style=none] (29) at (15, -1) {};
		\node [style=none] (30) at (14.5, 1) {};
		\node [style=none] (31) at (14.5, -1) {};
		\node [style=none] (32) at (13, -1) {};
		\node [style=none] (33) at (15, 1) {};
		\node [style=none] (34) at (12, -2.5) {Yang-Baxter};
		\node [style=none] (35) at (-4, 1.5) {$R$};
		\node [style=none] (36) at (3, 1.5) {$R$};
		\node [style=none] (37) at (9, 1.5) {$R$};
		\node [style=none] (38) at (13, 1.5) {$R$};
		\node [style=none] (39) at (-2, -1.5) {$R$};
		\node [style=none] (40) at (5, -1.5) {$R$};
		\node [style=none] (41) at (11, -1.5) {$R$};
		\node [style=none] (42) at (15, -1.5) {$R$};
		\node [style=none] (43) at (-11, 1.5) {$S$};
		\node [style=none] (44) at (-2, 1.5) {$S$};
		\node [style=none] (45) at (9.5, 1.5) {$S$};
		\node [style=none] (46) at (14.5, 1.5) {$S$};
		\node [style=none] (47) at (-4, -1.5) {$S$};
		\node [style=none] (48) at (-9, -1.5) {$S$};
		\node [style=none] (49) at (9.5, -1.5) {$S$};
		\node [style=none] (50) at (14.5, -1.5) {$S$};
		\node [style=none] (51) at (-9, 1.5) {$T$};
		\node [style=none] (52) at (5, 1.5) {$T$};
		\node [style=none] (53) at (11, 1.5) {$T$};
		\node [style=none] (54) at (15, 1.5) {$T$};
		\node [style=none] (55) at (-11, -1.5) {$T$};
		\node [style=none] (56) at (3, -1.5) {$T$};
		\node [style=none] (57) at (13, -1.5) {$T$};
		\node [style=none] (58) at (9, -1.5) {$T$};
		\node [style=none] (59) at (-11.75, 0) {$\triangleq$};
		\node [style=none] (60) at (-4.75, 0) {$\triangleq$};
		\node [style=none] (61) at (2.25, 0) {$\triangleq$};
	\end{pgfonlayer}
	\begin{pgfonlayer}{edgelayer}
		\draw [style=red edge, in=90, out=-90] (0.center) to (3.center);
		\draw [style=thick white, in=-90, out=90] (2.center) to (1.center);
		\draw [style=blue edge, in=-90, out=90] (5.center) to (4.center);
		\draw [style=yellow edge, in=90, out=-90] (7.center) to (10.center);
		\draw [style=thick white, in=-90, out=90] (9.center) to (8.center);
		\draw [style=blue edge, in=-90, out=90] (12.center) to (11.center);
		\draw [style=yellow edge, in=90, out=-90] (13.center) to (16.center);
		\draw [style=thick white, in=-90, out=90] (15.center) to (14.center);
		\draw [style=red edge, in=-90, out=90] (18.center) to (17.center);
		\draw [style=yellow edge] (23.center) to (24.center);
		\draw [style=thick white] (21.center) to (22.center);
		\draw [style=red edge] (21.center) to (22.center);
		\draw [style=thick white] (25.center) to (26.center);
		\draw [style=blue edge] (25.center) to (26.center);
		\draw [style=yellow edge] (28.center) to (29.center);
		\draw [style=thick white] (30.center) to (31.center);
		\draw [style=red edge] (30.center) to (31.center);
		\draw [style=thick white] (32.center) to (33.center);
		\draw [style=blue edge] (32.center) to (33.center);
	\end{pgfonlayer}
\end{tikzpicture}
\end{equation}

The result of Cheng~\cite[Theorem 1.6]{Cheng11a} can be rephrased as follows in the case of three monads. 

\begin{theorem}[Cheng] \label{theo:cheng}
    Let $R$, $S$, $T$ be monads and $\lambda \colon [S,T]$, $\sigma \colon [R,S]$, $\tau \colon [R,T]$ be distributive laws satisfying the Yang-Baxter equation.
    Then
    \begin{itemize}
        \item $R\lambda \circ \tau S \colon [\comp{R}{}{S},T]$ is a distributive law,
        \item $\sigma T \circ S \tau \colon [R,\comp{S}{}{T}]$ is a distributive law,
        \item these two laws generate the same monad, \textit{i.e.}, $\comp{(\comp{R}{}{S})}{}{T} = \comp{R}{}{(\comp{S}{}{T})}$.
    \end{itemize}
\end{theorem}

The result can be generalised to weak distributive laws. Let us consider separately what happens for the two possible composite laws. From now onwards, the natural transformations $\kappa$, $\pi$ and $\iota$ from Section~\ref{sec:laws} will mention explicitly which law they relate to -- \textit{e.g.}, given a weak distributive law $\lambda \colon [S,T]$, they are denoted by $\kappa^\lambda$, $\pi^\lambda$ and $\iota^\lambda$.

\subsubsection{First Composite Law} \label{subsec:first_comp}
Assume $\sigma \colon [R,S]$ is a weak distributive law. We will derive which axioms on $\lambda$ and $\tau$ are sufficient for
\begin{equation} 
\phi \triangleq
\begin{tikzcd}
  T(\wcomp{R}{}{S}) \arrow[r, "T\iota^\sigma", Rightarrow] & TRS \arrow[r, "\tau S", Rightarrow] & RTS \arrow[r, "R\lambda", Rightarrow] & RST \arrow[r, "\pi^\sigma T", Rightarrow] & (\wcomp{R}{}{S})T
  \end{tikzcd}
\end{equation}
to be a (weak) distributive law. The natural $\phi \colon [\wcomp{R}{}{S},T]$ can be depicted as follows, where $H$ is short notation for the functor $\wcomp{R}{}{S}$:
\begin{equation} 
    \begin{tikzpicture}
	\begin{pgfonlayer}{nodelayer}
		\node [style=none] (40) at (-0.5, 0) {};
		\node [style=none] (41) at (0, 0.5) {};
		\node [style=none] (42) at (0.5, 0) {};
		\node [style=none] (43) at (-0.025, 0.475) {};
		\node [style=none] (44) at (-0.025, 1) {};
		\node [style=none] (45) at (0.025, 0.475) {};
		\node [style=none] (46) at (0.025, 1) {};
		\node [style=none] (47) at (-0.5, 0) {};
		\node [style=none] (48) at (-0.025, -0.5) {};
		\node [style=none] (49) at (0.5, 0) {};
		\node [style=none] (50) at (-0.025, -1) {};
		\node [style=none] (51) at (-0.025, -0.475) {};
		\node [style=none] (52) at (0.025, -1) {};
		\node [style=none] (53) at (0.025, -0.475) {};
		\node [style=none] (54) at (-1, -1) {};
		\node [style=none] (55) at (1, 1) {};
		\node [style=none] (56) at (-1, -1) {};
		\node [style=none] (57) at (1, 1) {};
		\node [style=none] (58) at (-2.5, 0) {$\phi$};
		\node [style=none] (59) at (0, -1.5) {$H$};
		\node [style=none] (60) at (0, 1.5) {$H$};
		\node [style=none] (61) at (-1.75, 0) {$\triangleq$};
		\node [style=none] (62) at (1, 1.5) {$T$};
		\node [style=none] (63) at (-1, -1.5) {$T$};
	\end{pgfonlayer}
	\begin{pgfonlayer}{edgelayer}
		\draw [style=yellow edge, in=180, out=90] (40.center) to (41.center);
		\draw [style=red edge, in=0, out=90] (42.center) to (41.center);
		\draw [style=yellow edge] (43.center) to (44.center);
		\draw [style=red edge] (45.center) to (46.center);
		\draw [style=yellow edge, in=630, out=-180] (48.center) to (47.center);
		\draw [style=red edge, in=-90, out=0] (48.center) to (49.center);
		\draw [style=yellow edge] (50.center) to (51.center);
		\draw [style=red edge] (52.center) to (53.center);
		\draw [style=thick white, in=-90, out=90, looseness=1.50] (54.center) to (55.center);
		\draw [style=blue edge, in=-90, out=90, looseness=1.50] (56.center) to (57.center);
	\end{pgfonlayer}
\end{tikzpicture}
\end{equation}
The~\eqref{ax:etap} axiom for $\phi$ can be derived from equation~\eqref{eq:kappa_eta} which states compatibility between units and the idempotent, the~\eqref{ax:etap} axiom for $\tau$, and the~\eqref{ax:etap} axiom for $\lambda$: 
\begin{equation}
    \begin{tikzpicture}
	\begin{pgfonlayer}{nodelayer}
		\node [style=none] (40) at (-9, 0.5) {};
		\node [style=none] (41) at (-8.5, 1) {};
		\node [style=none] (42) at (-8, 0.5) {};
		\node [style=none] (43) at (-8.525, 0.975) {};
		\node [style=none] (44) at (-8.525, 1.25) {};
		\node [style=none] (45) at (-8.475, 0.975) {};
		\node [style=none] (46) at (-8.475, 1.25) {};
		\node [style=none] (47) at (-9, 0.5) {};
		\node [style=none] (48) at (-8.525, 0) {};
		\node [style=none] (49) at (-8, 0.5) {};
		\node [style=none] (50) at (-8.525, 0) {};
		\node [style=none] (51) at (-8.525, 0.025) {};
		\node [style=none] (52) at (-8.475, 0) {};
		\node [style=none] (53) at (-8.475, 0.025) {};
		\node [style=none] (54) at (-9.5, -0.25) {};
		\node [style=none] (55) at (-7.5, 1.25) {};
		\node [style=none] (56) at (-9.5, -0.25) {};
		\node [style=none] (57) at (-7.5, 1.25) {};
		\node [style=none] (58) at (-8.75, -0.5) {};
		\node [style=yellow dot] (59) at (-8.75, -0.7) {};
		\node [style=none] (60) at (-8.25, -0.5) {};
		\node [style=red dot] (61) at (-8.25, -0.7) {};
		\node [style=none] (62) at (-8.75, -0.5) {};
		\node [style=none] (63) at (-8.5, -0.25) {};
		\node [style=none] (64) at (-8.25, -0.5) {};
		\node [style=none] (65) at (-8.525, -0.275) {};
		\node [style=none] (66) at (-8.525, 0) {};
		\node [style=none] (67) at (-8.475, -0.275) {};
		\node [style=none] (68) at (-8.475, 0) {};
		\node [style=none] (69) at (-9.5, -1) {};
		\node [style=none] (70) at (-6, 0.125) {$=$};
		\node [style=none] (71) at (-3.5, 0.5) {};
		\node [style=none] (72) at (-3, 1) {};
		\node [style=none] (73) at (-2.5, 0.5) {};
		\node [style=none] (74) at (-3.025, 0.975) {};
		\node [style=none] (75) at (-3.025, 1.25) {};
		\node [style=none] (76) at (-2.975, 0.975) {};
		\node [style=none] (77) at (-2.975, 1.25) {};
		\node [style=none] (78) at (-3.5, 0.5) {};
		\node [style=none] (80) at (-2.5, 0.5) {};
		\node [style=yellow dot] (90) at (-3.5, -0.7) {};
		\node [style=red dot] (92) at (-2.5, -0.7) {};
		\node [style=none] (104) at (-4, -0.25) {};
		\node [style=none] (105) at (-2, 1.25) {};
		\node [style=none] (106) at (-4, -1) {};
		\node [style=none] (107) at (-4, -0.25) {};
		\node [style=none] (108) at (-2, 1.25) {};
		\node [style=none] (109) at (-4, -1) {};
		\node [style=none] (110) at (-0.5, 0.125) {$=$};
		\node [style=none] (111) at (2, 0.5) {};
		\node [style=none] (112) at (2.5, 1) {};
		\node [style=none] (113) at (3, 0.5) {};
		\node [style=none] (114) at (2.475, 0.975) {};
		\node [style=none] (115) at (2.475, 1.25) {};
		\node [style=none] (116) at (2.525, 0.975) {};
		\node [style=none] (117) at (2.525, 1.25) {};
		\node [style=none] (119) at (3, 0) {};
		\node [style=yellow dot] (120) at (2, 0.3) {};
		\node [style=red dot] (121) at (3, -0.7) {};
		\node [style=none] (128) at (3.5, 1.25) {};
		\node [style=none] (129) at (1.5, -1) {};
		\node [style=none] (130) at (1.5, -0.75) {};
		\node [style=none] (131) at (3.5, 0.75) {};
		\node [style=none] (132) at (1.5, -1) {};
		\node [style=none] (133) at (3.5, 1.25) {};
		\node [style=none] (134) at (1.5, -1) {};
		\node [style=none] (135) at (1.5, -0.75) {};
		\node [style=none] (136) at (3.5, 0.75) {};
		\node [style=none] (137) at (1.5, -1) {};
		\node [style=none] (138) at (5, 0.125) {$=$};
		\node [style=none] (139) at (7.5, 0.5) {};
		\node [style=none] (140) at (8, 1) {};
		\node [style=none] (141) at (8.5, 0.5) {};
		\node [style=none] (142) at (7.975, 0.975) {};
		\node [style=none] (143) at (7.975, 1.25) {};
		\node [style=none] (144) at (8.025, 0.975) {};
		\node [style=none] (145) at (8.025, 1.25) {};
		\node [style=none] (147) at (8.5, 0.5) {};
		\node [style=yellow dot] (148) at (7.5, 0.3) {};
		\node [style=red dot] (149) at (8.5, 0.3) {};
		\node [style=none] (150) at (9, 1.25) {};
		\node [style=none] (151) at (7, -1) {};
		\node [style=none] (154) at (7, -1) {};
		\node [style=none] (155) at (9, 1.25) {};
		\node [style=none] (156) at (7, -1) {};
		\node [style=none] (157) at (7, -1) {};
		\node [style=none] (158) at (9, 0.5) {};
		\node [style=none] (159) at (7, -1) {};
		\node [style=none] (160) at (-6, -0.5) {\eqref{eq:kappa_eta}};
		\node [style=none] (161) at (-0.5, -0.5) {$\tau$~\eqref{ax:etap}};
		\node [style=none] (162) at (5, -0.5) {$\lambda$~\eqref{ax:etap}};
		\node [style=none] (163) at (-8.5, 1.75) {$H$};
		\node [style=none] (164) at (-3, 1.75) {$H$};
		\node [style=none] (165) at (2.5, 1.75) {$H$};
		\node [style=none] (166) at (8, 1.75) {$H$};
		\node [style=none] (167) at (-7.5, 1.75) {$T$};
		\node [style=none] (168) at (-2, 1.75) {$T$};
		\node [style=none] (169) at (3.5, 1.75) {$T$};
		\node [style=none] (170) at (9, 1.75) {$T$};
		\node [style=none] (171) at (-9.5, -1.5) {$T$};
		\node [style=none] (172) at (-4, -1.5) {$T$};
		\node [style=none] (173) at (1.5, -1.5) {$T$};
		\node [style=none] (174) at (7, -1.5) {$T$};
	\end{pgfonlayer}
	\begin{pgfonlayer}{edgelayer}
		\draw [style=yellow edge, in=180, out=90] (40.center) to (41.center);
		\draw [style=red edge, in=0, out=90] (42.center) to (41.center);
		\draw [style=yellow edge] (43.center) to (44.center);
		\draw [style=red edge] (45.center) to (46.center);
		\draw [style=yellow edge, in=630, out=-180] (48.center) to (47.center);
		\draw [style=red edge, in=-90, out=0] (48.center) to (49.center);
		\draw [style=yellow edge] (50.center) to (51.center);
		\draw [style=red edge] (52.center) to (53.center);
		\draw [style=thick white, in=-90, out=90, looseness=1.50] (54.center) to (55.center);
		\draw [style=blue edge, in=-90, out=90, looseness=1.50] (56.center) to (57.center);
		\draw [style=yellow edge] (58.center) to (59);
		\draw [style=red edge] (60.center) to (61);
		\draw [style=yellow edge, in=180, out=90] (62.center) to (63.center);
		\draw [style=red edge, in=0, out=90] (64.center) to (63.center);
		\draw [style=yellow edge] (65.center) to (66.center);
		\draw [style=red edge] (67.center) to (68.center);
		\draw [style=blue edge] (56.center) to (69.center);
		\draw [style=yellow edge, in=180, out=90] (71.center) to (72.center);
		\draw [style=red edge, in=0, out=90] (73.center) to (72.center);
		\draw [style=yellow edge] (74.center) to (75.center);
		\draw [style=red edge] (76.center) to (77.center);
		\draw [style=yellow edge] (78.center) to (90);
		\draw [style=red edge] (80.center) to (92);
		\draw [style=thick white, in=-90, out=90, looseness=1.50] (104.center) to (105.center);
		\draw [style=thick white] (104.center) to (106.center);
		\draw [style=blue edge, in=-90, out=90, looseness=1.50] (107.center) to (108.center);
		\draw [style=blue edge] (107.center) to (109.center);
		\draw [style=yellow edge, in=180, out=90] (111.center) to (112.center);
		\draw [style=red edge, in=0, out=90] (113.center) to (112.center);
		\draw [style=yellow edge] (114.center) to (115.center);
		\draw [style=red edge] (116.center) to (117.center);
		\draw [style=red edge] (119.center) to (121);
		\draw [style=red edge] (113.center) to (119.center);
		\draw [style=thick white, in=-90, out=90, looseness=1.50] (130.center) to (131.center);
		\draw [style=thick white] (130.center) to (132.center);
		\draw [style=thick white] (131.center) to (128.center);
		\draw [style=blue edge, in=-90, out=90, looseness=1.50] (135.center) to (136.center);
		\draw [style=blue edge] (135.center) to (137.center);
		\draw [style=blue edge] (136.center) to (133.center);
		\draw [style=yellow edge, in=180, out=90] (139.center) to (140.center);
		\draw [style=red edge, in=0, out=90] (141.center) to (140.center);
		\draw [style=yellow edge] (142.center) to (143.center);
		\draw [style=red edge] (144.center) to (145.center);
		\draw [style=red edge] (147.center) to (149);
		\draw [style=red edge] (141.center) to (147.center);
		\draw [style=blue edge, in=-90, out=90, looseness=1.50] (157.center) to (158.center);
		\draw [style=blue edge] (157.center) to (159.center);
		\draw [style=blue edge] (158.center) to (155.center);
		\draw [style=yellow edge] (111.center) to (120);
		\draw [style=yellow edge] (139.center) to (148);
	\end{pgfonlayer}
\end{tikzpicture}
\end{equation}
    The~\eqref{ax:etam} axiom for $\phi$ can be derived from the~\eqref{ax:etam} axiom for $\tau$, the~\eqref{ax:etam} axiom for $\lambda$, and the retract equation $\pi \circ \iota = 1_{\wcomp{R}{}{S}}$:
\begin{equation}
    \begin{tikzpicture}
	\begin{pgfonlayer}{nodelayer}
		\node [style=none] (40) at (-9, 0.5) {};
		\node [style=none] (41) at (-8.5, 1) {};
		\node [style=none] (42) at (-8, 0.5) {};
		\node [style=none] (43) at (-8.525, 0.975) {};
		\node [style=none] (44) at (-8.525, 1.25) {};
		\node [style=none] (45) at (-8.475, 0.975) {};
		\node [style=none] (46) at (-8.475, 1.25) {};
		\node [style=none] (47) at (-9, 0.5) {};
		\node [style=none] (48) at (-8.525, 0) {};
		\node [style=none] (49) at (-8, 0.5) {};
		\node [style=none] (50) at (-8.525, 0) {};
		\node [style=none] (51) at (-8.525, 0.025) {};
		\node [style=none] (52) at (-8.475, 0) {};
		\node [style=none] (53) at (-8.475, 0.025) {};
		\node [style=none] (54) at (-9.5, -0.25) {};
		\node [style=none] (55) at (-7.5, 1.25) {};
		\node [style=none] (56) at (-9.5, -0.25) {};
		\node [style=none] (57) at (-7.5, 1.25) {};
		\node [style=none] (63) at (-8.5, -1) {};
		\node [style=none] (65) at (-8.525, -1.025) {};
		\node [style=none] (66) at (-8.525, 0) {};
		\node [style=none] (67) at (-8.475, -1.025) {};
		\node [style=none] (68) at (-8.475, 0) {};
		\node [style=blue dot] (69) at (-9.5, -0.5) {};
		\node [style=none] (70) at (-6, 0.125) {$=$};
		\node [style=none] (110) at (-0.75, 0.125) {$=$};
		\node [style=none] (138) at (4.5, 0.125) {$=$};
		\node [style=none] (160) at (-6, -0.5) {$\tau$~\eqref{ax:etam}};
		\node [style=none] (161) at (-0.75, -0.5) {$\lambda$~\eqref{ax:etam}};
		\node [style=none] (162) at (4.5, -0.5) {retract};
		\node [style=none] (163) at (-4, 0.5) {};
		\node [style=none] (164) at (-3.5, 1) {};
		\node [style=none] (165) at (-3, 0.5) {};
		\node [style=none] (166) at (-3.525, 0.975) {};
		\node [style=none] (167) at (-3.525, 1.25) {};
		\node [style=none] (168) at (-3.475, 0.975) {};
		\node [style=none] (169) at (-3.475, 1.25) {};
		\node [style=none] (170) at (-4, 0.5) {};
		\node [style=none] (171) at (-3.525, 0) {};
		\node [style=none] (172) at (-3, 0.5) {};
		\node [style=none] (173) at (-3.525, 0) {};
		\node [style=none] (174) at (-3.525, 0.025) {};
		\node [style=none] (175) at (-3.475, 0) {};
		\node [style=none] (176) at (-3.475, 0.025) {};
		\node [style=none] (181) at (-3.5, -1) {};
		\node [style=none] (182) at (-3.525, -1.025) {};
		\node [style=none] (183) at (-3.525, 0) {};
		\node [style=none] (184) at (-3.475, -1.025) {};
		\node [style=none] (185) at (-3.475, 0) {};
		\node [style=none] (186) at (-3.5, 0.5) {};
		\node [style=none] (187) at (-2.5, 1.25) {};
		\node [style=none] (188) at (-3.5, 0.5) {};
		\node [style=none] (189) at (-2.5, 1.25) {};
		\node [style=blue dot] (190) at (-3.5, 0.5) {};
		\node [style=none] (191) at (1, 0.5) {};
		\node [style=none] (192) at (1.5, 1) {};
		\node [style=none] (193) at (2, 0.5) {};
		\node [style=none] (194) at (1.475, 0.975) {};
		\node [style=none] (195) at (1.475, 1.25) {};
		\node [style=none] (196) at (1.525, 0.975) {};
		\node [style=none] (197) at (1.525, 1.25) {};
		\node [style=none] (198) at (1, 0.5) {};
		\node [style=none] (199) at (1.475, 0) {};
		\node [style=none] (200) at (2, 0.5) {};
		\node [style=none] (201) at (1.475, 0) {};
		\node [style=none] (202) at (1.475, 0.025) {};
		\node [style=none] (203) at (1.525, 0) {};
		\node [style=none] (204) at (1.525, 0.025) {};
		\node [style=none] (205) at (1.5, -1) {};
		\node [style=none] (206) at (1.475, -1.025) {};
		\node [style=none] (207) at (1.475, 0) {};
		\node [style=none] (208) at (1.525, -1.025) {};
		\node [style=none] (209) at (1.525, 0) {};
		\node [style=blue dot] (210) at (2.5, 0.5) {};
		\node [style=none] (211) at (2.5, 1.25) {};
		\node [style=none] (220) at (6.475, 1.25) {};
		\node [style=none] (222) at (6.475, 1.25) {};
		\node [style=none] (223) at (6.475, 1.275) {};
		\node [style=none] (224) at (6.525, 1.25) {};
		\node [style=none] (225) at (6.525, 1.275) {};
		\node [style=none] (226) at (6.5, -1) {};
		\node [style=none] (227) at (6.475, -1.025) {};
		\node [style=none] (228) at (6.475, 1.25) {};
		\node [style=none] (229) at (6.525, -1.025) {};
		\node [style=none] (230) at (6.525, 1.25) {};
		\node [style=blue dot] (231) at (7.25, 0.5) {};
		\node [style=none] (232) at (7.25, 1.25) {};
		\node [style=none] (233) at (-7.5, 1.75) {$T$};
		\node [style=none] (234) at (-2.5, 1.75) {$T$};
		\node [style=none] (235) at (2.5, 1.75) {$T$};
		\node [style=none] (236) at (7.25, 1.75) {$T$};
		\node [style=none] (237) at (-8.5, 1.75) {$H$};
		\node [style=none] (238) at (-3.5, 1.75) {$H$};
		\node [style=none] (239) at (1.5, 1.75) {$H$};
		\node [style=none] (240) at (6.5, 1.75) {$H$};
		\node [style=none] (241) at (6.5, -1.5) {$H$};
		\node [style=none] (242) at (1.5, -1.5) {$H$};
		\node [style=none] (243) at (-3.5, -1.5) {$H$};
		\node [style=none] (244) at (-8.5, -1.5) {$H$};
	\end{pgfonlayer}
	\begin{pgfonlayer}{edgelayer}
		\draw [style=yellow edge, in=180, out=90] (40.center) to (41.center);
		\draw [style=red edge, in=0, out=90] (42.center) to (41.center);
		\draw [style=yellow edge] (43.center) to (44.center);
		\draw [style=red edge] (45.center) to (46.center);
		\draw [style=yellow edge, in=630, out=-180] (48.center) to (47.center);
		\draw [style=red edge, in=-90, out=0] (48.center) to (49.center);
		\draw [style=yellow edge] (50.center) to (51.center);
		\draw [style=red edge] (52.center) to (53.center);
		\draw [style=thick white, in=-90, out=90, looseness=1.50] (54.center) to (55.center);
		\draw [style=blue edge, in=-90, out=90, looseness=1.50] (56.center) to (57.center);
		\draw [style=yellow edge] (65.center) to (66.center);
		\draw [style=red edge] (67.center) to (68.center);
		\draw [style=blue edge] (56.center) to (69);
		\draw [style=yellow edge, in=180, out=90] (163.center) to (164.center);
		\draw [style=red edge, in=0, out=90] (165.center) to (164.center);
		\draw [style=yellow edge] (166.center) to (167.center);
		\draw [style=red edge] (168.center) to (169.center);
		\draw [style=yellow edge, in=630, out=-180] (171.center) to (170.center);
		\draw [style=red edge, in=-90, out=0] (171.center) to (172.center);
		\draw [style=yellow edge] (173.center) to (174.center);
		\draw [style=red edge] (175.center) to (176.center);
		\draw [style=yellow edge] (182.center) to (183.center);
		\draw [style=red edge] (184.center) to (185.center);
		\draw [style=thick white, in=-90, out=90] (186.center) to (187.center);
		\draw [style=blue edge, in=-90, out=90] (188.center) to (189.center);
		\draw [style=yellow edge, in=180, out=90] (191.center) to (192.center);
		\draw [style=red edge, in=0, out=90] (193.center) to (192.center);
		\draw [style=yellow edge] (194.center) to (195.center);
		\draw [style=red edge] (196.center) to (197.center);
		\draw [style=yellow edge, in=630, out=-180] (199.center) to (198.center);
		\draw [style=red edge, in=-90, out=0] (199.center) to (200.center);
		\draw [style=yellow edge] (201.center) to (202.center);
		\draw [style=red edge] (203.center) to (204.center);
		\draw [style=yellow edge] (206.center) to (207.center);
		\draw [style=red edge] (208.center) to (209.center);
		\draw [style=blue edge] (211.center) to (210);
		\draw [style=yellow edge] (222.center) to (223.center);
		\draw [style=red edge] (224.center) to (225.center);
		\draw [style=yellow edge] (227.center) to (228.center);
		\draw [style=red edge] (229.center) to (230.center);
		\draw [style=blue edge] (232.center) to (231);
	\end{pgfonlayer}
\end{tikzpicture}
\end{equation}
    The~\eqref{ax:mup} axiom for $\phi$ can be derived from equation~\eqref{eq:kappa_mu} which states compatibility between multiplications and the idempotent, the retract equation $\pi \circ \iota = 1_{\wcomp{R}{}{S}}$, the~\eqref{ax:mup} axiom for both $\tau$ and $\lambda$, the Yang-Baxter equation, the retract equation again, and equation~\eqref{eq:kappa_mu} again:
\begin{equation} \label{eq:cheng_returns}
    \input{2_first_mup.tikz}
\end{equation}
    To prove the last axiom, we need a technical lemma.
    \begin{lemma} \label{lem:first_comp}
        If $\lambda$'s axioms~\eqref{ax:etap} and~\eqref{ax:mup} hold, and the Yang-Baxter equation holds, then the idempotent $\kappa^\sigma$ commutes with the natural $R\lambda \circ \tau S$ in the sense that 
        \begin{equation} 
            \kappa^\sigma T \circ R\lambda \circ \tau S = R\lambda \circ \tau S \circ T \kappa^\sigma 
        \end{equation}
    \end{lemma}
    \begin{proof}
    Easy to see using string diagrams:
        \begin{equation} 
            \input{2_lemma_asymmetry.tikz}
        \end{equation}
    \end{proof}
    The~\eqref{ax:mum} axiom for $\phi$ can then be derived from the retract equation $\pi \circ \iota = 1_{\wcomp{R}{}{S}}$, the~\eqref{ax:mum} axiom for both $\tau$ and $\lambda$, and the lemma relying on $\lambda$'s axioms~\eqref{ax:etap} and~\eqref{ax:mup}.
\begin{equation}
    \input{2_first_mum.tikz}
\end{equation}

Table~\ref{table:first_comp} sums up which axioms of $\tau$ and $\lambda$ are sufficient to make axioms hold for $\phi$.
\begin{table}[htp]
  \caption{Required axioms for the composite law $\phi$}
\begin{center}
\begin{tabular}{llll}
  \hline
  $\phi$ & requires & $\tau$ & $\lambda$ \\
  \hline 
  \eqref{ax:etap} & & \eqref{ax:etap} & \eqref{ax:etap} \\ 
  \eqref{ax:etam} & & \eqref{ax:etam} & \eqref{ax:etam} \\ 
  \eqref{ax:mup} & & \eqref{ax:mup} & \eqref{ax:mup} \\ 
  \eqref{ax:mum} & & \eqref{ax:mum} & \eqref{ax:mum}, \eqref{ax:etap}, \eqref{ax:mup} \\
  \hline
\end{tabular} \label{table:first_comp}
\end{center}
\end{table}

Using the information contained in Table~\ref{table:first_comp}, we get the following result.

\begin{theorem} \label{theo:first_comp}
    Let $R$, $S$, $T$ be monads and $\lambda \colon [S,T]$, $\sigma \colon [R,S]$, $\tau \colon [R,T]$ be weak distributive laws satisfying the Yang-Baxter equation. Then the composite $\phi \triangleq \pi^\sigma T \circ R\lambda \circ \tau S \circ T\iota^\sigma \colon [\wcomp{R}{}{S},T]$ is a weak distributive law. If, additionally, $\tau$ and $\lambda$ are distributive laws, then $\phi$ is a distributive law.
\end{theorem}

\subsubsection{Second Composite Law} \label{subsec:second_comp}
One could also compose the laws the other way around. Assume $\lambda \colon [S,T]$ is a weak distributive law. Consider the natural transformation
\begin{equation}
  \psi \triangleq
\begin{tikzcd}
  (\wcomp{S}{}{T})R \arrow[r, "\iota^\lambda R", Rightarrow] & STR \arrow[r, "S\tau", Rightarrow] & SRT \arrow[r, "\sigma T", Rightarrow] & RST \arrow[r, "R\pi^\lambda", Rightarrow] & R(\wcomp{S}{}{T})
  \end{tikzcd} 
\end{equation}
The natural $\psi \colon [R,\wcomp{S}{}{T}]$ can be depicted as follows, where $K$ is short notation for the functor $\wcomp{S}{}{T}$:
\begin{equation}
\begin{tikzpicture}
	\begin{pgfonlayer}{nodelayer}
		\node [style=none] (28) at (1, -1) {};
		\node [style=none] (29) at (-1, 1) {};
		\node [style=none] (30) at (-0.5, 0) {};
		\node [style=none] (31) at (0, 0.5) {};
		\node [style=none] (32) at (0.5, 0) {};
		\node [style=none] (33) at (-0.025, 0.475) {};
		\node [style=none] (34) at (-0.025, 1) {};
		\node [style=none] (35) at (0.025, 0.475) {};
		\node [style=none] (36) at (0.025, 1) {};
		\node [style=none] (37) at (-0.5, 0) {};
		\node [style=none] (39) at (0.5, 0) {};
		\node [style=none] (44) at (-0.5, 0) {};
		\node [style=none] (45) at (0, 0.5) {};
		\node [style=none] (46) at (0.5, 0) {};
		\node [style=none] (47) at (-0.025, 0.475) {};
		\node [style=none] (48) at (0.025, 0.475) {};
		\node [style=none] (49) at (-0.5, 0) {};
		\node [style=none] (50) at (-0.025, -0.5) {};
		\node [style=none] (51) at (0.5, 0) {};
		\node [style=none] (54) at (-2.5, 0) {$\psi$};
		\node [style=none] (57) at (0.025, -1) {};
		\node [style=none] (58) at (0.025, -0.475) {};
		\node [style=none] (59) at (-0.025, -0.5) {};
		\node [style=none] (60) at (0.5, 0) {};
		\node [style=none] (61) at (-0.025, -1) {};
		\node [style=none] (62) at (-0.025, -0.475) {};
		\node [style=none] (63) at (-1.75, 0) {$\triangleq$};
		\node [style=none] (64) at (-1, 1.5) {$R$};
		\node [style=none] (65) at (1, -1.5) {$R$};
		\node [style=none] (66) at (0, 1.5) {$K$};
		\node [style=none] (67) at (0, -1.5) {$K$};
	\end{pgfonlayer}
	\begin{pgfonlayer}{edgelayer}
		\draw [style=yellow edge, in=-90, out=90, looseness=1.50] (28.center) to (29.center);
		\draw [style=thick white, in=180, out=90] (30.center) to (31.center);
		\draw [style=thick white, in=0, out=90] (32.center) to (31.center);
		\draw [style=red edge] (33.center) to (34.center);
		\draw [style=blue edge] (35.center) to (36.center);
		\draw [style=red edge, in=180, out=90] (44.center) to (45.center);
		\draw [style=blue edge, in=0, out=90] (46.center) to (45.center);
		\draw [style=red edge, in=630, out=-180] (50.center) to (49.center);
		\draw [style=thick white, in=-90, out=0] (50.center) to (51.center);
		\draw [style=blue edge] (57.center) to (58.center);
		\draw [style=blue edge, in=-90, out=0] (59.center) to (60.center);
		\draw [style=red edge] (61.center) to (62.center);
	\end{pgfonlayer}
\end{tikzpicture}
\end{equation}

Again, assuming some axioms of distributive laws for $\sigma$ and $\tau$, one can infer axioms for $\psi$. Using string diagrams, one can easily see that the proofs are symmetrical from the ones in Section~\ref{subsec:first_comp}. As an instance, recall that Lemma~\ref{lem:first_comp} was required to prove the~\eqref{ax:mum} axiom for $\phi$. Now, the following symmetrical result is required to prove the~\eqref{ax:mup} axiom for $\psi$:

\begin{lemma} \label{lem:second_comp}
If $\tau$'s axioms~\eqref{ax:etam} and~\eqref{ax:mum} hold, and the Yang-Baxter equation holds, then the idempotent $\kappa^\lambda$ commutes with the natural $\sigma T \circ S\tau$ in the sense that
\begin{equation}
  R\kappa^\lambda \circ \sigma T\circ S\tau = \sigma T \circ S\tau \circ \kappa^\lambda R
\end{equation}
\end{lemma}

\begin{proof}
Easy to see using string diagrams:
\begin{equation}
  \input{2_lemma_asymmetry_2.tikz}
\end{equation}
\end{proof}

Table~\ref{table:second_comp} sums up which axioms of $\sigma$ and $\tau$ are sufficient to make axioms hold for $\psi$.

\begin{table}[htp]
  \caption{Required axioms for the composite law $\psi$}
\begin{center}
\begin{tabular}{llll}
    \hline
    $\psi$ & requires & $\sigma$ & $\tau$ \\ 
    \hline 
    \eqref{ax:etap} & & \eqref{ax:etap} & \eqref{ax:etap} \\
    \eqref{ax:etam} & & \eqref{ax:etam} & \eqref{ax:etam} \\
    \eqref{ax:mup} & & \eqref{ax:mup} & \eqref{ax:mup}, \eqref{ax:etam}, \eqref{ax:mum} \\ 
    \eqref{ax:mum} & & \eqref{ax:mum} & \eqref{ax:mum} \\
    \hline
\end{tabular} \label{table:second_comp}
\end{center}
\end{table}

Using the information contained in Table~\ref{table:second_comp}, we get the following result.

\begin{theorem} \label{theo:second_comp}
  Let $R$, $S$, $T$ be monads, $\lambda \colon [S,T]$, $\sigma \colon [R,S]$ be weak distributive laws, and $\tau \colon [R,T]$ be a distributive law, satisfying the Yang-Baxter equation. Then the composite $\psi \triangleq R\pi^\lambda \circ \sigma T \circ S\tau \circ \iota^\lambda R \colon [R,\wcomp{S}{}{T}]$ is a weak distributive law. If, additionally, $\sigma$ is a distributive law, then $\psi$ is a distributive law.
\end{theorem}

\subsubsection{Associativity of Weak Iteration}

According to the preceding paragraphs, provided the Yang-Baxter equation holds and enough axioms are satisfied, there are two ways of obtaining a composite weak distributive law. The first one, described in Section~\ref{subsec:first_comp}, yields by Theorem~\ref{theo:first_comp} a weak distributive law $\phi \colon [\wcomp{R}{}{S},T]$ and thus a monad $\wcomp{(\wcomp{R}{}{S})}{}{T}$. The second one, described in Section~\ref{subsec:second_comp}, yields by Theorem~\ref{theo:second_comp} a weak distributive law $\psi \colon [R,\wcomp{S}{}{T}]$ and thus a monad $\wcomp{R}{}{(\wcomp{S}{}{T})}$. Contrary to Theorem~\ref{theo:cheng}, it is not obvious that even the underlying functors of these two monads are the same. It turns out that they are, and that more broadly, one can identify these two monads, allowing us to write unambiguously $\wcomp{R}{}{\wcomp{S}{}{T}}$.

\begin{theorem} \label{theo:comp_equal}
    Let $R$, $S$, $T$ be monads, $\lambda \colon [S,T]$, $\sigma \colon [R,S]$ be weak distributive laws and $\tau \colon [R,T]$ be a distributive law, satisfying the Yang-Baxter equation. Then the weak distributive laws $\phi \triangleq \pi^\sigma T \circ R \lambda \circ \tau S \circ T\iota^\sigma \colon [\wcomp{R}{}{S},T]$ and $\psi \triangleq R\pi^\lambda \circ \sigma T \circ S\tau \circ \iota^\lambda R \colon [R,\wcomp{S}{}{T}]$ generate the same monad, \textit{i.e.}, $\wcomp{(\wcomp{R}{}{S})}{}{T} = \wcomp{R}{}{(\wcomp{S}{}{T})}$.
\end{theorem}

\begin{proof}
After some string-diagrammatic computations (see Appendix~\ref{app:iteration}), one can express the idempotents $\kappa^\phi$ and $\kappa^\psi$ as follows:
\begin{align} 
& \kappa^\phi = \pi^\sigma T \circ R\kappa^\lambda \circ \iota^\sigma T = \alpha \circ \beta \label{eq:kappa_phi} \\ 
& \kappa^\psi = R\pi^\lambda \circ \kappa^\sigma T \circ R\iota^\lambda = \beta \circ \alpha \label{eq:kappa_psi}
\end{align}
where $\alpha \triangleq \pi^\sigma T \circ R\iota^\lambda$ and $\beta \triangleq R\pi^\lambda \circ \iota^\sigma T$, as pictured below.
\begin{equation}
\begin{tikzpicture}
	\begin{pgfonlayer}{nodelayer}
		\node [style=none] (16) at (-2.5, 0) {$\kappa^\psi$};
		\node [style=none] (33) at (-9.5, 0) {$\kappa^\phi$};
		\node [style=none] (34) at (8, 0) {};
		\node [style=none] (35) at (8.5, 0.5) {};
		\node [style=none] (36) at (9, 0) {};
		\node [style=none] (37) at (8.475, 0.475) {};
		\node [style=none] (38) at (8.475, 1.5) {};
		\node [style=none] (39) at (8.525, 0.475) {};
		\node [style=none] (40) at (8.525, 1.5) {};
		\node [style=none] (41) at (9, 0) {};
		\node [style=none] (42) at (9.475, -0.5) {};
		\node [style=none] (43) at (10, 0) {};
		\node [style=none] (44) at (9.475, -1.5) {};
		\node [style=none] (45) at (9.475, -0.475) {};
		\node [style=none] (46) at (9.525, -1.5) {};
		\node [style=none] (47) at (9.525, -0.475) {};
		\node [style=none] (48) at (8, -1.5) {};
		\node [style=none] (49) at (10, 1.5) {};
		\node [style=none] (50) at (6.5, 0) {$\alpha$};
		\node [style=none] (51) at (16, 0) {};
		\node [style=none] (52) at (16.5, 0.5) {};
		\node [style=none] (53) at (17, 0) {};
		\node [style=none] (54) at (16.475, 0.475) {};
		\node [style=none] (55) at (16.475, 1.5) {};
		\node [style=none] (56) at (16.525, 0.475) {};
		\node [style=none] (57) at (16.525, 1.5) {};
		\node [style=none] (58) at (15, 0) {};
		\node [style=none] (59) at (15.475, -0.5) {};
		\node [style=none] (60) at (16, 0) {};
		\node [style=none] (61) at (15.475, -1.5) {};
		\node [style=none] (62) at (15.475, -0.475) {};
		\node [style=none] (63) at (15.525, -1.5) {};
		\node [style=none] (64) at (15.525, -0.475) {};
		\node [style=none] (65) at (15, 1.5) {};
		\node [style=none] (66) at (17, -1.5) {};
		\node [style=none] (67) at (13.5, 0) {$\beta$};
		\node [style=none] (68) at (-8, 0.75) {};
		\node [style=none] (69) at (-7.5, 1.25) {};
		\node [style=none] (70) at (-7, 0.75) {};
		\node [style=none] (71) at (-7.525, 1.225) {};
		\node [style=none] (72) at (-7.525, 1.5) {};
		\node [style=none] (73) at (-7.475, 1.225) {};
		\node [style=none] (74) at (-7.475, 1.5) {};
		\node [style=none] (75) at (-7, 0.75) {};
		\node [style=none] (76) at (-6.525, 0.25) {};
		\node [style=none] (77) at (-6, 0.75) {};
		\node [style=none] (78) at (-6.525, 0) {};
		\node [style=none] (79) at (-6.525, 0.275) {};
		\node [style=none] (80) at (-6.475, 0) {};
		\node [style=none] (81) at (-6.475, 0.275) {};
		\node [style=none] (82) at (-8, 0) {};
		\node [style=none] (83) at (-6, 1.5) {};
		\node [style=none] (84) at (-7, -0.75) {};
		\node [style=none] (85) at (-6.5, -0.25) {};
		\node [style=none] (86) at (-6, -0.75) {};
		\node [style=none] (87) at (-6.525, -0.275) {};
		\node [style=none] (88) at (-6.525, 0) {};
		\node [style=none] (89) at (-6.475, -0.275) {};
		\node [style=none] (90) at (-6.475, 0) {};
		\node [style=none] (91) at (-8, -0.75) {};
		\node [style=none] (92) at (-7.525, -1.25) {};
		\node [style=none] (93) at (-7, -0.75) {};
		\node [style=none] (94) at (-7.525, -1.5) {};
		\node [style=none] (95) at (-7.525, -1.225) {};
		\node [style=none] (96) at (-7.475, -1.5) {};
		\node [style=none] (97) at (-7.475, -1.225) {};
		\node [style=none] (98) at (-8, 0) {};
		\node [style=none] (99) at (-6, -1.5) {};
		\node [style=none] (100) at (-1, -0.75) {};
		\node [style=none] (101) at (-0.5, -0.25) {};
		\node [style=none] (102) at (0, -0.75) {};
		\node [style=none] (103) at (-0.525, -0.275) {};
		\node [style=none] (104) at (-0.525, 0) {};
		\node [style=none] (105) at (-0.475, -0.275) {};
		\node [style=none] (106) at (-0.475, 0) {};
		\node [style=none] (107) at (0, -0.75) {};
		\node [style=none] (108) at (0.475, -1.25) {};
		\node [style=none] (109) at (1, -0.75) {};
		\node [style=none] (110) at (0.475, -1.5) {};
		\node [style=none] (111) at (0.475, -1.225) {};
		\node [style=none] (112) at (0.525, -1.5) {};
		\node [style=none] (113) at (0.525, -1.225) {};
		\node [style=none] (114) at (-1, -1.5) {};
		\node [style=none] (115) at (1, 0) {};
		\node [style=none] (116) at (0, 0.75) {};
		\node [style=none] (117) at (0.5, 1.25) {};
		\node [style=none] (118) at (1, 0.75) {};
		\node [style=none] (119) at (0.475, 1.225) {};
		\node [style=none] (120) at (0.475, 1.5) {};
		\node [style=none] (121) at (0.525, 1.225) {};
		\node [style=none] (122) at (0.525, 1.5) {};
		\node [style=none] (123) at (-1, 0.75) {};
		\node [style=none] (124) at (-0.525, 0.25) {};
		\node [style=none] (125) at (0, 0.75) {};
		\node [style=none] (126) at (-0.525, 0) {};
		\node [style=none] (127) at (-0.525, 0.275) {};
		\node [style=none] (128) at (-0.475, 0) {};
		\node [style=none] (129) at (-0.475, 0.275) {};
		\node [style=none] (130) at (-1, 1.5) {};
		\node [style=none] (131) at (1, 0) {};
		\node [style=none] (132) at (-7.5, 2) {$H$};
		\node [style=none] (133) at (8.5, 2) {$H$};
		\node [style=none] (134) at (-7.5, -2) {$H$};
		\node [style=none] (135) at (15.5, -2) {$H$};
		\node [style=none] (136) at (0.5, 2) {$K$};
		\node [style=none] (137) at (9.5, -2) {$K$};
		\node [style=none] (138) at (16.5, 2) {$K$};
		\node [style=none] (139) at (0.5, -2) {$K$};
		\node [style=none] (140) at (-1, 2) {$R$};
		\node [style=none] (141) at (-1, -2) {$R$};
		\node [style=none] (142) at (15, 2) {$R$};
		\node [style=none] (143) at (8, -2) {$R$};
		\node [style=none] (144) at (-6, 2) {$T$};
		\node [style=none] (145) at (-6, -2) {$T$};
		\node [style=none] (146) at (10, 2) {$T$};
		\node [style=none] (147) at (17, -2) {$T$};
		\node [style=none] (148) at (-8.75, -0.15) {$=$};
		\node [style=none] (149) at (-1.75, -0.15) {$=$};
		\node [style=none] (150) at (7.25, 0) {$\triangleq$};
		\node [style=none] (151) at (14.25, 0) {$\triangleq$};
	\end{pgfonlayer}
	\begin{pgfonlayer}{edgelayer}
		\draw [style=yellow edge, in=180, out=90] (34.center) to (35.center);
		\draw [style=red edge, in=0, out=90] (36.center) to (35.center);
		\draw [style=yellow edge] (37.center) to (38.center);
		\draw [style=red edge] (39.center) to (40.center);
		\draw [style=red edge, in=630, out=-180] (42.center) to (41.center);
		\draw [style=blue edge, in=-90, out=0] (42.center) to (43.center);
		\draw [style=red edge] (44.center) to (45.center);
		\draw [style=blue edge] (46.center) to (47.center);
		\draw [style=yellow edge] (34.center) to (48.center);
		\draw [style=blue edge] (49.center) to (43.center);
		\draw [style=red edge, in=180, out=90] (51.center) to (52.center);
		\draw [style=blue edge, in=0, out=90] (53.center) to (52.center);
		\draw [style=red edge] (54.center) to (55.center);
		\draw [style=blue edge] (56.center) to (57.center);
		\draw [style=yellow edge, in=630, out=-180] (59.center) to (58.center);
		\draw [style=red edge, in=-90, out=0] (59.center) to (60.center);
		\draw [style=yellow edge] (61.center) to (62.center);
		\draw [style=red edge] (63.center) to (64.center);
		\draw [style=yellow edge] (65.center) to (58.center);
		\draw [style=blue edge] (53.center) to (66.center);
		\draw [style=yellow edge, in=180, out=90] (68.center) to (69.center);
		\draw [style=red edge, in=0, out=90] (70.center) to (69.center);
		\draw [style=yellow edge] (71.center) to (72.center);
		\draw [style=red edge] (73.center) to (74.center);
		\draw [style=red edge, in=630, out=-180] (76.center) to (75.center);
		\draw [style=blue edge, in=-90, out=0] (76.center) to (77.center);
		\draw [style=red edge] (78.center) to (79.center);
		\draw [style=blue edge] (80.center) to (81.center);
		\draw [style=yellow edge] (68.center) to (82.center);
		\draw [style=blue edge] (83.center) to (77.center);
		\draw [style=red edge, in=180, out=90] (84.center) to (85.center);
		\draw [style=blue edge, in=0, out=90] (86.center) to (85.center);
		\draw [style=red edge] (87.center) to (88.center);
		\draw [style=blue edge] (89.center) to (90.center);
		\draw [style=yellow edge, in=630, out=-180] (92.center) to (91.center);
		\draw [style=red edge, in=-90, out=0] (92.center) to (93.center);
		\draw [style=yellow edge] (94.center) to (95.center);
		\draw [style=red edge] (96.center) to (97.center);
		\draw [style=yellow edge] (98.center) to (91.center);
		\draw [style=blue edge] (86.center) to (99.center);
		\draw [style=yellow edge, in=180, out=90] (100.center) to (101.center);
		\draw [style=red edge, in=0, out=90] (102.center) to (101.center);
		\draw [style=yellow edge] (103.center) to (104.center);
		\draw [style=red edge] (105.center) to (106.center);
		\draw [style=red edge, in=630, out=-180] (108.center) to (107.center);
		\draw [style=blue edge, in=-90, out=0] (108.center) to (109.center);
		\draw [style=red edge] (110.center) to (111.center);
		\draw [style=blue edge] (112.center) to (113.center);
		\draw [style=yellow edge] (100.center) to (114.center);
		\draw [style=blue edge] (115.center) to (109.center);
		\draw [style=red edge, in=180, out=90] (116.center) to (117.center);
		\draw [style=blue edge, in=0, out=90] (118.center) to (117.center);
		\draw [style=red edge] (119.center) to (120.center);
		\draw [style=blue edge] (121.center) to (122.center);
		\draw [style=yellow edge, in=630, out=-180] (124.center) to (123.center);
		\draw [style=red edge, in=-90, out=0] (124.center) to (125.center);
		\draw [style=yellow edge] (126.center) to (127.center);
		\draw [style=red edge] (128.center) to (129.center);
		\draw [style=yellow edge] (130.center) to (123.center);
		\draw [style=blue edge] (118.center) to (131.center);
	\end{pgfonlayer}
\end{tikzpicture}
\end{equation}
Now let us split $\kappa^\phi = \iota^\phi \circ \pi^\phi$ with $\pi^\phi \circ \iota^\phi = 1_{\wcomp{(\wcomp{R}{}{S})}{}{T}}$. We show that this splitting generates also a splitting for $\kappa^\psi$. Indeed, one can choose $\iota^\psi \triangleq \beta \circ \iota^\phi$ and $\pi^\psi \triangleq \pi^\phi \circ \alpha$ because 
\begin{align} 
& \iota^\psi \circ \pi^\psi = \beta \circ \iota^\phi \circ \pi^\phi \circ \alpha = \beta \circ \kappa^\phi \circ \alpha = \beta \circ \alpha \circ \beta \circ \alpha = \kappa^\psi \circ \kappa^\psi = \kappa^\psi \\ 
& \pi^\psi \circ \iota^\psi = \pi^\phi \circ \alpha \circ \beta \circ \iota^\phi = \pi^\phi \circ \kappa^\phi \circ \iota^\phi = (\pi^\phi \circ \iota^\phi) \circ (\pi^\phi \circ \iota^\phi) = 1_{\wcomp{(\wcomp{R}{}{S})}{}{T}}
\end{align}
Hence the splitting of $\kappa^\psi$ is as follows:
\begin{equation} 
\begin{tikzcd}
    R(\wcomp{S}{}{T}) \arrow[rr, "\pi^\psi"', Rightarrow] \arrow[rd, "\alpha", Rightarrow] \arrow[rrrr, "\kappa^\psi", shift left=3, Rightarrow] &                                         & \wcomp{R}{}{(\wcomp{S}{}{T})} \arrow[rr, "\iota^\psi"', Rightarrow] &                                       & R(\wcomp{S}{}{T}) \\
                                                                                                             & (\wcomp{R}{}{S})T \arrow[r, "\pi^\phi", Rightarrow] & \wcomp{(\wcomp{R}{}{S})}{}{T} \arrow[r, "\iota^\phi", Rightarrow]   & (\wcomp{R}{}{S})T \arrow[ru, "\beta", Rightarrow] &                  
    \end{tikzcd}
\end{equation}
Therefore, the functors $\wcomp{(\wcomp{R}{}{S})}{}{T}$ and $\wcomp{R}{}{(\wcomp{S}{}{T})}$ can be identified. We denote it by $\wcomp{R}{}{\wcomp{S}{}{T}}$ (short notation $\Omega$) and picture it as follows:
\begin{equation}
\begin{tikzpicture}
	\begin{pgfonlayer}{nodelayer}
		\node [style=none] (1) at (-0.05, -1.025) {};
		\node [style=none] (2) at (-0.05, 1) {};
		\node [style=none] (3) at (0, -1.025) {};
		\node [style=none] (4) at (0, 1) {};
		\node [style=none] (5) at (0.05, -1.025) {};
		\node [style=none] (6) at (0.05, 1) {};
		\node [style=none] (7) at (0, -1.5) {$\Omega$};
		\node [style=none] (8) at (0, 1.5) {$\Omega$};
		\node [style=none] (9) at (-1.75, 0) {$1_{\Omega}$};
		\node [style=none] (10) at (-0.75, 0) {$\triangleq$};
	\end{pgfonlayer}
	\begin{pgfonlayer}{edgelayer}
		\draw [style=yellow edge] (1.center) to (2.center);
		\draw [style=red edge] (3.center) to (4.center);
		\draw [style=blue edge] (5.center) to (6.center);
	\end{pgfonlayer}
\end{tikzpicture}
\end{equation}
Moreover, one can separate the three functors -- or merge the three functors -- in any order \emph{up to an idempotent}, as shown in the following diagrams (proved in Appendix~\ref{app:iteration}):
\begin{equation} \label{eq:separate_merge}
\input{2_separate_merge.tikz}
\end{equation}
Using these compatibility properties along with aforementioned properties of $\kappa$, it is then a matter of lengthy but straightforward computations to show that units and multiplications of $\wcomp{(\wcomp{R}{}{S})}{}{T}$ and $\wcomp{R}{}{(\wcomp{S}{}{T})}$ coincide (see Appendix~\ref{app:iteration}).
\end{proof}

\subsection{Examples} \label{subsec:examples}

In this section, we illustrate many cases where the results of the previous section can be applied by proving the Yang-Baxter equation for some usual monads and (weak) distributive laws.

\subsubsection{Iterating with Trivial Weak Distributive Laws}

When one of the three laws involved is trivial, the Yang-Baxter equation usually simplifies to a condition relating the underlying monad morphism to the other laws. We mention the two following cases.

\begin{proposition} \label{prop:trivial_yb}
Let $\tau \colon [R,T]$ be a weak distributive law.
\begin{itemize} 
\item Let $\sigma \colon [R,S]$ be a distributive law and $\lambda \triangleq \eta^S \mu^T \circ T\gamma \colon [S,T]$ be a trivial weak distributive law arising from a monad morphism $\gamma \colon S \nto T$ such that $R\gamma \circ \sigma = \tau \circ \gamma R$. Then the Yang-Baxter equation is satisfied.
\item Let $\lambda \colon [S,T]$ be a weak distributive law and $\sigma \triangleq \eta^R \mu^S \circ S\gamma \colon [R,S]$ be a trivial weak distributive law arising from a monad morphism $\gamma \colon R \nto S$ such that $\gamma T \circ \tau = \lambda \circ \gamma T$. Then the Yang-Baxter equation is satisfied.
\end{itemize}
\end{proposition}


\subsubsection{Iterating with Exceptions}

The exception monad $E$ from Example~\ref{ex:exception} can be described by seeing $\inl \colon 1 \nto E$ and $\inr \colon \underline{1} \nto E$ as natural transformations, where $1$ is the identity functor and $\underline{1}$ is the constant functor $\underline{1}X \triangleq 1 = \{*\}$, $\underline{1}f \triangleq 1_{\{*\}}$. One can then define a well-known generic distributive law relying directly on these two coproduct injections.

\begin{example} For any $\Set$ monad $T$, there is a distributive law $\epsilon^T \colon [T,E]$ defined by 
\begin{align}
& \epsilon^T \circ \inl T = T\eta^E \\ 
& \epsilon^T \circ \inr T = \eta^T E \circ \inr 
\end{align}
\end{example}

Moreover, for monads $S$ and $T$, the distributive laws $\epsilon^S$ and $\epsilon^T$ are coherent with each other in the sense of the next proposition.

\begin{proposition} \label{prop:yb_te}
Let $\lambda \colon [S,T]$ be a natural transformation such that either the~\eqref{ax:etap} or the~\eqref{ax:etam} axiom holds. Then the following Yang-Baxter equation holds:
\begin{equation} 
    \lambda E \circ T \epsilon^S \circ \epsilon^T S = S\epsilon^T \circ \epsilon^S T \circ E \lambda 
\end{equation}
\end{proposition}


\begin{remark}
    The above proposition still holds true when $E$ implements multiple exceptions. More generally, it can be generalised to any category with binary coproducts and any exception monad of the form $EX = X + e$, where $e$ is an object of exceptions.
\end{remark}    

Consequently, by applying Theorem~\ref{theo:comp_equal}, given any weak distributive law $[S,T]$, one can safely compose the resulting monad with the exception monad $E$ and get a monad $S \bullet T \circ E$. We give two explicit examples.
\begin{itemize} 
    \item Applying Theorem~\ref{theo:second_comp} to the weak distributive law $\lambda \colon [P,D]$ from Example~\ref{ex:law_pd}, we derive a new weak distributive law $\lambda E \circ D\epsilon^P \colon [P,D_{\leq}]$ between the powerset monad $P$ and the \emph{subdistribution monad} $D_{\leq} = D \circ E$. The expression of this new law is as in~\eqref{eq:pd_def}.
    \item Remarking that $P = \comp{P_*}{}{E}$ where $P_*$ is the \emph{non-empty} powerset monad, we identify the weak distributive law $\lambda \colon [P,P]$ from Example~\ref{ex:law_pp} as an iterated law built from $\epsilon^P$, $\epsilon^{P_*}$, and a weak distributive law of type $[P,P_*]$ whose expression is as in~\eqref{eq:pp_def}.
\end{itemize}

\subsubsection{Iterating with the Reader}

Recall that $A$ is a set of labels and the reader monad in $\Set$ is defined as $RX = X^A$. For any $a \in A$, define the function $\hat{a}_X \colon RX \to X$ by $\hat{a}_X(h) = h(a)$.

\begin{example} \label{ex:rt} For any $\Set$ monad $T$, there is a distributive law $\rho^T \colon [R,T]$ defined for any set $X$ and any $z \in TRX$ by
\begin{equation}
    \rho^T_X(z) \triangleq a \mapsto T\hat{a}_X(z)
\end{equation}
\end{example}

Moreover, for monads $S$ and $T$, the distributive laws $\rho^S$ and $\rho^T$ are coherent with each other in the sense of the next proposition.

\begin{proposition} \label{prop:yb_rt}
Let $\lambda \colon [S,T]$ be a natural transformation. Then the following Yang-Baxter equation holds:
\begin{equation} 
\rho^S T \circ S \rho^T \circ \lambda R = R \lambda \circ \rho^T S \circ T \rho^S
\end{equation}
\end{proposition}


Applying Theorem~\ref{theo:comp_equal} with \textit{e.g.} the weak $\lambda \colon [P,D]$ from Example~\ref{ex:law_pd}, we deduce that the monad $\comp{R}{}{\wcomp{P}{}{D}}$ can be obtained either using the distributive law $\psi = \rho^{\wcomp{P}{}{D}} : [R,\wcomp{P}{}{D}]$ or using the weak distributive law $\phi \triangleq R\lambda \circ \rho^D P \colon [\comp{R}{}{P},D]$ given by the expression 
\begin{equation}
     \phi_X \left( \sum_{h \in (PX)^A} p_h \cdot h \right) = a \mapsto \left\{ \sum_{h \in (PX)^A} p_h q^h \mid q^h \in DX, \, \supp(q^h) \subseteq h(a) \right\}
\end{equation}

\subsubsection{Iterating with the Writer}

Recall that $M$ is a monoid and the writer monad in $\Set$ is defined as $MX = M \times X$. For any $m \in M$, let $\overline{m}_X \colon X \to M \times X$ be defined by $\overline{m}_X(x) = (m,x)$.

\begin{example} \label{ex:tw} For any $\Set$ monad $T$, there is a distributive law $\omega^T \colon [T,W]$ defined for any set $X$ and any $(m,t) \in WTX$ by 
\begin{equation}
    \omega^T_X(m,t) \triangleq T\overline{m}_X(t)
\end{equation}
\end{example}

Moreover, for monads $S$ and $T$, the distributive laws $\omega^S$ and $\omega^T$ are coherent with each other in the sense of the next proposition.

\begin{proposition} \label{prop:yb_tw}
Let $\lambda \colon [S,T]$ be a natural transformation. Then the following Yang-Baxter equation holds:
\begin{equation} 
\lambda W \circ T\omega^S \circ \omega^T S = S\omega^T \circ \omega^S T \circ W\lambda
\end{equation}
\end{proposition}


Applying Theorem~\ref{theo:comp_equal} with \textit{e.g.} the weak $\lambda : [P,D]$ from Example~\ref{ex:law_pd}, we deduce that the monad $\comp{\wcomp{P}{}{D}}{}{W}$ can be obtained either using the distributive law $\phi = \omega^{\wcomp{P}{}{D}} : [\wcomp{P}{}{D},W]$ or using the weak distributive law $\psi \triangleq \lambda W \circ T \omega^S : [P,\comp{D}{}{W}]$ given by the expression 
\begin{equation}
    \psi_X \left( \sum_{i,j} p_{i,j} \cdot (m_i,U_j) \right) = \left\{ \sum_{i,j} p_{i,j} q^{i,j} \mid q^{i,j} \in DWX, \, \supp(q^{i,j}) \subseteq \{(m_i,x) \mid x \in U_j \} \right\}
\end{equation}

\subsection{Algebras as Distributive Laws, and Iterating them} \label{subsec:outputs}

In this short section we sketch how Eilenberg-Moore algebras are a particular kind of distributive laws, and how in this view $\lambda$-algebras are precisely Yang-Baxter equations. Given a monad $T$, an \emph{Eilenberg-Moore algebra} -- or \emph{$T$-algebra} for short -- is a pair $(A,t)$ where $A$ is an object of $\C$ and $t \colon TA \to A$ is a morphism such that $t \circ \eta^T_A = 1_A$ (unitality) and $t \circ \mu^T_A = t \circ Tt$ (associativity). A morphism of $T$-algebras $f \colon (A,t) \to (B,u)$ is a morphism $f \colon A \to B$ in $\C$ such that $f \circ t = u \circ Tf$.
For any functor $F \colon \C \to \C$, an \emph{Eilenberg-Moore distributive law} -- or \emph{EM-law} for short -- is a natural transformation $\lambda \colon [F,T]$ such that the~\eqref{ax:etam} and~\eqref{ax:mum} axioms are satisfied.
Let $S$ and $T$ be two monads and $\lambda \colon [S,T]$ be a natural transformation. A \emph{$\lambda$-algebra} is a triple $(A,s,t)$ such that $(A,s)$ is an $S$-algebra, $(A,t)$ is a $T$-algebra, and the following condition is satisfied:
\begin{equation} 
s \circ St \circ \lambda_X = t \circ Ts \tag{$\lambda$-condition} \label{eq:lambda_cond}
\end{equation}
A morphism of $\lambda$-algebras is a simultaneous $S$-algebra and $T$-algebra morphism.

The following result is widely known:
\begin{theorem}[\cite{Beck69,Garner20}] \label{theo:algebra_iso}
Let $\lambda \colon [S,T]$ be a natural transformation.
\begin{itemize} 
\item If $\lambda$ is a distributive law, then the categories of $(\comp{S}{\lambda}{T})$-algebras and $\lambda$-algebras are isomorphic.
\item If $\lambda$ is a weak distributive law, then the categories of $(\wcomp{S}{\lambda}{T})$-algebras and $\lambda$-algebras are isomorphic.
\end{itemize}
\end{theorem}

Let $A$ be an object of $\C$ and let $\underline{A} \colon \C \to \C$ be the constant functor $\underline{A}X \triangleq A$, $\underline{A}f \triangleq 1_A$.

\begin{proposition} \label{prop:output_1}
Let $t \colon TA \to A$ be a morphism. The natural transformation $[t] \colon [\underline{A},T]$ defined by $[t]_X = t$ is an EM-law if and only if $t$ is a $T$-algebra.
\end{proposition}

\begin{proof}
Naturality of $[t]$ amounts to the obvious equation $1_A \circ t = t \circ T(1_A)$. The~\eqref{ax:etam} axiom of $[t]$ is equivalent to the equation $[t]_X \circ \eta^T_{\underline{A}X} = \underline{A} \eta^T_X$ for all objects $X$, which simplifies into the unitality axiom $t \circ \eta^T_A = 1_A$. The~\eqref{ax:mum} axiom of $[t]$ is equivalent to the equation $\underline{A}\mu^T_X \circ [t]_{TX} \circ T[t]_X = [t]_X \circ \mu^T_{\underline{A}X}$ for all objects $X$, which simplifies into the associativity axiom $t \circ Tt = t \circ \mu^T_A$.
\end{proof}

\begin{proposition} \label{prop:output_2}
Let $(A,s)$ be an $S$-algebra, $(A,t)$ be a $T$-algebra, and $\lambda \colon [S,T]$ be a natural transformation. Then $(A,s,t)$ satisfies the~\eqref{eq:lambda_cond} if and only if the Yang-Baxter equation holds between $[s] \colon [\underline{A},S]$, $[t] \colon [\underline{A},T]$ and $\lambda \colon [S,T]$.
\end{proposition}

\begin{proof}
The Yang-Baxter equation is $[s]_{TX} \circ S[t]_X \circ \lambda_{\underline{A}X} = \underline{A} \lambda_X \circ [t]_{SX} \circ T[s]_X$ for all objects $X$, which simplifies into the~\eqref{eq:lambda_cond} $s \circ St \circ \lambda_A = 1_A \circ t \circ Ts$.
\end{proof}

Therefore 
\begin{itemize}
    \item the functor $\lambda$-algebras $\to (\comp{S}{\lambda}{T})$-algebras from Theorem~\ref{theo:algebra_iso} is a particular case of Theorem~\ref{theo:cheng}; 
    \item the functor $\lambda$-algebras $\to (\wcomp{S}{\lambda}{T})$-algebras from Theorem~\ref{theo:algebra_iso} is a particular case of Theorem~\ref{theo:second_comp}.
\end{itemize} 

Let us illustrate \textit{e.g.} the second point. Let $\lambda \colon [S,T]$ be a weak distributive law and $(A,s,t)$ be a $\lambda$-algebra. According to Propositions~\ref{prop:output_1} and~\ref{prop:output_2}, we are in the presence of a Yang-Baxter equation with one weak distributive law $\lambda$ and two EM-laws $[s]$ and $[t]$. By applying a mild variation of Theorem~\ref{theo:second_comp} -- consisting in forgetting the monad structure of the functor $R = \underline{A}$ -- we get an EM-law $\psi = \underline{A} \pi^\lambda \circ [s]T \circ S[t] \circ \iota^\lambda \underline{A} \colon [\underline{A},\wcomp{S}{\lambda}{T}]$. That law satisfies $\psi_X = s \circ St \circ \iota^\lambda_A$ for all objects $X$. Via Proposition~\ref{prop:output_1}, we get that $s \circ St \circ \iota^\lambda_A$ is an $(\wcomp{S}{\lambda}{T})$-algebra. As expected, this procedure amounts to the action of the functor $\lambda$-algebras $\to (\wcomp{S}{\lambda}{T})$-algebras, that is, $(A,s,t) \mapsto s \circ St \circ \iota^\lambda_A$.

\section{Discussion and Future Work} \label{sec:conclusion}

\subsection{Advocating String Diagrams} \label{subsec:advocating}

One of the aims of this paper is to advocate string diagrams for proof research and communication. Let us provide two more examples that support our position.

\subsubsection{Cheng's theorem} \label{subsec:cheng}
Let us begin with Cheng's result~\cite[Theorem~1.6]{Cheng11a}, proved in~\cite[Appendix~A]{Cheng11a}. The largest commutative diagram occurring in the proof for the initialisation case with $3$ monads comprises $12$ elementary commutative polygons, out of which $9$ are naturality squares. Therefore, there are only $3$ significant steps in the underlying proof. In an equivalent string-diagrammatic representation of the proof, these three steps are remarkably well identified:
\begin{equation} \label{eq:cheng_diagram}
    \input{0_cheng.tikz}
\end{equation}
That computation (more precisely, its symmetric) has been generalised to weak distributive laws in the present paper: compare diagrams~\eqref{eq:cheng_diagram} and~\eqref{eq:cheng_returns}.

\subsubsection{Winter's theorem} \label{subsec:winter}

Another striking example arises from~\cite{Winter16}. As already mentioned in the introduction, the proof of~\cite[Theorem~5]{Winter16} required a Prolog program to find a convenient tiling of the largest required commutative diagram, depicted in \cite[Appendix~A]{Winter16}. That diagram comprises $13+10=23$ elementary commutative polygons, out of which $15$ are naturality squares. Therefore, there are only $8$ significant steps in the underlying proof. An equivalent string-diagrammatic formalisation of the diagram in question is presented in Appendix~\ref{app:conclusion}. As expected, it contains only $8$ steps, and the computation can be carried out in one single line. Winter's proof research was obfuscated by what Hinze and Marsden call \emph{bookkeeping}~\cite{Hinze16}, \textit{i.e.}, administrative operations devoid of significant mathematical content.

\subsection{Conclusion} \label{subsec:conclusion}
This paper generalises iteration to weak distributive laws using the graphical calculus of string diagrams. It is an additional step in developing the theory of weak distributive laws, as is our line of research in~\cite{Goy20,Goy21,GoyTh}. The paper revolves around one main theoretical result, split in Theorems~\ref{theo:first_comp},~\ref{theo:second_comp} and~\ref{theo:comp_equal}, and provides several examples. A possible extension would be to derive  -- in the style of Section~\ref{subsec:examples} -- a generic Yang-Baxter equation stating compatibility of \emph{monotone} laws. Monotone laws are a certain class of laws of type $[P,T]$ and often carry significant semantic content. This direction is particularly interesting, as all known non-trivial examples of weak distributive laws are based on monotone laws. One may also ask whether the existence of a composite law entails the Yang-Baxter equation -- this question has been positively solved by Winter in the case of distributive laws~\cite[Proposition~11]{Winter16}.

As hinted in Section~\ref{subsec:outputs}, our results easily adapt to EM-laws (and their duals, called Kleisli-laws), \textit{i.e.}, distributive laws between a monad and a \emph{functor}. This paves the way for applications of iterated laws in coalgebra theory -- as already attempted in~\cite{Winter16} -- where both EM-laws and Kleisli-laws are ubiquitous. Another possible follow-up to the vision developed in Section~\ref{subsec:outputs} is to \emph{weaken algebras}. Indeed, as some EM-laws turn out to be algebras, some \emph{weak} EM-laws turn out to be so-called \emph{semi}algebras (as defined in~\cite{Garner20}). Dually, \emph{coweak} Kleisli-laws probably correspond to an interesting class of free algebras. Clarifying the whole picture will most likely enhance the still underdeveloped theory of \emph{coweak} distributive laws.

Iteration of standard distributive laws has been heavily exploited in the context of categories-as-monads~\cite{Zanasi15,Zanasi15b,Zanasi17,Zanasi18}. Phrasing our results in a $2$-categorical setting, as did Cheng for iterated distributive laws~\cite{Cheng11a}, would enable new case studies to arise in that area.

On another level, our key message is to advocate string diagrams to deal with categorical proofs in contexts where traditional equational or commutative-diagrammatic proofs become unreadable. String diagrams are often more convenient in such proofs, as they allow to take a step further by hiding part of the complexity in the graphical formalism. Nevertheless, they also have limitations. In the main result of Cheng~\cite[Theorem~1.6]{Cheng11a}, iteration is performed not only for three monads but for any finite number $n$ of monads: the $n = 3$ case does only initialise an induction. Generalising the $n = 3$ result to weak laws becomes reachable using string diagrams -- though some diagrams in the proof of Theorem~\ref{theo:comp_equal} already have substantial size. Even with a string-diagrammatic approach, generalizing further to $n$ monads may require tedious proofs. Our opinion is that ingenious string-diagrammatic reasoning significantly pushes back the -- fuzzy and subjective -- boundary between proofs that are intelligible and those that are not.

\bibliographystyle{entics}
\bibliography{ada_biblio}

\appendix 


\section{Proofs of Section~\ref{sec:laws}} \label{app:laws}

In this section we prove that $\wcomp{S}{\lambda}{T}$ is indeed a monad.
To prove the first unitality axiom, \textit{i.e.}, that $\mu^{\wcomp{S}{}{T}} \circ \eta^{\wcomp{S}{}{T}} K = 1_K$, we need equation~\eqref{eq:kappa_eta}, the equation $\iota \circ \pi = \kappa$, the retract equation $\pi \circ \iota = 1_K$, $\eta^S$ unitality and the retract equation again.

\begin{equation}
\input{A1_weak_composite_1.tikz}
\end{equation}

To prove the second unitality axiom, \textit{i.e.}, that $\mu^{\wcomp{S}{}{T}} \circ K \eta^{\wcomp{S}{}{T}} = 1_K$, we need equation~\eqref{eq:kappa_eta}, $\lambda$~\eqref{ax:etap} axiom, unitality of both $\eta^S$ and $\eta^T$, and the retract equation.

\begin{equation}
\input{A1_weak_composite_2.tikz}
\end{equation}

As for distributive laws, it is possible to prove, using only the~\eqref{ax:mup} and~\eqref{ax:mum} axioms, that $\theta \triangleq \mu^S\mu^T \circ S\lambda T$ is associative, \textit{i.e.}, $\theta \circ ST\theta = \theta \circ \theta ST$. To prove the associativity axiom of $\wcomp{S}{}{T}$, \textit{i.e.}, that $\mu^{\wcomp{S}{}{T}} \circ K \mu^{\wcomp{S}{}{T}} = \mu^{\wcomp{S}{}{T}} \circ \mu^{\wcomp{S}{}{T}} K$, we need equation~\eqref{eq:kappa_mu}, the retract equation, $\theta$ associativity, the retract equation again, and equation~\eqref{eq:kappa_mu} again. 

\begin{equation}
\input{A1_weak_composite_3.tikz}
\end{equation}

\section{Proofs of Section~\ref{sec:iteration}} \label{app:iteration}

In this section we provide the missing details from the proof sketch of Theorem~\ref{theo:comp_equal}. 

\subsection{Expressions of $\kappa^\phi$ and $\kappa^\psi$}

First, let us justify the expressions of $\kappa^\phi$ and $\kappa^\psi$ in equations~\eqref{eq:kappa_phi} and~\eqref{eq:kappa_psi}. Recall that the definition of the idempotent derived from a weak distributive law $\lambda \colon [S,T]$ is given in equation~\eqref{eq:kappa_def} and in the third diagram of~\eqref{eq:kappa_notations}.

Computing $\kappa^\phi$ is easy by using the definition of the idempotent $\kappa^\phi$, the~\eqref{ax:etam} diagram for $\tau$ and the definition of the idempotent $\kappa^\lambda$.
\begin{equation} 
\begin{tikzpicture}
	\begin{pgfonlayer}{nodelayer}
		\node [style=none] (0) at (-5, 0) {$\kappa^\phi$};
		\node [style=none] (1) at (-4, 0) {$=$};
		\node [style=none] (2) at (-2.5, 0) {};
		\node [style=none] (3) at (-2, 0.5) {};
		\node [style=none] (4) at (-1.5, 0) {};
		\node [style=none] (5) at (-2.025, 0.475) {};
		\node [style=none] (6) at (-2.025, 1.5) {};
		\node [style=none] (7) at (-1.975, 0.475) {};
		\node [style=none] (8) at (-1.975, 1.5) {};
		\node [style=none] (9) at (-2.5, 0) {};
		\node [style=none] (10) at (-2.025, -0.5) {};
		\node [style=none] (11) at (-1.5, 0) {};
		\node [style=none] (12) at (-2.025, -1.5) {};
		\node [style=none] (13) at (-2.025, -0.475) {};
		\node [style=none] (14) at (-1.975, -1.5) {};
		\node [style=none] (15) at (-1.975, -0.475) {};
		\node [style=none] (26) at (0, 0) {$=$};
		\node [style=none] (66) at (-1, 1.5) {};
		\node [style=none] (67) at (-1, -1.5) {};
		\node [style=none] (68) at (-1, 0.5) {};
		\node [style=none] (69) at (-3, -0.75) {};
		\node [style=none] (70) at (-3, -1) {};
		\node [style=none] (71) at (-1, 0.5) {};
		\node [style=none] (72) at (-3, -0.75) {};
		\node [style=blue dot] (73) at (-3, -1) {};
		\node [style=blue dot] (74) at (-1, 0.5) {};
		\node [style=none] (75) at (4.5, 0.75) {};
		\node [style=none] (76) at (5, 1.25) {};
		\node [style=none] (77) at (5.5, 0.75) {};
		\node [style=none] (78) at (4.975, 1.225) {};
		\node [style=none] (79) at (4.975, 1.5) {};
		\node [style=none] (80) at (5.025, 1.225) {};
		\node [style=none] (81) at (5.025, 1.5) {};
		\node [style=none] (82) at (5.5, 0.75) {};
		\node [style=none] (83) at (5.975, 0.25) {};
		\node [style=none] (84) at (6.5, 0.75) {};
		\node [style=none] (85) at (5.975, 0) {};
		\node [style=none] (86) at (5.975, 0.275) {};
		\node [style=none] (87) at (6.025, 0) {};
		\node [style=none] (88) at (6.025, 0.275) {};
		\node [style=none] (89) at (4.5, 0) {};
		\node [style=none] (90) at (6.5, 1.5) {};
		\node [style=none] (91) at (5.5, -0.75) {};
		\node [style=none] (92) at (6, -0.25) {};
		\node [style=none] (93) at (6.5, -0.75) {};
		\node [style=none] (94) at (5.975, -0.275) {};
		\node [style=none] (95) at (5.975, 0) {};
		\node [style=none] (96) at (6.025, -0.275) {};
		\node [style=none] (97) at (6.025, 0) {};
		\node [style=none] (98) at (4.5, -0.75) {};
		\node [style=none] (99) at (4.975, -1.25) {};
		\node [style=none] (100) at (5.5, -0.75) {};
		\node [style=none] (101) at (4.975, -1.5) {};
		\node [style=none] (102) at (4.975, -1.225) {};
		\node [style=none] (103) at (5.025, -1.5) {};
		\node [style=none] (104) at (5.025, -1.225) {};
		\node [style=none] (105) at (4.5, 0) {};
		\node [style=none] (106) at (6.5, -1.5) {};
		\node [style=none] (107) at (3.5, 0) {$=$};
		\node [style=none] (108) at (1.55, -0.25) {};
		\node [style=blue dot] (109) at (1.55, -0.25) {};
		\node [style=none] (110) at (2.55, 1.5) {};
		\node [style=blue dot] (111) at (2.55, 0.25) {};
		\node [style=none] (112) at (1.55, -0.25) {};
		\node [style=none] (113) at (2.55, -0.25) {};
		\node [style=none] (114) at (2.55, -1.5) {};
		\node [style=none] (115) at (2, 0.5) {};
		\node [style=none] (116) at (2, -0.5) {};
		\node [style=blue dot] (118) at (2.55, 0.25) {};
		\node [style=none] (119) at (1.55, -0.25) {};
		\node [style=blue dot] (120) at (2.55, 0.25) {};
		\node [style=none] (121) at (1.55, -0.25) {};
		\node [style=none] (122) at (1, 0.5) {};
		\node [style=none] (123) at (1.5, 1) {};
		\node [style=none] (124) at (2, 0.5) {};
		\node [style=none] (125) at (1.475, 0.975) {};
		\node [style=none] (126) at (1.475, 1.5) {};
		\node [style=none] (127) at (1.525, 0.975) {};
		\node [style=none] (128) at (1.525, 1.5) {};
		\node [style=none] (129) at (1, 0.5) {};
		\node [style=none] (130) at (2, 0.5) {};
		\node [style=none] (131) at (1, -0.5) {};
		\node [style=none] (132) at (2, -0.5) {};
		\node [style=none] (133) at (1, -0.5) {};
		\node [style=none] (134) at (1.475, -1) {};
		\node [style=none] (135) at (2, -0.5) {};
		\node [style=none] (136) at (1.475, -1.5) {};
		\node [style=none] (137) at (1.475, -0.975) {};
		\node [style=none] (138) at (1.525, -1.5) {};
		\node [style=none] (139) at (1.525, -0.975) {};
		\node [style=none] (140) at (-2, 2) {$H$};
		\node [style=none] (141) at (1.5, 2) {$H$};
		\node [style=none] (142) at (5, 2) {$H$};
		\node [style=none] (143) at (-2, -2) {$H$};
		\node [style=none] (144) at (1.5, -2) {$H$};
		\node [style=none] (145) at (5, -2) {$H$};
		\node [style=none] (146) at (-1, 2) {$T$};
		\node [style=none] (147) at (2.55, 2) {$T$};
		\node [style=none] (148) at (6.5, 2) {$T$};
		\node [style=none] (149) at (-1, -2) {$T$};
		\node [style=none] (150) at (2.55, -2) {$T$};
		\node [style=none] (151) at (6.5, -2) {$T$};
	\end{pgfonlayer}
	\begin{pgfonlayer}{edgelayer}
		\draw [style=yellow edge, in=180, out=90] (2.center) to (3.center);
		\draw [style=red edge, in=0, out=90] (4.center) to (3.center);
		\draw [style=yellow edge] (5.center) to (6.center);
		\draw [style=red edge] (7.center) to (8.center);
		\draw [style=yellow edge, in=630, out=-180] (10.center) to (9.center);
		\draw [style=red edge, in=-90, out=0] (10.center) to (11.center);
		\draw [style=yellow edge] (12.center) to (13.center);
		\draw [style=red edge] (14.center) to (15.center);
		\draw [style=blue edge] (66.center) to (67.center);
		\draw [style=thick white, in=90, out=-135, looseness=0.75] (68.center) to (69.center);
		\draw [style=blue edge, in=90, out=-135, looseness=0.75] (71.center) to (72.center);
		\draw [style=blue edge] (72.center) to (70.center);
		\draw [style=yellow edge, in=180, out=90] (75.center) to (76.center);
		\draw [style=red edge, in=0, out=90] (77.center) to (76.center);
		\draw [style=yellow edge] (78.center) to (79.center);
		\draw [style=red edge] (80.center) to (81.center);
		\draw [style=red edge, in=630, out=-180] (83.center) to (82.center);
		\draw [style=blue edge, in=-90, out=0] (83.center) to (84.center);
		\draw [style=red edge] (85.center) to (86.center);
		\draw [style=blue edge] (87.center) to (88.center);
		\draw [style=yellow edge] (75.center) to (89.center);
		\draw [style=blue edge] (90.center) to (84.center);
		\draw [style=red edge, in=180, out=90] (91.center) to (92.center);
		\draw [style=blue edge, in=0, out=90] (93.center) to (92.center);
		\draw [style=red edge] (94.center) to (95.center);
		\draw [style=blue edge] (96.center) to (97.center);
		\draw [style=yellow edge, in=630, out=-180] (99.center) to (98.center);
		\draw [style=red edge, in=-90, out=0] (99.center) to (100.center);
		\draw [style=yellow edge] (101.center) to (102.center);
		\draw [style=red edge] (103.center) to (104.center);
		\draw [style=yellow edge] (105.center) to (98.center);
		\draw [style=blue edge] (93.center) to (106.center);
		\draw [style=blue edge] (108.center) to (109);
		\draw [style=blue edge] (110.center) to (111);
		\draw [style=blue edge] (111) to (113.center);
		\draw [style=blue edge] (113.center) to (114.center);
		\draw [style=red edge, in=-90, out=90] (116.center) to (115.center);
		\draw [style=thick white, in=90, out=-150] (118) to (119.center);
		\draw [style=blue edge, in=90, out=-150] (120) to (121.center);
		\draw [style=yellow edge, in=180, out=90] (122.center) to (123.center);
		\draw [style=red edge, in=0, out=90] (124.center) to (123.center);
		\draw [style=yellow edge] (125.center) to (126.center);
		\draw [style=red edge] (127.center) to (128.center);
		\draw [style=yellow edge, in=630, out=-180] (134.center) to (133.center);
		\draw [style=red edge, in=-90, out=0] (134.center) to (135.center);
		\draw [style=yellow edge] (136.center) to (137.center);
		\draw [style=red edge] (138.center) to (139.center);
		\draw [style=yellow edge] (129.center) to (133.center);
	\end{pgfonlayer}
\end{tikzpicture}
\end{equation}

Computing $\kappa^\psi$ is longer because it contains the composite multiplication $\mu^{\wcomp{S}{}{T}}$. We use the definition of $\kappa^\psi$, Lemma~\ref{lem:second_comp}, the retract equation, equation~\eqref{eq:kappa_eta}, the~\eqref{ax:etam} axiom for $\tau$, definitions of $\kappa^\sigma$ and $\kappa^\lambda$, and the retract equation again.
\begin{equation}
\input{A1_kappa_psi.tikz}
\end{equation}

\subsection{Separate and Merge the Triple Functor}

Let us prove the four equations displayed in~\eqref{eq:separate_merge}, which we recall here:
\begin{equation}
    \input{A2_separate_merge.tikz}
\end{equation}
For once we may restate and prove these properties using plain equational reasoning -- they will nonetheless be applied in string-diagrammatic form in the next section. Equationally, the properties are as follows:
\begin{align} 
R \iota^\lambda \circ \iota^\psi = R \kappa^\lambda \circ \iota^\sigma T \circ \iota^\phi \label{eq:s1} \tag{S1} \\
\iota^\sigma T \circ \iota^\phi = \kappa^\sigma T \circ R \iota^\lambda \circ \iota^\psi \label{eq:s2} \tag{S2} \\
\pi^\psi \circ R \pi^\lambda = \pi^\phi \circ \pi^\sigma T \circ R\kappa^\lambda \label{eq:m1} \tag{M1} \\
\pi^\phi \circ \pi^\sigma T = \pi^\psi \circ R \pi^\lambda \circ \kappa^\sigma T \label{eq:m2} \tag{M2}
\end{align}
To prove~\eqref{eq:s1} and~\eqref{eq:m1}, recall that we defined $\iota^\psi = \beta \circ \iota^\phi$ and $\pi^\psi = \pi^\phi \circ \alpha$, where $\alpha = \pi^\sigma T \circ R\iota^\lambda$ and $\beta = R\pi^\lambda \circ \iota^\sigma T$. Hence 
\begin{align} 
R\iota^\lambda \circ \iota^\psi 
& = R\iota^\lambda \circ \beta \circ \iota^\phi \\ 
& = R\iota^\lambda \circ R\pi^\lambda \circ \iota^\sigma T \circ \iota^\phi \\ 
& = R\kappa^\lambda \circ \iota^\sigma T \circ \iota^\phi
\end{align}
and 
\begin{align} 
\pi^\psi \circ R \pi^\lambda 
& = \pi^\phi \circ \alpha \circ R\pi^\lambda \\ 
& = \pi^\phi \circ \pi^\sigma T \circ R\iota^\lambda \circ R\pi^\lambda \\ 
& = \pi^\phi \circ \pi^\sigma T \circ R\kappa^\lambda
\end{align}
To prove~\eqref{eq:s2} and~\eqref{eq:m2}, notice that $\iota^\phi = \kappa^\phi \circ \iota^\phi = \alpha \circ \beta \circ \iota^\phi = \alpha \circ \iota^\psi$, and similarly $\pi^\phi = \pi^\psi \circ \beta$. Hence 
\begin{align} 
\iota^\sigma T \circ \iota^\phi
& = \iota^\sigma T \circ \alpha \circ \iota^\psi \\ 
& = \iota^\sigma T \circ \pi^\sigma T \circ R\iota^\lambda \circ \iota^\psi \\
& = \kappa^\sigma T \circ R\iota^\lambda \circ \iota^\psi
\end{align}
and 
\begin{align} 
\pi^\phi \circ \pi^\sigma T 
& = \pi^\psi \circ \beta \circ \pi^\sigma T \\ 
& = \pi^\psi \circ R\pi^\lambda \circ \iota^\sigma T \circ \pi^\sigma T \\ 
& = \pi^\psi \circ R\pi^\lambda \circ \kappa^\sigma T
\end{align}

\subsection{Units and Multiplications Coincide} 

At this point, we already know that the weak distributive laws $\phi \colon [\wcomp{R}{}{S},T]$ and $\psi \colon [R,\wcomp{S}{}{T}]$ generate respectively a monad $\wcomp{(\wcomp{R}{}{S})}{}{T}$ and a monad $\wcomp{R}{}{(\wcomp{S}{}{T})}$ whose underlying functors coincide. It remains to prove that the units and multiplications coincide. Let us begin with units.

According to the general theory of weak distributive laws from Section~\ref{sec:laws}, the units of the monads generated by $\phi \colon [\wcomp{R}{}{S},T]$ and $\psi \colon [R,\wcomp{S}{}{T}]$ are respectively:
\begin{equation} 
\begin{tikzpicture}
	\begin{pgfonlayer}{nodelayer}
		\node [style=none] (0) at (-0.05, 1.025) {};
		\node [style=none] (1) at (0, 1) {};
		\node [style=none] (2) at (0.05, 1) {};
		\node [style=none] (3) at (0.5, 0.5) {};
		\node [style=none] (4) at (-0.525, 0.5) {};
		\node [style=none] (5) at (-0.475, 0.5) {};
		\node [style=none] (6) at (0.05, 0.975) {};
		\node [style=none] (7) at (0.05, 1.5) {};
		\node [style=none] (8) at (0, 0.975) {};
		\node [style=none] (9) at (0, 1.5) {};
		\node [style=none] (10) at (-0.05, 1) {};
		\node [style=none] (11) at (-0.05, 1.5) {};
		\node [style=none] (12) at (-1, -0.25) {};
		\node [style=none] (13) at (-0.5, 0.25) {};
		\node [style=none] (14) at (0, -0.25) {};
		\node [style=none] (15) at (-0.525, 0.225) {};
		\node [style=none] (16) at (-0.525, 0.5) {};
		\node [style=none] (17) at (-0.475, 0.225) {};
		\node [style=none] (18) at (-0.475, 0.5) {};
		\node [style=blue dot] (19) at (0.5, -0.75) {};
		\node [style=red dot] (20) at (0, -0.75) {};
		\node [style=yellow dot] (21) at (-1, -0.75) {};
		\node [style=none] (22) at (-4, 0) {$\eta^{\wcomp{(\wcomp{R}{}{S})}{}{T}}$};
		\node [style=none] (23) at (-2, 0) {$=$};
		\node [style=none] (24) at (7, 0) {$\eta^{\wcomp{R}{}{(\wcomp{S}{}{T})}}$};
		\node [style=none] (25) at (9, 0) {$=$};
		\node [style=none] (26) at (10.45, 0.975) {};
		\node [style=none] (27) at (10.5, 0.975) {};
		\node [style=none] (28) at (10, 0.5) {};
		\node [style=none] (29) at (10.975, 0.5) {};
		\node [style=none] (30) at (10.55, 1) {};
		\node [style=none] (31) at (11.025, 0.5) {};
		\node [style=none] (32) at (10.45, 0.95) {};
		\node [style=none] (33) at (10.45, 1.5) {};
		\node [style=none] (34) at (10.5, 0.95) {};
		\node [style=none] (35) at (10.5, 1.5) {};
		\node [style=none] (36) at (10.55, 0.975) {};
		\node [style=none] (37) at (10.55, 1.5) {};
		\node [style=none] (38) at (10.5, -0.25) {};
		\node [style=none] (39) at (11, 0.25) {};
		\node [style=none] (40) at (11.5, -0.25) {};
		\node [style=none] (41) at (10.975, 0.225) {};
		\node [style=none] (42) at (10.975, 0.5) {};
		\node [style=none] (43) at (11.025, 0.225) {};
		\node [style=none] (44) at (11.025, 0.5) {};
		\node [style=blue dot] (45) at (11.5, -0.75) {};
		\node [style=red dot] (46) at (10.5, -0.75) {};
		\node [style=yellow dot] (47) at (10, -0.75) {};
		\node [style=none] (48) at (0, 2) {$\Omega$};
		\node [style=none] (49) at (10.5, 2) {$\Omega$};
	\end{pgfonlayer}
	\begin{pgfonlayer}{edgelayer}
		\draw [style=blue edge, in=90, out=0] (2.center) to (3.center);
		\draw [style=red edge, in=180, out=90] (5.center) to (1.center);
		\draw [style=yellow edge, in=180, out=90] (4.center) to (0.center);
		\draw [style=blue edge] (6.center) to (7.center);
		\draw [style=red edge] (8.center) to (9.center);
		\draw [style=yellow edge] (10.center) to (11.center);
		\draw [style=yellow edge, in=180, out=90] (12.center) to (13.center);
		\draw [style=red edge, in=0, out=90] (14.center) to (13.center);
		\draw [style=yellow edge] (15.center) to (16.center);
		\draw [style=red edge] (17.center) to (18.center);
		\draw [style=blue edge] (3.center) to (19);
		\draw [style=red edge] (14.center) to (20);
		\draw [style=yellow edge] (12.center) to (21);
		\draw [style=yellow edge, in=90, out=180] (26.center) to (28.center);
		\draw [style=red edge, in=0, out=90] (29.center) to (27.center);
		\draw [style=blue edge, in=90, out=0] (30.center) to (31.center);
		\draw [style=yellow edge] (32.center) to (33.center);
		\draw [style=red edge] (34.center) to (35.center);
		\draw [style=blue edge] (36.center) to (37.center);
		\draw [style=red edge, in=180, out=90] (38.center) to (39.center);
		\draw [style=blue edge, in=0, out=90] (40.center) to (39.center);
		\draw [style=red edge] (41.center) to (42.center);
		\draw [style=blue edge] (43.center) to (44.center);
		\draw [style=blue edge] (40.center) to (45);
		\draw [style=red edge] (38.center) to (46);
		\draw [style=yellow edge] (28.center) to (47);
	\end{pgfonlayer}
\end{tikzpicture} 
\end{equation}
To prove that these are equal, we just need to use equation~\eqref{eq:kappa_eta} for $\kappa^\lambda$, and equation~\eqref{eq:m1}.
\begin{equation} 
\begin{tikzpicture}
	\begin{pgfonlayer}{nodelayer}
		\node [style=none] (0) at (-4.05, 1.525) {};
		\node [style=none] (1) at (-4, 1.5) {};
		\node [style=none] (2) at (-3.95, 1.5) {};
		\node [style=none] (3) at (-3.5, 1) {};
		\node [style=none] (4) at (-4.525, 1) {};
		\node [style=none] (5) at (-4.475, 1) {};
		\node [style=none] (6) at (-3.95, 1.475) {};
		\node [style=none] (7) at (-3.95, 2) {};
		\node [style=none] (8) at (-4, 1.475) {};
		\node [style=none] (9) at (-4, 2) {};
		\node [style=none] (10) at (-4.05, 1.5) {};
		\node [style=none] (11) at (-4.05, 2) {};
		\node [style=none] (12) at (-5, 0.25) {};
		\node [style=none] (13) at (-4.5, 0.75) {};
		\node [style=none] (14) at (-4, 0.25) {};
		\node [style=none] (15) at (-4.525, 0.725) {};
		\node [style=none] (16) at (-4.525, 1) {};
		\node [style=none] (17) at (-4.475, 0.725) {};
		\node [style=none] (18) at (-4.475, 1) {};
		\node [style=blue dot] (19) at (-3.5, -1.25) {};
		\node [style=red dot] (20) at (-4, -1.25) {};
		\node [style=yellow dot] (21) at (-5, -1.25) {};
		\node [style=none] (22) at (2.95, 1.475) {};
		\node [style=none] (23) at (3, 1.475) {};
		\node [style=none] (24) at (2.5, 1) {};
		\node [style=none] (25) at (3.475, 1) {};
		\node [style=none] (26) at (3.05, 1.5) {};
		\node [style=none] (27) at (3.525, 1) {};
		\node [style=none] (28) at (2.95, 1.45) {};
		\node [style=none] (29) at (2.95, 2) {};
		\node [style=none] (30) at (3, 1.45) {};
		\node [style=none] (31) at (3, 2) {};
		\node [style=none] (32) at (3.05, 1.475) {};
		\node [style=none] (33) at (3.05, 2) {};
		\node [style=none] (34) at (3, -1) {};
		\node [style=none] (35) at (3.5, -0.5) {};
		\node [style=none] (36) at (4, -1) {};
		\node [style=none] (37) at (3.475, -0.525) {};
		\node [style=none] (38) at (3.475, 1) {};
		\node [style=none] (39) at (3.525, -0.525) {};
		\node [style=none] (40) at (3.525, 1) {};
		\node [style=blue dot] (41) at (4, -1.25) {};
		\node [style=red dot] (42) at (3, -1.25) {};
		\node [style=yellow dot] (43) at (2.5, -1.25) {};
		\node [style=none] (44) at (-2.5, 0) {$=$};
		\node [style=none] (45) at (-0.5, -1) {};
		\node [style=none] (46) at (0, -0.5) {};
		\node [style=none] (47) at (0.5, -1) {};
		\node [style=none] (48) at (-0.025, -0.525) {};
		\node [style=none] (49) at (-0.025, -0.25) {};
		\node [style=none] (50) at (0.025, -0.525) {};
		\node [style=none] (51) at (0.025, -0.25) {};
		\node [style=none] (52) at (-0.5, 0.25) {};
		\node [style=none] (53) at (-0.025, -0.25) {};
		\node [style=none] (54) at (0.5, 0.25) {};
		\node [style=none] (55) at (-0.025, -0.25) {};
		\node [style=none] (56) at (-0.025, -0.225) {};
		\node [style=none] (57) at (0.025, -0.25) {};
		\node [style=none] (58) at (0.025, -0.225) {};
		\node [style=none] (59) at (-0.5, -1) {};
		\node [style=red dot] (60) at (-0.5, -1.25) {};
		\node [style=none] (61) at (0.5, -1) {};
		\node [style=blue dot] (62) at (0.5, -1.25) {};
		\node [style=none] (63) at (-0.3, 1.525) {};
		\node [style=none] (64) at (-0.25, 1.5) {};
		\node [style=none] (65) at (-0.2, 1.5) {};
		\node [style=none] (66) at (0.5, 1) {};
		\node [style=none] (67) at (-1.025, 1) {};
		\node [style=none] (68) at (-0.975, 1) {};
		\node [style=none] (69) at (-0.2, 1.475) {};
		\node [style=none] (70) at (-0.2, 2) {};
		\node [style=none] (71) at (-0.25, 1.475) {};
		\node [style=none] (72) at (-0.25, 2) {};
		\node [style=none] (73) at (-0.3, 1.5) {};
		\node [style=none] (74) at (-0.3, 2) {};
		\node [style=none] (75) at (-1.5, 0.25) {};
		\node [style=none] (76) at (-1, 0.75) {};
		\node [style=none] (77) at (-0.5, 0.25) {};
		\node [style=none] (78) at (-1.025, 0.725) {};
		\node [style=none] (79) at (-1.025, 1) {};
		\node [style=none] (80) at (-0.975, 0.725) {};
		\node [style=none] (81) at (-0.975, 1) {};
		\node [style=none] (82) at (0.5, 0.25) {};
		\node [style=none] (83) at (-0.5, 0.25) {};
		\node [style=yellow dot] (84) at (-1.5, -1.25) {};
		\node [style=none] (85) at (1.5, 0) {$=$};
		\node [style=none] (86) at (-4, 2.5) {$\Omega$};
		\node [style=none] (87) at (-0.25, 2.5) {$\Omega$};
		\node [style=none] (88) at (3, 2.5) {$\Omega$};
	\end{pgfonlayer}
	\begin{pgfonlayer}{edgelayer}
		\draw [style=blue edge, in=90, out=0] (2.center) to (3.center);
		\draw [style=red edge, in=180, out=90] (5.center) to (1.center);
		\draw [style=yellow edge, in=180, out=90] (4.center) to (0.center);
		\draw [style=blue edge] (6.center) to (7.center);
		\draw [style=red edge] (8.center) to (9.center);
		\draw [style=yellow edge] (10.center) to (11.center);
		\draw [style=yellow edge, in=180, out=90] (12.center) to (13.center);
		\draw [style=red edge, in=0, out=90] (14.center) to (13.center);
		\draw [style=yellow edge] (15.center) to (16.center);
		\draw [style=red edge] (17.center) to (18.center);
		\draw [style=blue edge] (3.center) to (19);
		\draw [style=red edge] (14.center) to (20);
		\draw [style=yellow edge] (12.center) to (21);
		\draw [style=yellow edge, in=90, out=180] (22.center) to (24.center);
		\draw [style=red edge, in=0, out=90] (25.center) to (23.center);
		\draw [style=blue edge, in=90, out=0] (26.center) to (27.center);
		\draw [style=yellow edge] (28.center) to (29.center);
		\draw [style=red edge] (30.center) to (31.center);
		\draw [style=blue edge] (32.center) to (33.center);
		\draw [style=red edge, in=180, out=90] (34.center) to (35.center);
		\draw [style=blue edge, in=0, out=90] (36.center) to (35.center);
		\draw [style=red edge] (37.center) to (38.center);
		\draw [style=blue edge] (39.center) to (40.center);
		\draw [style=blue edge] (36.center) to (41);
		\draw [style=red edge] (34.center) to (42);
		\draw [style=yellow edge] (24.center) to (43);
		\draw [style=red edge, in=180, out=90] (45.center) to (46.center);
		\draw [style=blue edge, in=0, out=90] (47.center) to (46.center);
		\draw [style=red edge] (48.center) to (49.center);
		\draw [style=blue edge] (50.center) to (51.center);
		\draw [style=red edge, in=630, out=-180] (53.center) to (52.center);
		\draw [style=blue edge, in=-90, out=0] (53.center) to (54.center);
		\draw [style=red edge] (55.center) to (56.center);
		\draw [style=blue edge] (57.center) to (58.center);
		\draw [style=red edge] (59.center) to (60);
		\draw [style=blue edge] (61.center) to (62);
		\draw [style=blue edge, in=90, out=0] (65.center) to (66.center);
		\draw [style=red edge, in=180, out=90] (68.center) to (64.center);
		\draw [style=yellow edge, in=180, out=90] (67.center) to (63.center);
		\draw [style=blue edge] (69.center) to (70.center);
		\draw [style=red edge] (71.center) to (72.center);
		\draw [style=yellow edge] (73.center) to (74.center);
		\draw [style=yellow edge, in=180, out=90] (75.center) to (76.center);
		\draw [style=red edge, in=0, out=90] (77.center) to (76.center);
		\draw [style=yellow edge] (78.center) to (79.center);
		\draw [style=red edge] (80.center) to (81.center);
		\draw [style=blue edge] (66.center) to (82.center);
		\draw [style=red edge] (77.center) to (83.center);
		\draw [style=yellow edge] (75.center) to (84);
	\end{pgfonlayer}
\end{tikzpicture}
\end{equation}

The situation is more intricated for multiplications. According to the general theory of weak distributive laws from Section~\ref{sec:laws}, the multiplications of the monads generated by $\phi \colon [\wcomp{R}{}{S},T]$ and $\psi \colon [R,\wcomp{S}{}{T}]$ are respectively:
\begin{equation}
\input{A2_mult_both.tikz}
\end{equation}

To prove that these are equal, we will manipulate separately both of them until their string-diagrammatic representations coincide. The first multiplication can be transformed as follows, using equation~\eqref{eq:m2}, equation~\eqref{eq:kappa_mu}, the retract equation, Lemma~\eqref{lem:first_comp} and the retract equation again. 
\begin{equation} 
\input{A2_mult_phi.tikz}
\end{equation}

The second multiplication, in turn, transforms as follows, using equation~\eqref{eq:s1}, Lemma~\eqref{lem:second_comp}, the retract equation, equation~\eqref{eq:kappa_mu} and the retract equation again.
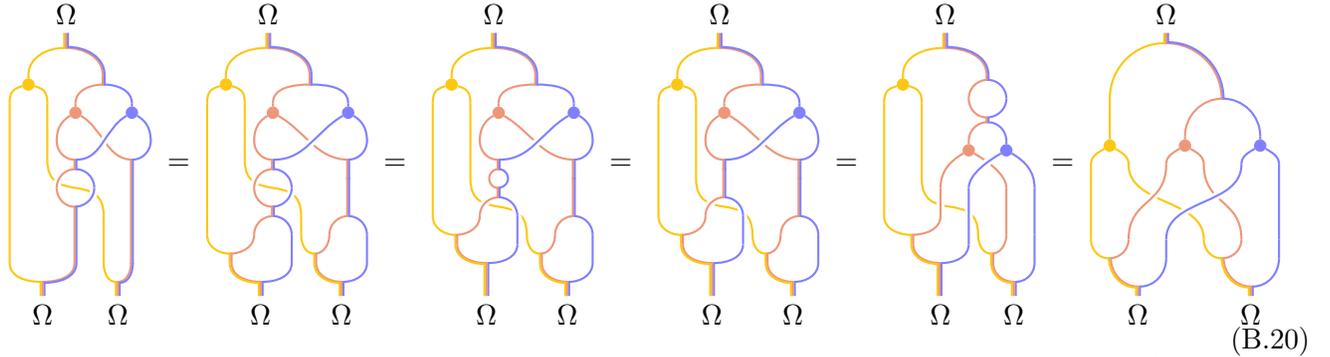
\begin{equation}
\input{A2_mult_psi.tikz}
\end{equation}
Hence both multiplications are equal. This achieves the proof that $\wcomp{(\wcomp{R}{}{S})}{}{T}$ and $\wcomp{R}{}{(\wcomp{S}{}{T})}$ coincide as monads.

\section{Proofs of Section~\ref{sec:conclusion}} \label{app:conclusion}

In this section we restate one implication of~\cite[Theorem~5]{Winter16} and prove it using string diagrams.

\begin{definition}
Let $F$ be a functor, $S$ and $T$ be monads, and $\comp{S}{}{T}$ be their composite monad with respect to some distributive law of type $[S,T]$.
A Winter-law is a natural transformation $\lambda \colon TF \nto FST$ such that 
 \begin{align}
    & \lambda \circ \eta^T F = F\eta^{\comp{S}{}{T}} \label{ax:etad} \tag{$\eta^*$} \\
    & \lambda \circ \mu^T F = F\mu^{\comp{S}{}{T}} \circ \lambda ST \circ T\lambda \label{ax:mud} \tag{$\mu^*$}
 \end{align}
\end{definition}

\begin{theorem}[Winter's theorem converse implication]
Let $F$ be a functor, $S$ and $T$ be monads, $\lambda^0 \colon [S,T]$ be a distributive law, $\lambda^1 \colon TF \nto FST$ be a Winter-law and $\lambda^2 \colon [F,S]$ be an EM-law. Assume the following coherence axiom is satisfied:
\begin{equation}
    F\mu^S T \circ FS\lambda^0 \circ \lambda^1 S \circ T\lambda^2 = F\mu^S T \circ \lambda^2 ST \circ S\lambda^1 \circ \lambda^0 F \label{eq:coh} \tag{coh}
\end{equation}
Let
\begin{equation}
\hat{\lambda} \triangleq
    \begin{tikzcd}
    STF \arrow[r, "S\lambda^1", Rightarrow] & SFST \arrow[r, "\lambda^2 ST", Rightarrow] & FSST \arrow[r, "F\mu^S T", Rightarrow] & FST
    \end{tikzcd}
\end{equation}
Then $\hat{\lambda} \colon [F,\comp{S}{}{T}]$ is an EM-law.
\end{theorem}

\begin{proof}
First, we express all available data using string diagrams.
\begin{equation}
    \begin{tikzpicture}
	\begin{pgfonlayer}{nodelayer}
		\node [style=none] (0) at (-15, 1) {};
		\node [style=none] (1) at (-15, -1) {};
		\node [style=none] (2) at (3, -1) {};
		\node [style=none] (3) at (3, 1) {};
		\node [style=blue dot] (4) at (7, 0) {};
		\node [style=none] (5) at (7, 1) {};
		\node [style=none] (6) at (12, 1) {};
		\node [style=blue dot] (7) at (12, 0) {};
		\node [style=none] (8) at (11, -1) {};
		\node [style=none] (9) at (13, -1) {};
		\node [style=none] (10) at (-11, -1) {};
		\node [style=none] (11) at (-11, 1) {};
		\node [style=red dot] (12) at (-7, 0) {};
		\node [style=none] (13) at (-7, 1) {};
		\node [style=none] (14) at (-2, 1) {};
		\node [style=red dot] (15) at (-2, 0) {};
		\node [style=none] (16) at (-3, -1) {};
		\node [style=none] (17) at (-1, -1) {};
		\node [style=none] (18) at (-16.5, 0) {$1_F$};
		\node [style=none] (19) at (1.5, 0) {$1_T$};
		\node [style=none] (20) at (5.5, 0) {$\eta^T$};
		\node [style=none] (21) at (9.5, 0) {$\mu^T$};
		\node [style=none] (22) at (-8.5, 0) {$\eta^S$};
		\node [style=none] (23) at (-4.5, 0) {$\mu^S$};
		\node [style=none] (24) at (-12.5, 0) {$1_S$};
		\node [style=none] (58) at (-15.75, 0) {$\triangleq$};
		\node [style=none] (59) at (-11.75, 0) {$\triangleq$};
		\node [style=none] (60) at (-3.75, 0) {$\triangleq$};
		\node [style=none] (61) at (2.25, 0) {$\triangleq$};
		\node [style=none] (62) at (6.25, 0) {$\triangleq$};
		\node [style=none] (63) at (10.25, 0) {$\triangleq$};
		\node [style=none] (67) at (-7.75, 0) {$\triangleq$};
		\node [style=none] (68) at (-15, 1.5) {$F$};
		\node [style=none] (69) at (-15, -1.5) {$F$};
		\node [style=none] (74) at (-11, 1.5) {$S$};
		\node [style=none] (75) at (-11, -1.5) {$S$};
		\node [style=none] (76) at (-7, 1.5) {$S$};
		\node [style=none] (77) at (-2, 1.5) {$S$};
		\node [style=none] (78) at (-3, -1.5) {$S$};
		\node [style=none] (79) at (-1, -1.5) {$S$};
		\node [style=none] (84) at (3, 1.5) {$T$};
		\node [style=none] (85) at (3, -1.5) {$T$};
		\node [style=none] (86) at (7, 1.5) {$T$};
		\node [style=none] (87) at (12, 1.5) {$T$};
		\node [style=none] (88) at (11, -1.5) {$T$};
		\node [style=none] (89) at (13, -1.5) {$T$};
	\end{pgfonlayer}
	\begin{pgfonlayer}{edgelayer}
		\draw [style=black edge] (0.center) to (1.center);
		\draw [style=blue edge] (3.center) to (2.center);
		\draw [style=blue edge] (5.center) to (4);
		\draw [style=blue edge] (6.center) to (7);
		\draw [style=blue edge, in=90, out=-150] (7) to (8.center);
		\draw [style=blue edge, in=90, out=-30] (7) to (9.center);
		\draw [style=red edge] (11.center) to (10.center);
		\draw [style=red edge] (13.center) to (12);
		\draw [style=red edge] (14.center) to (15);
		\draw [style=red edge, in=90, out=-150] (15) to (16.center);
		\draw [style=red edge, in=90, out=-30] (15) to (17.center);
	\end{pgfonlayer}
\end{tikzpicture}
\end{equation}
\begin{equation}
    \begin{tikzpicture}
	\begin{pgfonlayer}{nodelayer}
		\node [style=none] (0) at (-12.5, 1) {};
		\node [style=none] (1) at (-10.5, 1) {};
		\node [style=none] (2) at (-12.5, -1) {};
		\node [style=none] (3) at (-10.5, -1) {};
		\node [style=none] (4) at (-10.5, 1) {};
		\node [style=none] (5) at (-12.5, -1) {};
		\node [style=none] (6) at (-14, 0) {$\lambda^0$};
		\node [style=none] (7) at (3.5, 1) {};
		\node [style=none] (8) at (5.5, 1) {};
		\node [style=none] (9) at (3.5, -1) {};
		\node [style=none] (10) at (5.5, -1) {};
		\node [style=none] (11) at (5.5, 1) {};
		\node [style=none] (12) at (3.5, -1) {};
		\node [style=none] (13) at (2, 0) {$\lambda^2$};
		\node [style=none] (14) at (-6, 0) {$\lambda^1$};
		\node [style=none] (15) at (-3.5, 0) {};
		\node [style=none] (16) at (-3.5, 0) {};
		\node [style=none] (17) at (-3.5, 1) {};
		\node [style=none] (18) at (-4.5, 1) {};
		\node [style=none] (19) at (-2.5, 1) {};
		\node [style=none] (20) at (-4.5, -1) {};
		\node [style=none] (21) at (-2.5, -1) {};
		\node [style=none] (22) at (-2.5, 1) {};
		\node [style=none] (23) at (-4.5, -1) {};
		\node [style=none] (24) at (-13.25, 0) {$\triangleq$};
		\node [style=none] (25) at (-5.25, 0) {$\triangleq$};
		\node [style=none] (26) at (2.75, 0) {$\triangleq$};
		\node [style=none] (27) at (-4.5, 1.5) {$F$};
		\node [style=none] (28) at (-2.5, -1.5) {$F$};
		\node [style=none] (29) at (3.5, 1.5) {$F$};
		\node [style=none] (30) at (5.5, -1.5) {$F$};
		\node [style=none] (31) at (-12.5, 1.5) {$S$};
		\node [style=none] (32) at (-10.5, -1.5) {$S$};
		\node [style=none] (33) at (5.5, 1.5) {$S$};
		\node [style=none] (34) at (3.5, -1.5) {$S$};
		\node [style=none] (35) at (-10.5, 1.5) {$T$};
		\node [style=none] (36) at (-12.5, -1.5) {$T$};
		\node [style=none] (37) at (-2.5, 1.5) {$T$};
		\node [style=none] (38) at (-4.5, -1.5) {$T$};
		\node [style=none] (39) at (-3.5, 1.5) {$S$};
	\end{pgfonlayer}
	\begin{pgfonlayer}{edgelayer}
		\draw [style=red edge, in=90, out=-90] (0.center) to (3.center);
		\draw [style=thick white, in=-90, out=90] (2.center) to (1.center);
		\draw [style=blue edge, in=-90, out=90] (5.center) to (4.center);
		\draw [style=black edge, in=90, out=-90] (7.center) to (10.center);
		\draw [style=thick white, in=-90, out=90] (9.center) to (8.center);
		\draw [style=red edge, in=-90, out=90] (12.center) to (11.center);
		\draw [style=red edge] (17.center) to (16.center);
		\draw [style=black edge, in=90, out=-90] (18.center) to (21.center);
		\draw [style=thick white, in=-90, out=90] (20.center) to (19.center);
		\draw [style=blue edge, in=-90, out=90] (23.center) to (22.center);
	\end{pgfonlayer}
\end{tikzpicture}
\end{equation}
Axioms specific to Winter's framework are displayed below:
\begin{equation}
    \input{W_axioms.tikz}
\end{equation}
By definition, the natural transformation $\hat{\lambda}$ is as follows:
\begin{equation}
    \begin{tikzpicture}
	\begin{pgfonlayer}{nodelayer}
		\node [style=none] (0) at (-0.5, 0) {};
		\node [style=none] (1) at (0.5, 0) {};
		\node [style=none] (2) at (0, 1) {};
		\node [style=red dot] (3) at (0, 0.5) {};
		\node [style=none] (4) at (-0.5, 0) {};
		\node [style=none] (5) at (0.5, 0) {};
		\node [style=none] (6) at (0.5, 0) {};
		\node [style=none] (7) at (0.5, 0) {};
		\node [style=none] (8) at (0.5, -0.5) {};
		\node [style=none] (9) at (-0.5, 0) {};
		\node [style=none] (10) at (-0.5, -1) {};
		\node [style=none] (11) at (-1, 0.25) {};
		\node [style=none] (12) at (1, -1) {};
		\node [style=none] (13) at (0, -1) {};
		\node [style=none] (14) at (1, 0) {};
		\node [style=none] (16) at (1, 1) {};
		\node [style=none] (17) at (-2.5, 0) {$\hat{\lambda}$};
		\node [style=none] (18) at (-1, 1) {};
		\node [style=none] (19) at (-1.75, 0) {$\triangleq$};
		\node [style=none] (20) at (-1, 1.5) {$F$};
		\node [style=none] (21) at (1, -1.5) {$F$};
		\node [style=none] (22) at (0, 1.5) {$S$};
		\node [style=none] (23) at (-0.5, -1.5) {$S$};
		\node [style=none] (24) at (1, 1.5) {$T$};
		\node [style=none] (25) at (0, -1.5) {$T$};
	\end{pgfonlayer}
	\begin{pgfonlayer}{edgelayer}
		\draw [style=red edge] (2.center) to (3);
		\draw [style=red edge, in=90, out=-150] (3) to (4.center);
		\draw [style=red edge, in=90, out=-30] (3) to (5.center);
		\draw [style=red edge] (7.center) to (8.center);
		\draw [style=black edge, in=90, out=-90, looseness=0.75] (11.center) to (12.center);
		\draw [style=thick white] (9.center) to (10.center);
		\draw [style=red edge] (9.center) to (10.center);
		\draw [style=thick white, in=-90, out=90, looseness=0.75] (13.center) to (14.center);
		\draw [style=blue edge, in=-90, out=90, looseness=0.75] (13.center) to (14.center);
		\draw [style=blue edge] (16.center) to (14.center);
		\draw [style=black edge] (18.center) to (11.center);
	\end{pgfonlayer}
\end{tikzpicture}
\end{equation}
The~\eqref{ax:etam} axiom for $\hat{\lambda}$ can be proved using $\lambda^2$~\eqref{ax:etam} axiom, $\eta^S$ unitality and $\lambda^1$~\eqref{ax:etad} axiom.
\begin{equation}
    \begin{tikzpicture}
	\begin{pgfonlayer}{nodelayer}
		\node [style=none] (0) at (-3.5, 0.5) {};
		\node [style=none] (1) at (-2.5, 0.5) {};
		\node [style=none] (2) at (-3, 1.5) {};
		\node [style=red dot] (3) at (-3, 1) {};
		\node [style=none] (4) at (-3.5, 0.5) {};
		\node [style=none] (5) at (-2.5, 0.5) {};
		\node [style=none] (6) at (-2.5, 0.5) {};
		\node [style=none] (7) at (-2.5, 0.5) {};
		\node [style=none] (8) at (-2.5, 0) {};
		\node [style=none] (9) at (-3.5, 0.5) {};
		\node [style=none] (10) at (-3.5, -0.5) {};
		\node [style=none] (11) at (-4, 0.75) {};
		\node [style=none] (12) at (-2, -0.5) {};
		\node [style=none] (13) at (-3, -0.5) {};
		\node [style=none] (14) at (-2, 0.5) {};
		\node [style=none] (15) at (-2, 1.5) {};
		\node [style=none] (16) at (-4, 1.5) {};
		\node [style=red dot] (17) at (-3.5, -1) {};
		\node [style=blue dot] (18) at (-3, -1) {};
		\node [style=none] (19) at (-2, -1.5) {};
		\node [style=none] (20) at (0.5, 0.5) {};
		\node [style=none] (21) at (1.5, 0.5) {};
		\node [style=none] (22) at (1, 1.5) {};
		\node [style=red dot] (23) at (1, 1) {};
		\node [style=none] (24) at (0.5, 0.5) {};
		\node [style=none] (25) at (1.5, 0.5) {};
		\node [style=none] (26) at (1.5, 0.5) {};
		\node [style=none] (27) at (1.5, 0.5) {};
		\node [style=none] (28) at (1.5, -0.5) {};
		\node [style=none] (29) at (0.5, 0.5) {};
		\node [style=none] (30) at (0.5, 0.5) {};
		\node [style=none] (31) at (0, 0.25) {};
		\node [style=none] (32) at (2, -1) {};
		\node [style=none] (33) at (1, -1) {};
		\node [style=none] (34) at (2, 0) {};
		\node [style=none] (35) at (2, 1.5) {};
		\node [style=none] (36) at (0, 1.5) {};
		\node [style=red dot] (37) at (0.5, 0.25) {};
		\node [style=blue dot] (38) at (1, -1.25) {};
		\node [style=none] (39) at (2, -1.5) {};
		\node [style=none] (40) at (-1, 0) {$=$};
		\node [style=none] (41) at (3, 0) {$=$};
		\node [style=none] (43) at (5.5, 1.5) {};
		\node [style=none] (47) at (5.5, 1.5) {};
		\node [style=none] (48) at (5.5, 1.5) {};
		\node [style=none] (49) at (5.5, 1.5) {};
		\node [style=none] (50) at (5.5, -0.5) {};
		\node [style=none] (53) at (4, 0.25) {};
		\node [style=none] (54) at (6, -1) {};
		\node [style=none] (55) at (5, -1) {};
		\node [style=none] (56) at (6, 0) {};
		\node [style=none] (57) at (6, 1.5) {};
		\node [style=none] (58) at (4, 1.5) {};
		\node [style=blue dot] (60) at (5, -1.25) {};
		\node [style=none] (61) at (6, -1.5) {};
		\node [style=none] (62) at (7, 0) {$=$};
		\node [style=none] (63) at (8, 1.5) {};
		\node [style=none] (64) at (8, -1.5) {};
		\node [style=red dot] (65) at (8.5, 0) {};
		\node [style=none] (66) at (8.5, 1) {};
		\node [style=none] (67) at (9, 1) {};
		\node [style=blue dot] (68) at (9, 0) {};
		\node [style=none] (69) at (9, 1.5) {};
		\node [style=none] (70) at (8.5, 1.5) {};
		\node [style=none] (71) at (-4, 2) {$F$};
		\node [style=none] (72) at (0, 2) {$F$};
		\node [style=none] (73) at (4, 2) {$F$};
		\node [style=none] (74) at (8, 2) {$F$};
		\node [style=none] (75) at (-2, -2) {$F$};
		\node [style=none] (76) at (2, -2) {$F$};
		\node [style=none] (77) at (6, -2) {$F$};
		\node [style=none] (78) at (8, -2) {$F$};
		\node [style=none] (79) at (-3, 2) {$S$};
		\node [style=none] (80) at (1, 2) {$S$};
		\node [style=none] (81) at (5.5, 2) {$S$};
		\node [style=none] (82) at (8.5, 2) {$S$};
		\node [style=none] (83) at (-2, 2) {$T$};
		\node [style=none] (84) at (2, 2) {$T$};
		\node [style=none] (85) at (6, 2) {$T$};
		\node [style=none] (86) at (9, 2) {$T$};
	\end{pgfonlayer}
	\begin{pgfonlayer}{edgelayer}
		\draw [style=red edge] (2.center) to (3);
		\draw [style=red edge, in=90, out=-150] (3) to (4.center);
		\draw [style=red edge, in=90, out=-30] (3) to (5.center);
		\draw [style=red edge] (7.center) to (8.center);
		\draw [style=black edge, in=90, out=-90, looseness=0.75] (11.center) to (12.center);
		\draw [style=thick white] (9.center) to (10.center);
		\draw [style=red edge] (9.center) to (10.center);
		\draw [style=thick white, in=-90, out=90, looseness=0.75] (13.center) to (14.center);
		\draw [style=blue edge, in=-90, out=90, looseness=0.75] (13.center) to (14.center);
		\draw [style=blue edge] (15.center) to (14.center);
		\draw [style=black edge] (16.center) to (11.center);
		\draw [style=red edge] (10.center) to (17);
		\draw [style=blue edge] (13.center) to (18);
		\draw [style=black edge] (12.center) to (19.center);
		\draw [style=red edge] (22.center) to (23);
		\draw [style=red edge, in=90, out=-150] (23) to (24.center);
		\draw [style=red edge, in=90, out=-30] (23) to (25.center);
		\draw [style=red edge] (27.center) to (28.center);
		\draw [style=black edge, in=90, out=-90, looseness=0.75] (31.center) to (32.center);
		\draw [style=thick white] (29.center) to (30.center);
		\draw [style=red edge] (29.center) to (30.center);
		\draw [style=thick white, in=-90, out=90, looseness=0.75] (33.center) to (34.center);
		\draw [style=blue edge, in=-90, out=90, looseness=0.75] (33.center) to (34.center);
		\draw [style=blue edge] (35.center) to (34.center);
		\draw [style=black edge] (36.center) to (31.center);
		\draw [style=red edge] (30.center) to (37);
		\draw [style=blue edge] (33.center) to (38);
		\draw [style=black edge] (32.center) to (39.center);
		\draw [style=red edge] (49.center) to (50.center);
		\draw [style=black edge, in=90, out=-90, looseness=0.75] (53.center) to (54.center);
		\draw [style=thick white, in=-90, out=90, looseness=0.75] (55.center) to (56.center);
		\draw [style=blue edge, in=-90, out=90, looseness=0.75] (55.center) to (56.center);
		\draw [style=blue edge] (57.center) to (56.center);
		\draw [style=black edge] (58.center) to (53.center);
		\draw [style=blue edge] (55.center) to (60);
		\draw [style=black edge] (54.center) to (61.center);
		\draw [style=black edge] (63.center) to (64.center);
		\draw [style=red edge] (66.center) to (65);
		\draw [style=blue edge] (67.center) to (68);
		\draw [style=blue edge] (69.center) to (67.center);
		\draw [style=red edge] (70.center) to (66.center);
	\end{pgfonlayer}
\end{tikzpicture}
\end{equation}
The~\eqref{ax:mum} axiom for $\hat{\lambda}$ can be proved using $\lambda^1$~\eqref{ax:mud} axiom, $\lambda^2$~\eqref{ax:mum} axiom, $\mu^S$ associativity twice, equation~\eqref{eq:coh}, $\mu^S$ associativity twice again, and $\lambda^0$~\eqref{ax:mup} axiom.
\begin{equation}
    \input{W_computation.tikz}
\end{equation}
\end{proof}

\end{document}